\documentclass{article}
\usepackage[utf8]{inputenc}
\usepackage{amsmath}
\usepackage{amssymb}
\usepackage{braket}
\usepackage{graphicx}
\usepackage{float}
\usepackage{braket}
\usepackage{bbold}
\usepackage[margin=1in]{geometry}
\usepackage{amsthm}
\usepackage[dvipsnames]{xcolor}
\usepackage{comment}
\usepackage{mathtools}
\usepackage{hyperref}
\usepackage{thm-restate}
\usepackage{thmtools}
\usepackage[capitalise]{cleveref}
\usepackage{tikz}
\usepackage{standalone}
\usepackage{textcomp}
\usepackage{cancel}
\usepackage{stmaryrd}
\usepackage{subcaption}
\usepackage[sorting=none,style=numeric-comp,date=year]{biblatex}
\usepackage{float}
\usepackage{subcaption}
\usepackage{graphicx}
\usepackage{soul}
\usepackage{authblk}
\usetikzlibrary{arrows.meta,calc}

\bibliography{references}

\newcommand{\C}[1]{\mathcal{#1}}
\newcommand{\PBP}[0]{\mathfrak{P} = (\C{C}, (\{\C{M}^{A_k}_{x_k}\}_{x_k})_{k=0}^{N+1})}
\newcommand{\PBPp}[0]{\mathfrak{P}' = (\C{C}', (\{\C{M}^{\prime A_k}_{x_k}\}_{x_k})_{k=0}^{N+1})}
\newcommand{\PBPchoose}[3]{#1 = (#2, (\{#3\}_{x_k})_{k=0}^{N+1})}
\newcommand{\QP}[0]{\PBPchoose{\mathfrak{Q}}{\C{Q}}{\M}}
\newcommand{\QPp}[0]{\PBPchoose{\mathfrak{Q}}{\C{Q}}{\Mp}}
\newcommand{\CM}[2]{(#1 \circ \bigotimes_{k=0}^{N+1} #2)^{\hookrightarrow}}
\newcommand{\M}[0]{\C{M}^{A_k}_{x_k}}
\newcommand{\Mp}[0]{\C{M}^{\prime A_k}_{x_k}}
\newcommand{\lc}[1]{(#1)^{\hookrightarrow}}

\newcommand{\F}[2]{\C{F}^{#1, #2}}

\newcommand{\FI}[0]{\C{F}^{A^I_k, \C{T}^I_k}}

\newcommand{\FIall}[0]{\C{F}^{A^I_k, \C{T}}}

\newcommand{\FO}[0]{\C{F}^{A^O_k, \C{T}^O_k}}

\newcommand{\FOall}[0]{\C{F}^{A^O_k, \C{T}}}

\newcommand{\FIO}[0]{\C{F}^{A^{I/O}_k, \C{T}^{I/O}_k}}

\newcommand{\FIOall}[0]{\C{F}^{A^{I/O}_k, \C{T}}}

\newtheorem{defi}{Definition}

\Crefname{prop}{Proposition}{Propositions}
\crefname{prop}{Prop.}{Props.}
\Crefname{defi}{Definition}{Definitions}
\crefname{defi}{Def.}{Defs.}
\Crefname{prop}{Proposition}{Propositions}
\crefname{lemma}{Lemma}{Lemmas}
\crefname{coro}{Cor.}{Cors.}
\Crefname{coro}{Corollary}{Corollaries}
\crefname{example}{Ex.}{Exs.}
\Crefname{example}{Example}{Examples}
\crefname{remark}{Remark}{Remarks}
\crefname{conj}{Conj.}{Conjs.}
\Crefname{conj}{Conjecture}{Conjectures}
\crefname{thm}{Thm.}{Thms.}
\Crefname{thm}{Theorem}{Theorems}

\hypersetup{
    colorlinks=true,
    linkcolor=blue,
    filecolor=magenta,      
    urlcolor=cyan,
    citecolor=green,
    }

\DeclareFieldFormat{title}{\mkbibquote{#1}}

\title{Higher-order quantum processes respecting closed labs in a spacetime have quantum controlled causal order}

\author[1,2]{Matthias Salzger}
\author[3,2]{V.~Vilasini}

\affil[1]{International Centre for Theory of Quantum Technologies, University of Gda\'{n}sk, 80--309 Gda\'{n}sk, Poland}
\affil[2]{Institute for Theoretical Physics, ETH Zurich, 8093 Z\"{u}rich, Switzerland}
\affil[3]{Universit\'e Grenoble Alpes, Inria, 38000 Grenoble, France}

\date{\today}

\begin{document}

\maketitle

\begin{abstract}
In quantum causality and quantum information, there is a vast landscape of abstract quantum protocols that permit cyclic or non-acyclic causal structures between quantum operations. This includes widely studied frameworks for indefinite causal order and higher-order quantum processes, such as process matrices. However, a longstanding open question has been which is the largest class of such abstract processes that admit physical realisations without post-selection. In this work, we provide a rigorous answer by adopting a top-down approach grounded in relativistic causality principles, motivated by the fact that physical experiments are implemented consistently with such principles in spacetimes with acyclic lightcone structures. Building on the framework of causal boxes, which characterise the most general quantum information-processing protocols compatible with fixed background spacetimes, we formalise additional physically motivated constraints (Acting Once + Local Order) capturing the closed-laboratory assumptions of the process matrix framework at a fine-grained spacetime level. We prove that any protocol realisable in a classical acyclic spacetime and satisfying these spatiotemporal closed-lab conditions is behaviourally equivalent to a quantum circuit with quantum control of causal orders (QC-QC), providing a top-down derivation of QC-QCs from physical principles. Our results therefore show that QC-QCs constitute precisely the class of higher-order quantum processes, including those with indefinite orders, that can be physically realised within classical spacetime, clearly ruling out the possibility of any experiment in this regime that could realise more general non-causal processes under such a closed-labs assumption. This clarifies the relationship between abstract higher-order process matrix frameworks and experimentally accessible quantum protocols, as well as the interplay between coarse-grained cyclic and fine-grained acyclic operational causal structures. We also develop characterisation techniques and results for process box protocols that lead to new causality-based open questions concerning spacetime quantum protocols and relativistic quantum experiments.

\end{abstract}

\pagenumbering{arabic}

\tableofcontents

\section{Introduction}
\label{sec:intro}

Many frameworks used in quantum information theory are formulated without explicit reference to any underlying spacetime. This abstraction has led to a broad landscape of general quantum protocols, where the (information-theoretic) causal structure between quantum operations need not be definite and acyclic, allowing for intriguing situations involving paradox-free quantum causal loops and protocols with indefinite causal order (ICO) of agents’ quantum operations. The most general such abstract and spacetime agnostic quantum protocols that can be formulated in a compositional manner arise by considering sets of quantum operations that are composed in a well-defined manner, but possibly through feedback loops. These possibly cyclic structures can equivalently be described in several frameworks such as postselected closed timelike curves (P-CTCs) \cite{Lloyd_2011}, category-theoretic approaches to quantum theory \cite{pinzani2019categorical,selby2025generalised,coecke2017picturing}, cyclic quantum networks \cite{VilasiniRennerPRA,VilasiniRennerPRL}, and the cyclic quantum causal modelling formalism of \cite{ferradini2025}\footnote{A prior formalism for cyclic quantum causal models \cite{Barrett_2021} also exists, but this models a strict subset of P-CTCs which correspond to valid process operators.}.

However, current day quantum experiments are implemented within the regime of a classical background spacetime, whose light cone structure is acyclic. This raises the question of the connection between abstract quantum protocols and actual experiments and, once this has been clarified, which of these quantum protocols can then be realised within such experiments without invoking post-selection (as it is known that all such general cyclic protocols can be realised experimentally if we allow post-selection). A first step toward answering this from a top-down perspective was taken in \cite{VilasiniRennerPRA}, where it was shown that the most general quantum protocols (including those with cycles such as P-CTCs) that can be realised in a background spacetime compatible with relativistic causality, correspond to so-called causal boxes \cite{Portmann_2017}. Causal boxes model quantum protocols on an explicit background spacetime and allow for physical feedback loops consistent with relativistic causality (as opposed to causal loops permitting retrocausality), multiple rounds of information processing, superpositions in the spacetime position of quantum systems and are fully closed under composition. However, this leaves open the precise relationship between the aforementioned abstract protocols and protocols realisable within spacetime.

Within the spacetime agnostic landscape, a broad set of abstract quantum protocols that has received considerable interest goes variously under the names indefinite causal order, higher-order processes or process matrices \cite{Hardy2005,Oreshkov_2012,Chiribella2013}. These frameworks model quantum protocols where agents perform arbitrary quantum operations within “closed labs”, each agent receiving an input quantum system from and releasing an output quantum system to an “environment” exactly once, and they allow for protocols that lack a definite acyclic causal order between the quantum operations performed by the different agents.
Such process matrices are known to be a linear subset of the aforementioned P-CTCs \cite{Ara_jo_2017_CTC}. Various works have studied them for possible applications and advantages compared to protocols with a more standard fixed causal order \cite{Chiribella_2012, Colnaghi_2012, Ara_jo_2014, liu2024, Gu_rin_2016, Chiribella_2021, Zhao_2020, Guha_2020, felce2021refrigeration, Simonov_2025}. In particular, some indefinite causal order processes are said to be causally non-separable, meaning that they cannot be decomposed as convex mixtures of processes with a definite (but possibly dynamic) acyclic causal order \cite{Oreshkov_2012,Ara_jo_2015, Oreshkov_2016}. This can be thought of as analogous to entangled or non-separable quantum states, which cannot be expressed as convex mixtures of product states. Other process matrices can violate so-called causal inequalities which yield device-independent certification of the absence of such definite acyclic causal order, which can be thought of as analogous to Bell inequalities \cite{Oreshkov_2012,Branciard_2015}. Several other interesting subsets and examples of process matrices (or equivalently, quantum super maps or higher-order quantum processes) have been identified and continue to be widely studied (see e.g., \cite{Baumeler_2016, Guerin_2018,Wechs_2019, Wechs_2021, salzger2023mscthesis, mothe2025, Steffinlongo_2026}).

However, a longstanding open question is the following: \textit{Which indefinite causal order processes can be physically realised?} This question has been approached by imposing information-theoretic principles such as purifiability to narrow down the space of possible ICO processes \cite{Ara_jo_2017,Chiribella_2010}, but this still permits processes such as the Lugano process \cite{Baumeler_2014} which violate causal inequalities, where no clear physical interpretation (within known physics) has been found.\footnote{Although explicit models of this process within several operationally motivated frameworks such as routed quantum circuits and time delocalised subsystems have been proposed \cite{vanrietvelde2021routed,Wechs_2023}.}
The central question considered in this work adds a relativistic refinement to the above, in the spirit of \cite{VilasiniRennerPRA} while considering the physical regime of currently performable quantum experiments: \textit{Which indefinite causal order processes can be physically realised through experiments (in principle), in a meaningful way, in a classical acyclic spacetime?} Answering this question involves clarifying what ``meaningful'' entails, as generic experiments may be incompatible with the setup assumptions of the process matrix framework such as “closed labs” needed to ensure that certifications of indefinite order developed in that framework remain reliable. The situation is analogous to Bell scenarios and experiments. If an experiment produces a Bell violation, this is insufficient to call this a genuine violation of a Bell inequality or regard this as a genuinely non-classical phenomenon, even if the correlations produced are non-signalling. Instead, there are certain additional conditions which must be satisfied in the experimental setup (e.g., that the agents are unable to communicate due to spacelike separation) in order to ensure that the observed Bell inequality violation can be regarded as ``loophole free’’. Only under such constraints do we consider the Bell inequality violation to be meaningful  \cite{Hensen2015,Giustina2015,Shalm2015}. 

Understanding the physical structure of process matrices has also been approached in a bottom-up manner in \cite{Wechs_2021} by taking the notion of a quantum circuit and generalising it to allow for indefinite causal order, but without explicit reference to spacetime. The resulting quantum circuits with quantum control of causal orders (QC-QCs) exhibit quantum superpositions of the order of the agents’ operations which can be controlled by some quantum system or dynamically by the choice of other agents’ operations. In a previous work \cite{salzger2024mappingindefinitecausalorder}, we have shown that all QC-QCs can be embedded into classical spacetime as causal boxes satisfying a spatiotemporal version of the closed laboratory assumption of process matrices. We also showed there, in a constructive manner, that all QC-QCs, even those with indefinite causal order, once realised in spacetime as above, admit a description in terms of a definite acyclic causal order between more fine-grained operations, consistent with relativistic causality  (using the notion of fine-graining defined in \cite{VilasiniRennerPRA}). This makes precise in what sense QC-QCs can be thought of as physical protocols in classical spacetime but leaves the question of whether they are the most general process matrices that could be physically realised in this regime, which is the question we address here. 

Finally, it is important to mention that several sophisticated table-top experiments \cite{Procopio_2015, Rubino_2017,Goswami_2018,Wei_2019, Rubino_2022, Guo_2020, Goswami_2020, Taddei_2021} have already claimed to physically implement a canonical example of a QC-QC having ICO, the \emph{quantum switch} \cite{Chiribella2013}. The physical interpretation of such experiments, and whether they have succeeded in genuinely implementing indefinite causal order in Minkowski spacetime has been a subject of intense debate, that has touched upon deep foundational questions on the very definition of events and their localisation in quantum theory \cite{VilasiniRennerPRA,Chiribella2013,Paunkovi__2020, de_la_Hamette_2025, Kabel_2025, Oreshkov_2019, Ormrod_2023, Felce_2022,vilasini2025}. We will not go into the details of this debate here, but we will briefly return to it in \cref{sec:conclusion} to clarify how our results and conclusions hold independently of the ``side'' that one takes on this debate.

{\bf Summary of contributions and structure of paper:} 
The main contribution of this work is to show that any protocol or experiment realisable in an acyclic background spacetime that respects the closed labs assumption of the process matrix framework, necessarily behaves as a QC-QC. This answers the longstanding question on physical realisability of process matrices within this classical spacetime regime and provides a top-down derivation of the structure of QC-QCs from physical principles. Our result includes a large class of spacetimes allowed by general relativity and covers the regime of quantum experiments to date, in particular Minkowski spacetime as well as curved, globally hyperbolic spacetimes, but treated as a classical background that does not share back-reaction with the quantum process being realised.

In achieving this main result, we introduce a number of new techniques, additional results and examples which also outline future directions that might be of independent interest (see also \cref{fig:results_overview} for a visual summary of some of these). We review the causal box framework in \cref{sec:cb} and the higher-order processes and QC-QC frameworks in \cref{sec:hop}. A detailed summary of all main contributions along with the structure of the rest of the paper are given below. 

\begin{figure}[t!]
    \centering
\begin{tikzpicture}[scale=0.6, transform shape,
    thick,
    font=\small,
    >={Stealth}
]

\def\Rout{4.8}
\def\Rmid{3.4}
\def\Rin{1.8}
\def\Dx{8}

\begin{scope}[shift={(-\Dx,0)}]
\node at (0,6) {\Large{\shortstack{Subsets obtained via physical \\principles in spacetime}}};
\draw[blue] (0,0) circle (\Rout);
\draw (0,0) circle (\Rmid);
\draw (0,0) circle (\Rin);

\node[blue] at (0,4.0) {All q.\ causal structures (cyclic)};
\node at (0,2.6) {+rel.\ caus.\ in spacetime};
\node at (0,2.2) {$= \mathrm{CB}$};
\node at (0,1) {+closed labs};
\node at (0,0.6) {$= \mathrm{PB}$};

\fill (2,1.5) coordinate (GYNIa) circle (2.3pt);
\node[below] at (2,1.5) {$\mathrm{GYNI}^{a}$};

\fill (2.4,-0.9) coordinate (Ex) circle (2.3pt);

\fill (0,-2.5) coordinate (Luganoa) circle (2.3pt);
\node[left] at (0,-2.5) {Lugano$^{a}$};

\end{scope}

\begin{scope}[shift={(\Dx,0)}]
\node at (0,6) {\Large{\shortstack{Subsets obtained via abstract \\ information-theoretic approaches}}};
\draw[blue] (0,0) circle (\Rout);
\draw (0,0) circle (\Rmid);
\draw (0,0) circle (\Rin);

\node[blue] at (0,4.0) {All q.\ causal structures (cyclic)};
\node at (0,2.6) {PM};
\node at (0,1) {QC-QC};

\fill (-3,3) coordinate (GYNIC) circle (2.3pt);
\node[right] at (-3,3) {$\mathrm{GYNI}^{c}$};

\fill (0,-2.5) coordinate (Luganoc) circle (2.3pt);
\node[right] at (0,-2.5) {Lugano$^{c}$};

\fill (0,-0.9) coordinate (QCQCpt) circle (2.3pt);

\end{scope}


\draw[<->,red]
(-\Dx+\Rin,0) --node[anchor=south]{\cref{cor:bothmaps} (from \cref{theorem:pbtoqcqc}+\cite{salzger2024mappingindefinitecausalorder})} (\Dx-\Rin,0);

\draw[<->,red]
(Ex)
to[out=-5,in=-175] node[anchor=south]{\cref{ex:nolo}}
(QCQCpt);

\draw[<->,red]
(Luganoa)
to[out=-10,in=-170] node[anchor=south]{\cref{ex:lugano}}
(Luganoc)
;

\draw[<->,red]
(GYNIa)
to[out=10,in=170] node[anchor=south]{\cref{ex:trivvio}}
(GYNIC)
;

\end{tikzpicture}

    \caption{{\bf Visual overview of main result and examples} On both sides, the largest set labelled as ``all quantum causal structures (cyclic)'' corresponds more precisely to post-selected closed timelike curves (P-CTCs) \cite{Lloyd_2011}. On the left, relativistic causality in spacetime, gives causal boxes \cite{VilasiniRennerPRA}, while on the right, process matrices (PM) are known to be a linear subset of P-CTCs \cite{Ara_jo_2017_CTC}. Here (building on \cite{Vilasini_2020}), we define process boxes by further restricting causal box protocols with (spatiotemporal) closed labs, in the top-down approach shown on the left. On the right, QC-QCs are defined in a bottom-up manner \cite{Wechs_2021}, explicitly as a subset of PMs.  Our main result, \cref{theorem:pbtoqcqc} maps PBs to QC-QCs, and our prior work \cite{salzger2024mappingindefinitecausalorder}  maps QC-QCs to PB, both in a manner that preserves relevant operational behaviour (cf. \cref{cor:bothmaps}). \cref{ex:trivvio} shows that there exists a CB protocol that maximally violates the GYNI causal inequality \cite{Branciard_2015} trivially in a definite acyclic causal structure (denoted by the point GYNI$^a$), but  on the right side, such maximal GYNI violation require cyclic scenarios with perfect 2-way signalling (denoted GYNI$^c$) which are not valid PMs \cite{Barrett_2021,Yokojima_2021}. \cref{ex:lugano} shows that the Lugano process, originally with a cyclic structure (Lugano$^c$ on right) admits an acyclification in spacetime (Lugano$^a$ on left) but this must violate (spatiotemporal) closed labs. Finally, \cref{ex:nolo} shows that there also exist CB protocols violating this closed labs condition which can still map to QC-QCs. Different types of violations of closed labs (weak vs strong) are distinguished in \cref{sec:relaxing} where we find that \cref{ex:trivvio} and \cref{ex:lugano} (which do not map to QC-QCs) entail strong violations while \cref{ex:nolo} (which maps to a QC-QC) entails a weak violation. }
    \label{fig:results_overview}
\end{figure}
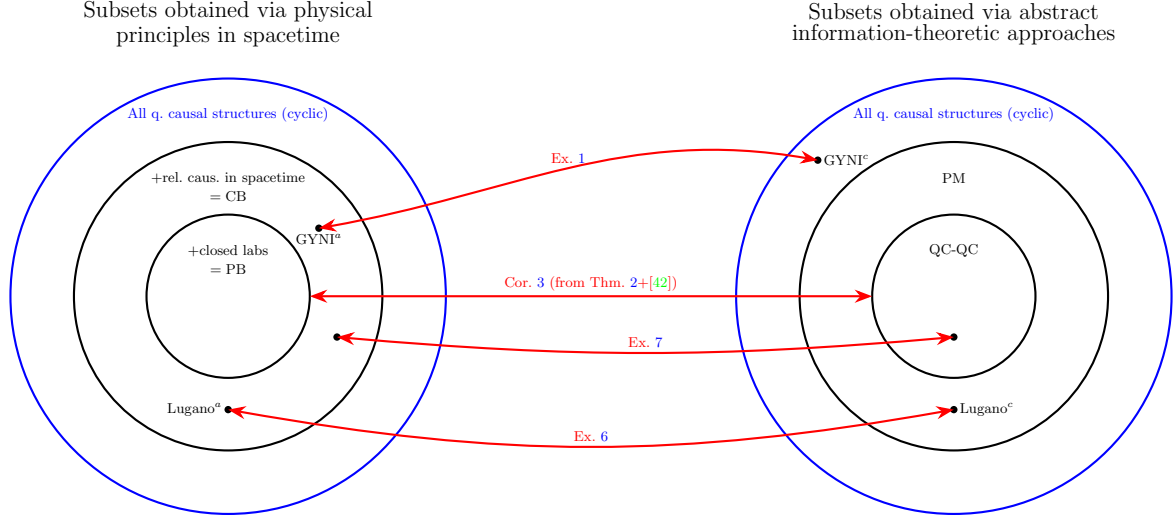

\begin{itemize}
    \item {\bf The framework (causal box and process box protocols):} In \cref{sec:multipartite}, we first introduce multi-partite causal box (CB) protocols, for describing relativistic quantum information protocols involving $N$ agents interacting with each other via an ``environment'', where the agents' operations and the environment are both modelled as causal boxes. This brings CBs one step closer to the set-up of process matrices (PMs). However, observing that such CB protocols still allow for trivial, maximal violations of causal inequalities (\cref{ex:trivvio}) that is not possible with PMs, in \cref{sec:pbdef}, we restrict them further via two key principles capturing the closed labs assumption of the PM framework at the fine-grained, spacetime level: Acting Once (AO) and Local Order (LO), together called \emph{spatiotemporal closed labs}, which define the subset of CB protocols called process box (PB) protocols. Originally the idea of restricting CBs through such conditions, to define PBs was proposed in \cite{Vilasini_2020}. Here we fully formalise this at the level of CB and PB protocols, showing that these are joint properties of the ``environment'' and the agent operations, and thus providing a more general formalisation of AO and LO than in \cite{Vilasini_2020}, with the interest of allowing the most general possible spacetime scenarios that can be said to respect the set-up assumptions of the process matrix framework. Moreover, we have provided an alternative formalisation of these conditions in a compositional manner, allowing easier characterisation and connection to other frameworks. 
    \item {\bf Characterisation results for process boxes:} Even with the spatiotemporal closed labs assumption imposed, general PB protocols can look very different from known process matrices or relevant subsets thereof. In \cref{sec:characterisation}, we detail these differences, and then present a key characterisation result (\cref{thm:simplifying}) which shows that the structure of PB protocols can in fact be considerably simplified while preserving their relevant operational behaviour. In particular, this includes simplifying the partially ordered spacetime to a totally ordered one, which uses an important property of generalised invariance under relabelling (of spacetime locations) of CB protocols established in \cref{lemma:relabeling} and from time-dependent to time-independent agent operations in PB protocols.
    \item {\bf Mapping process boxes to QC-QCs:} In \cref{sec:mapping}, we first present an example of a physically motivated PB protocol, based on the experiment for a long coherence time quantum switch \cite{Goswami_2018,Goswami_2020}, and its mapping to QC-QCs, highlighting that multiple PB protocols (more fine-grained description) can map to the same QC-QC protocol (coarse-grained description), just as an abstract process can have many different physical realisations. We then proceed to give our general mapping from process boxes to QC-QCs, defining what it means for them to be \emph{behaviourally equivalent}, and culminating in the main theorem (\cref{theorem:pbtoqcqc}) which shows that for every PB protocol, there exists a behaviourally equivalent QC-QC protocol.
    \item {\bf Physical principles and ruling out the Lugano process:} An important question in quantum causality has been to understand why the so-called Lugano process, a highly studied ICO process violating causal inequalities, can or cannot be physically realised. Previous work \cite{Wechs_2023} showed that the Lugano process admits a realisation in the framework of time-delocalised subsystems \cite{Oreshkov_2019}. Our results, on the other hand, provide a top-down approach from physical principles, to explain why it is impossible to realise the Lugano process in any experiment in a fixed spacetime which respects the closed labs assumption even on a fine-grained level. We compare our results and assumptions with those of \cite{Wechs_2023} in detail, to illustrate how there is no mathematical contradiction between the two results. In essence, by virtue of being a regular quantum circuit, the Lugano process construction of \cite{Wechs_2023} can be described in terms of valid causal boxes \cite{Portmann_2017} and thus respects the general definitions of spatiotemporal realisation and relativistic causality proposed in \cite{VilasiniRennerPRA, VilasiniRennerPRL}.\footnote{\cite{VilasiniRennerPRA, VilasiniRennerPRL} only defined necessary conditions for spatiotemporal realisations of processes and relativistic causality, sufficient for the general no-go theorems considered there. Those conditions are insufficient to guarantee that such a ``realisation'' is loophole free i.e., respects the set-up assumptions like closed labs, which are further formalised at the fine-grained level, in the process box framework considered here.} However, it is not a valid process box protocol as it violates the spatiotemporal closed labs condition of this framework. There is no contradiction because
   the closed labs condition of \cite{Wechs_2023} is strictly weaker than our spatiotemporal closed labs, such that the Lugano process satisfies the former but necessarily violates the latter. We argue that for loophole free tests of causal inequalities (analogous to such tests of Bell inequalities), one must impose such assumptions at the spacetime implementation level as we do (as opposed to the abstract level), and the analysis highlights that realisation with time delocalised subsystems (although interesting in its own right) is not sufficient to ensure physical and loophole free experimental realisations in a background spacetime. 
    \item {\bf Possible relaxations of the assumptions:} In \cref{sec:relaxing}, we detail the scope for further relaxing our spatiotemporal closed labs (i.e., AO+LO) assumption, distinguishing weak vs strong form of violation of each of AO and LO. While our example (\cref{ex:trivvio}) for trivial violation of causal inequalities, as well as the Lugano process (\cref{ex:lugano}) violate at least one of these conditions strongly (and cannot be mapped to QC-QCs), we present a novel example in \cref{ex:nolo} which violates both of them weakly and can still be mapped to a behaviourally equivalent QC-QC protocol. This outlines clear scope for further generalising our results, as well as a number of open questions on the characterisation of causal box protocols, and their relationships (in particular, via coarse-graining) with different classes of abstract cyclic quantum causal structures (see also \cref{fig:results_overview}). 
\end{itemize}

\section{Review}

\subsection{Causal boxes}\label{sec:cb}

\subsubsection{Messages and Fock spaces}\label{sec:fock}

The causal box framework models a background spacetime minimally as a partially ordered set $\C{T}$ as a partial order is sufficient to capture the lightcone structure, which is the main object of interest for causality. In the original work \cite{Portmann_2017}, the set $\C{T}$ can be infinite, whereas for us the finite case will be sufficient. 

A quantum message $\ket{\psi}$ sent/received at time $t$ is represented as a state $\ket{\psi}^A \otimes \ket{t}^{\C{T}}$. we will usually write this as $\ket{\psi, t}^{A, \C{T}}$ or $\ket{\psi, t}^A$ for compactness, dropping the superscript $\C{T}$ if the set of positions is unambiguous. The element $\ket{\psi}^A \in \C{H}^A$ encodes the information of the message sent or received on wire $A$ while $\ket{t} \in \C{H}^\C{T} \cong \mathbb{C}^{|\C{T}|}$ tells us the position of the message in spacetime. 

The space $\C{H}^A \otimes \C{H}^{\C{T}}$, which we will usually write compactly as $\C{H}^{A, \C{T}}$ is then the 1-message space. This is, however, not the full state space. More generally, the framework allows an arbitrary number of messages distributed over spacetime, including superpositions of different number of such messages, which can be modelled using the symmetric Fock space associated with this 1-message space
\begin{equation}\label{eq:fockdef}
    \C{F}^{A, \C{T}} := \C{F}(\C{H}^{A, \C{T}}) = \bigoplus_{n=0}^\infty \vee^n \C{H}^{A, \C{T}}.
\end{equation}
The symbol $\vee^n$ denotes $n$-fold symmetric tensor product. The symmetric Fock space is used here as the position labels already include all the ordering relations. This is particularly important for messages at the same position, as the two states $\ket{0, t}^A \otimes \ket{1, t}^A$ and $\ket{1, t}^A \otimes \ket{0, t}^A$ should really be undistinguishable. The one-dimensional subspace $(\C{H}^{A, \C{T}})^{\otimes 0}$ of the Fock space $\C{F}^{A,\C{T}} $ corresponds to the vacuum state $\ket{\Omega}^{A,\C{T}}\cong \ket{\Omega,\C{T}}^A$, 
i.e., there are no messages on the wire $A$ at any position in $\C{T}$. 

There are two wire isomorphisms in the causal box framework \cite{Portmann_2017} which we will use frequently and implicitly (that is to say, we will use equality signs, even if the equality only holds up to these isomorphisms) in this work. For any $\C{H}^A, \C{H}^B$ and any $\C{T}' \subseteq \C{T}$, we have
\begin{gather}\label{eq:wireiso}
\begin{aligned}
    \F{C}{\C{T}} &\cong \F{A}{\C{T}} \otimes \F{B}{\C{T}} \\
    \F{A}{\C{T}} &\cong \F{A}{\C{T}'} \otimes \F{A}{\C{T}\backslash \C{T}'}
\end{aligned}
\end{gather}
where $\C{H}^C = \C{H}^A \oplus \C{H}^B$. These isomorphisms tell us that we can split and combine wires. Two wires carrying messages corresponding to the Hilbert spaces $\C{H}^A, \C{H}^B$ can be combined into a single wire carrying messages corresponding to the direct sum. Similarly, two wires corresponding to two different spacetime regions can be combined into a single wire covering both regions. As we will see, these isomorphisms allow to first define causal boxes with just a single input and output wire which can then be split up into multiple in and outputs later when it is convenient.

We can also naturally embed any wire corresponding to some spacetime region $\C{T}'$, $\F{A}{\C{T}'}$, into a larger spacetime region $\C{T}$, by appending the vacuum on $ \C{T} \backslash \C{T}'$.
\begin{equation}
\label{eq:wireisostate}
   \F{A}{\C{T}'} \cong \F{A}{\C{T}}\otimes \ket{\Omega, \C{T} \backslash \C{T}'}. 
\end{equation}

\subsubsection{Definition of causal boxes}\label{sec:defcb}

We give now a simplified definition of causal boxes that captures those causal boxes defined on a finite $\C{T}$. This captures that causal boxes respect special relativity, meaning that there should be no signalling outside the lightcone (given by the partial order of $\C{T}$). This is formalised with the help of causality functions. 
\begin{defi}[Causality function \cite{Portmann_2017}]\label{def:chi}
A causality function $\chi$ for a finite spacetime $\C{T}$ maps bottom-closed subsets\footnote{A set $\C{T}' \subseteq \C{T}$ is bottom-closed (in $\C{T}$) iff for all $t, t' \in \C{T}$ with $t \preceq t'$, $t' \in \C{T}'$ implies $t \in \C{T}'$. In \cite{Portmann_2017}, the more involved notions of cuts and bounded cuts (Def. 4.2 there) of the spacetime is used. However, in our case (i.e., finite spacetimes) there is no distinction between bottom-closed subsets, cuts and bounded cuts.} of $\C{T}$ to bottom-closed subsets of $\C{T}$ and satisfies the following conditions for any non-empty $\C{T}', \C{T}'' \subseteq \C{T}$ which are bottom-closed
\begin{gather}
\begin{aligned}
    \chi(\C{T}' \cup \C{T}'') &= \chi(\C{T}') \cup \chi(\C{T}'')\\
    \C{T}' \subseteq \C{T}'' &\implies \chi(\C{T}') \subseteq \chi(\C{T}'') \\
    \chi(\C{T}') &\subsetneq \C{T}'.
\end{aligned}
\end{gather}
\end{defi}
The causality function tells us that in order to compute the outputs in some bottom-closed spacetime region $\C{T}'$, it is sufficient to know the inputs in a strict bottom-closed subset of that region given by $\chi(\C{T}')$. The fact that it is proper subset and that $\C{T}'$ is bottom-closed (i.e., it includes the past lightcones of all its elements) ensures consistency with special relativity.
\begin{defi}[Causal boxes \cite{Portmann_2017}]\label{def:causalbox}
A causal box is a system with an input wire $X$ and an output wire $Y$, together with a CPTP map
\begin{equation}
    \C{C}: \mathcal{L}(\F{X}{\C{T}}) \rightarrow \mathcal{L}(\F{Y}{\C{T}})
\end{equation}
that fulfils for all $\C{T}' \subseteq \C{T}$ which are bottom-closed
\begin{equation}\label{eq:cbreq}
    \text{tr}_{\C{T} \backslash \C{T}'} \circ \C{C} = \text{tr}_{\C{T} \backslash \C{T}'} \circ \C{C} \circ \text{tr}_{\C{T} \backslash \chi(\C{T}')} 
\end{equation}
for some causality function $\chi$, where $\text{tr}_{\C{T} \backslash \C{T}'}$ ($\text{tr}_{\C{T} \backslash \chi(\C{T}')}$) corresponds to tracing out all messages with positions in $\C{T} \backslash \C{T}'$ ($\C{T} \backslash \chi(\C{T}')$) and replacing it with the vacuum (this can also be understood in light of \cref{eq:wireisostate}). We call \cref{eq:cbreq} the causality condition. 
\end{defi}

Note that from this definition it directly follows that a causal box restricted to some set of positions $\C{T}' \subseteq \C{T}$ which is bottom-closed is once again a causal box. 

We can also consider subnormalised causal boxes. For this to be well-defined it is not sufficient to just replace the CPTP map in \cref{def:causalbox} with a CP map but to explicitly define subnormalised causal boxes as particular outcomes of causal boxes.
\begin{restatable}[Subnormalised causal boxes]{defi}{subcb}\label{def:subcb}
We call $\C{C}: \mathcal{L}(\F{X}{\C{T}}) \rightarrow \mathcal{L}(\F{Y}{\C{T}})$ a subnormalised causal box if there exists a (normalised) causal box $\C{C}':\C{L}(\F{X}{\C{T}}) \rightarrow \mathcal{L}(\F{YR}{\C{T}})$ and a product state $\ket{\psi}^R = \bigotimes_{t \in \C{T}} \ket{\psi_t, t}$ such that
\begin{equation}\label{eq:subcbcond}
    \C{C}(\cdot) = \bra{\psi} \C{C}'(\cdot) \ket{\psi}.
\end{equation}
\end{restatable}
Operationally, subnormalised causal boxes thus correspond to post-selecting the state on wire $R$. We note that in \cite{Portmann_2017}, the state $\ket{\psi}^R$ is assumed to always be the vacuum state. The two definitions are equivalent, but the above definition makes it easier to model different measurement outcomes as different subnormalised causal boxes obtained from the same normalised causal box, which will be useful for us later on. 

As the vacuum state is a product state it is clear that our definition here is at least as general as the definition of \cite{Portmann_2017}. On the other hand, if we have $\C{C}(\cdot) = \bra{\psi} \C{C}'(\cdot) \ket{\psi}$ as in the definition above with $\ket{\psi}^R = \bigotimes_{t \in \C{T}} \ket{\psi_t, t}$, then there exists for each $t \in \C{T}$ a unitary $U_t: \F{R}{t} \rightarrow \F{R}{t}$ such that $U_t \ket{\psi_t, t} = \ket{\Omega, t}$. From this it follows immediately that $\bra{\Omega, \C{T}} \bigotimes_{t \in \C{T}} U_t \C{C}'(\cdot)\bigotimes_{t' \in \C{T}} U_{t'}^\dagger \ket{\Omega, \C{T}} = \bigotimes_{t \in \C{T}} \bra{\psi, t} \C{C}'(\cdot) \bigotimes_{t' \in \C{T}} \ket{\psi_{t'}, t'} = \C{C}(\cdot)$. The map $\bigotimes_{t \in \C{T}} U_t \C{C}'(\cdot)\bigotimes_{t' \in \C{T}} U_{t'}$ is a causal box as it is evidently CPTP and $\text{tr}_t U_t \cdot U_t^\dagger = \text{tr}_t (\cdot)$ due to unitarity of $U_t$. Indeed, this argument shows that for any product state $\ket{\psi}$ there exists a normalised causal box $\C{C}'$ such that \cref{eq:subcbcond} holds. 

\subsubsection{Representations of causal boxes}\label{sec:repcb}

There are two particularly useful representations of causal boxes, one of these being the Choi-Jamiołkowski representation and the other the so-called sequence representation \cite{Portmann_2017}. The Choi-Jamiołkowski representation, however, involves some subtleties as the normal Choi-Jamiołkowski operator may become unbounded and thus ill-defined when considering infinite-dimensional Hilbert spaces \cite{Holevo_2011}. In this case, the Choi-Jamiołkowski representation of a CPTP map $\C{C}: \C{L}(\C{H}^X) \rightarrow \C{L}(\C{H}^Y)$ is a sesquilinear form defined on  $\C{H}^X \times \C{H}^Y=\text{span}\{\ket{\psi^X}\otimes \ket{\psi^Y}|\ket{\psi^X}
\in \C{H}^X,\ket{\psi^Y}
\in \C{H}^Y\}$ (i.e., the \textit{finite} linear span) as follows
\begin{gather}
\begin{aligned}
    R_{\C{C}}&: (\C{H}^X \times \C{H}^Y) \times (\C{H}^X \times \C{H}^Y) \rightarrow \mathbb{C} \\
    &   \ket{\psi^X}\otimes \ket{\psi^Y} \otimes  \ket{\varphi^X} \otimes \ket{\varphi^Y}\mapsto \bra{\psi^Y} \C{C}(\ket{\bar{\psi}^X} \bra{\bar{\varphi}^X}) \ket{\varphi^Y}
\end{aligned}
\end{gather}
where $\ket{\bar{\psi}^X} = \sum_i \ket{i} \overline{\braket{i | \psi^X}}$. This is more compactly written as
\begin{equation}
     R_{\C{C}}(\psi^X\otimes \psi^Y,\varphi^X\otimes \varphi^Y):= \bra{\psi^Y} \C{C}(\ket{\bar{\psi}^X} \bra{\bar{\varphi}^X}) \ket{\varphi^Y}
\end{equation}
If $\C{H}^X \times \C{H}^Y$ coincides with $\C{H}^X \otimes \C{H}^Y$, then the usual Choi-Jamiołkowski operator is bounded and can be recovered from this sesquilinear form \cite{Holevo_2011,Portmann_2017}.

The sequence representation is a sequence of isometries $V_1,..., V_M$, such that $V = V_M \circ V_{M-1} \circ ... \circ V_1$ is a purification of the causal box. Further, each $V_m$ acts exclusively on some disjoint slice of spacetime $\C{T}_m$ and outputs on the ``next'' slice $\C{T}_{m+1}$, the idea thus being that $V_m$ represents the action of $\C{C}$ on this particular slice of the overall spacetime $\C{T}$. More explicitly, for a causal box with input wire $X$, output wire $Y$, $V_1$ prepares a pure normalised state on $\C{F}^{Y, \C{T}_{1}} \otimes \C{H}^{\alpha_1}$, $V_m: \C{F}^{X, \C{T}_{m-1}} \otimes \C{H}^{\alpha_{m-1}} \rightarrow \C{F}^{Y, \C{T}_{m}} \otimes \C{H}^{\alpha_m}$ for $1<m<M$ and $V_M: \C{F}^{X, \C{T}_{M-1}} \otimes \C{H}^{\alpha_{M-1}} \rightarrow \C{H}^{\alpha_M}$. The sets $\{\C{T}_m\}_{m=1}^{M-1}$ form a partition of $\C{T}$, i.e., $\C{T} = \bigsqcup_{m=1}^{M-1} \C{T}_m$ (where $\bigsqcup$ denotes disjoint union) and are ordered in the sense that for $m < m'$, there exists no $t \in \C{T}_m, t' \in \C{T}_{m'}$ with $t \succ t'$. The (possibly infinite-dimensional) Hilbert spaces $\C{H}^{\alpha_m}$ can be viewed as internal memory wires. The existence of this representation essentially follows from the Stinespring dilation \cite{Stinespring_1955} of CPTP maps. Appropriate slices can be determined in a recursive fashion via a causality function of the causal box as $\chi(\bigsqcup_{m=1}^{n+1} \C{T}_m) = \bigsqcup_{m=1}^{n} \C{T}_m$. 

\subsubsection{Composition of causal boxes}\label{sec:composition}

Causal boxes are fully compositional under a general set of physical compositions. Indeed, besides the usual parallel and sequential ways of composing CPTP maps, we can also feed the outputs of a causal box back as an input, which is called loop composition. The causality condition then ensures that these looped back outputs cannot affect themselves (or anything outside of their future lightcone). 

\begin{defi}[Loop composition \cite{Portmann_2017}]\label{def:comp}
Let $\C{C}: \mathcal{L}(\F{AB}{\C{T}})\rightarrow\mathcal{L}(\F{CD}{\C{T}})$ be a subnormalised causal box with input systems $A$ and $B$ and output systems $C$ and $D$ with $\C{H}^B \cong \C{H}^C$. Let $\{\ket{k}^C\}_k$ be any orthonormal basis of $\F{C}{\C{T}}$, and denote with $\{\ket{k}^B\}_k$ the corresponding basis of $\F{B}{\C{T}}$ i.e. for all $k$, $\ket{k}^C\cong\ket{k}^B$. The result of looping the output system $C$ to the input system $B$ is denoted as $\C{C}^{C \hookrightarrow B}$ and is given by its Choi-Jamiołkowski representation as
\begin{gather}
\begin{aligned}
    \label{eq:loop}
    R_{\C{C}^{C\hookrightarrow B}} (\psi^A\otimes \psi^D,\varphi^A\otimes \varphi^D) &= \sum_{k,l} R_{\C{C}} (\psi^A\otimes \bar{k} \otimes k \otimes \psi^D, \varphi^A \otimes \bar{l} \otimes l \otimes \varphi^D) \\
    &= \sum_{k,l} \bra{\psi^D}  \bra{k}\C{C}(\ket{\bar{\psi}^A} \bra{\bar{\phi}^A} \otimes \ket{k}\bra{l})\ket{l} \ket{\varphi^D}
\end{aligned}
\end{gather}
\end{defi}
The sequential composition of two subnormalised causal boxes $\C{C}_A: \C{L}(\F{A}{\C{T}}) \rightarrow \C{L}(\F{C}{\C{T}})$ and $\C{C}_B: \C{L}(\F{B}{\C{T}}) \rightarrow \C{L}(\F{D}{\C{T}})$ can be represented as $(\C{C}_A \otimes \C{C}_B)^{C\hookrightarrow B}$. Causal boxes are closed under arbitrary parallel, sequential and loop compositions and the result of composition is invariant under different orders of composition \cite{Portmann_2017, VilasiniRennerPRA, VilasiniRennerPRL}. In particular, this implies that for a (sub)normalised causal box $\C{C}$ from inputs $A$ and $B$ to outputs $C$ and $D$, the loop composition $\C{C}^{C \hookrightarrow B}$ in the above definition is also a valid (sub)normalised causal box from input $A$ to output $D$.

We will generally use the same label to denote isomorphic pairs of input and output systems such as $B$ and $C$ above, in which case we will use the symbol $\hookrightarrow$ as a superscript to denote full loop composition (i.e., loop compose all outputs with matching inputs if and only if such an input exists) of a map $\C{C}$, i.e., $\C{C}^{\hookrightarrow}$, similar to how $\circ$ is often used to denote full sequential composition (i.e., sequentially compose all outputs of one map with inputs of the other map if and only if such an input exists). For example, consider two maps $\C{C}_1: \C{L}(\F{A}{\C{T}} \otimes \F{B}{\C{T}})\rightarrow \C{L}(\F{A}{\C{T}} \otimes \F{C}{\C{T}})$ and $\C{C}_2: \C{L}(\F{C}{\C{T}} \otimes \F{D}{\C{T}})\rightarrow \C{L}(\F{B}{\C{T}} \otimes \F{E}{\C{T}})$. Then, $\C{C}_2 \circ \C{C}_1$ consists of feeding the output $C$ of $\C{C}_1$ into the matching input $C$ of $\C{C}_2$ and we obtain a map from the systems $A, B, D$ to the systems $A, B, E$. On the other hand, $\lc{\C{C}_2 \circ \C{C}_1}$ would additionally loop the output $A$ of $\C{C}_1$ to the input $A$ of $\C{C}_1$ and the output $B$ of $\C{C}_2$ to the input $B$ of $\C{C}_1$, ultimately yielding a map from $D$ to $E$. Note that for an arbitrary number of maps $\C{C}_1,..., \C{C}_N$, as a consequence of order-independence of composition of causal boxes, loop composition is in particular invariant under cyclic permutations (it only depends on the pairs of systems, one output and one input, being composed) \cite{Portmann_2017, VilasiniRennerPRA, VilasiniRennerPRL},
\begin{equation}
    (\C{C}_N \circ ... \circ \C{C}_1)^{\hookrightarrow} = (\C{C}_k \circ \C{C}_{k-1}
\circ ... \circ \C{C}_{k+1})^{\hookrightarrow}
\end{equation}
for all $k=1,..., N$. If further there exist no two output systems which are the same and no two input systems which are the same among the sets of input respectively output systems of $\C{C}_1, ..., \C{C}_N$, we can even unambiguously write
\begin{equation}
    (\C{C}_N \circ ... \circ \C{C}_1)^{\hookrightarrow} = (\C{C}_N \otimes ... \otimes \C{C}_1)^{\hookrightarrow}
\end{equation}
by using the fact that sequential composition can be expressed in terms of parallel and loop composition. Note that the RHS is ambiguous when multiple input respectively output systems have the same label. For example, for maps $\C{C}_1: A \rightarrow A, \C{C}_2: A \rightarrow A$, it is ambiguous whether $\lc{\C{C}_1 \otimes \C{C}_2}$ refers to $\lc{\C{C}_2 \circ \C{C}_1}$ or $\lc{\C{C}_1} \otimes \lc{\C{C}_2}$. 

Instead of working with the sesquilinear form, one can also use the following representation which relies on the CP(TP) maps directly.

\begin{restatable}[An alternative representation for loop composition]{lemma}{looprep}\label{lemma:looprep}
    Let $\C{C}: \C{L}(\C{F}^{XY, \C{T}}) \rightarrow \C{L}(\C{F}^{YZ, \C{T}})$ be a (subnormalised) causal box. Let $\C{T}_1,...,\C{T}_M$ be a disjoint partition of $\C{T}$ such that $\chi(\bigsqcup_{m=1}^{n+1} \C{T}_m) = \bigsqcup_{m=1}^{n} \C{T}_m$ for all $m$ for some causality function $\chi$ of $\C{C}$. Then,
    \begin{equation}
        \C{C}^{\hookrightarrow}(\cdot) = \sum_{ij} \bra{i}^Y \C{C}(\cdot \otimes \ket{i}\bra{j}^Y) \ket{j}^Y
    \end{equation}
    for any basis $\{\ket{i}^Y\}$ of $\C{F}^{Y, \C{T}}$ which is a product basis in terms of the decomposition $\F{Y}{\C{T}} \cong \bigotimes_{m=1}^M \F{Y}{\C{T}_m}$.
\end{restatable}

\subsection{Higher-order quantum processes}\label{sec:hop}

The standard quantum circuit paradigm describes protocols with a well-defined acyclic ordering between the different operations (or quantum gates). This is consistent with a clear arrow of time. Recently, frameworks have been proposed which go beyond this to define abstract information-processing protocols without assuming the existence of a definite acyclic order between the different operations, or the existence of a background spacetime structure \cite{Barrett_2021, Oreshkov_2012, Chiribella2013, Wechs_2021}.

Mathematically, a higher-order quantum process or supermap takes a fixed number $N$ of CP maps $\C{M}_{A_k}: \C{L}(\C{H}^{A^I_k}) \rightarrow \C{L}(\C{H}^{A^O_k})$ and maps them to another CP map
\begin{equation}
  \C{W}:  (\C{M}_{A_1},...,\C{M}_{A_N}) \mapsto \C{M}: \C{L}(\C{H}^{P}) \rightarrow \C{L}(\C{H}^{F}).
\end{equation}
These supermaps must satisfy conditions which  are analogous to the conditions on CPTP maps. They map CP(TP) maps to a CP(TP) map and they must do so even if the maps $\C{M}_{A_k}$ are actually extended maps whose input and output both include an arbitrary ancilla. Physically, we think of the maps $\C{M}_{A_k}$ as modelling the local operations of agents $A_k$ in their own labs and the application of $\C{W}$ models how these combine into a global situation. The mathematical conditions imposed on the supermap ensure some basic consistency. For example, the probability for a given setup is $\text{tr}(\C{W}(\C{M}_{A_1},...,\C{M}_{A_N})(\rho))$ where $\rho \in \C{L}(\C{H}^P)$ and the fact that $\C{W}$ maps CPTP maps to CPTP maps ensures that the overall probability of any experiment adds up to 1. 

The supermap $\C{W}$ can be represented with a Choi-Jamiołkowski matrix, which is usually referred to as a process matrix \cite{Oreshkov_2012}. We can then write its action as
\begin{equation}\label{eq:supermap}
    W(\C{M}_{A_1},...,\C{M}_{A_N}) \coloneqq (M_{A_1} \otimes ... \otimes M_{A_N}) * W \in \C{L}(\C{H}^{PF})
\end{equation}
where $M_{A_i}$ denotes the Choi-Jamiołkowski matrix of the local operation $\C{M}_{A_i}$ and $*$ is the link product \cite{Chiribella_2008, Chiribella_2009}. Additionally, we proved in \cite{salzger2024mappingindefinitecausalorder} that the link product can equivalently be expressed using loop composition, hence, we will write the action of supermaps as 
\begin{equation}
    \C{W}(\C{M}_{A_1},...,\C{M}_{A_N}) = \lc{\C{W} \circ \bigotimes_{k=1}^N \C{M}_{A_k}}.
\end{equation}
Note that on the RHS $\C{W}$ should be interpreted as the CP(TP) map which is uniquely associated with the supermap.

\subsubsection{Quantum circuits with quantum control of causal order}

The framework of process matrices describes very general scenarios as no background spacetime is imposed. In particular, some of them go beyond what is achievable through processes compatible with some fixed order and violate causal inequalities \cite{Oreshkov_2012}, similar to how some entangled states violate Bell inequalities. However, at the same time, most process matrices currently lack a clear physical interpretation. 

Quantum circuits with quantum control of causal order (QC-QC) \cite{Wechs_2021} are a subclass of process matrices (including those with indefinite causal order) which generalise the notion of a quantum circuit to  allow for controlled superpositions of (possibly dynamic) causal orders\footnote{A similar framework was developed in \cite{Purves_2021}. We will work with the framework of QC-QCs.}. In a QC-QC the causal order of the agents' operations is either classically or quantumly controlled. The overall action is depicted in \cref{fig:qcqc}, illustrating the sense in which these processes can be interpreted as quantum circuits. We will briefly describe the most important features of QC-QCs as well as a mathematical characterisation which will be useful for us. For a detailed introduction we refer to the original work \cite{Wechs_2021}. 

\begin{figure}
    \centering
    \includegraphics[width=0.8\linewidth]{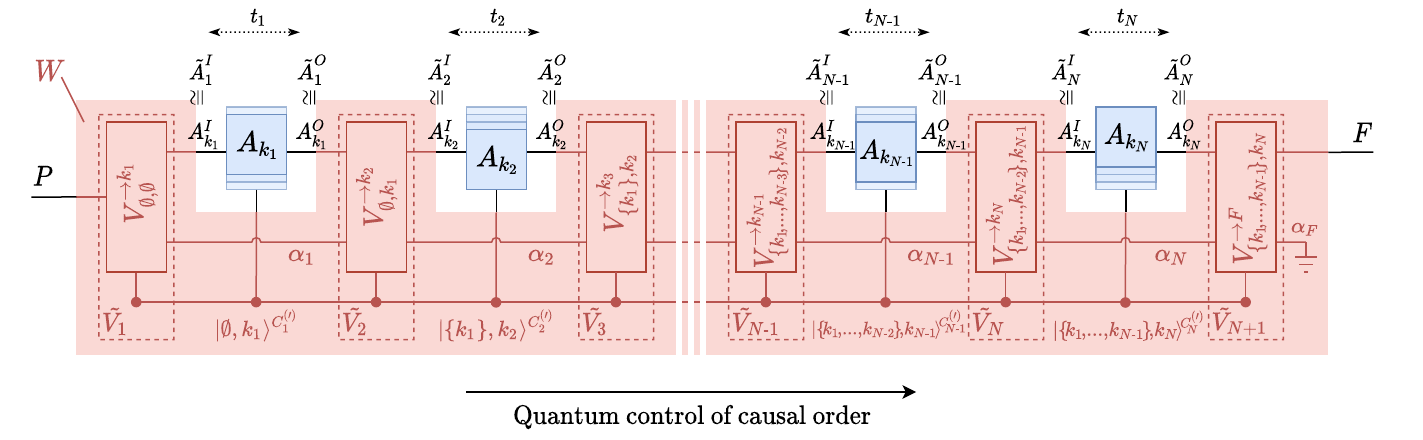}
    \caption{Graphical depiction of a quantum circuit with quantum control of causal order (QC-QC). At each step, the QC-QC (shaded pink) applies an internal operation conditioned on some control system (carried on the bottom wire). In each hole an agent operation is plugged in, but it is not fixed which one. Instead, this can depend classically or quantumly on the state of some control or dynamically on the action of previous operations. The figure is taken from \cite{Wechs_2021}, where it is listed as Fig. 10.}
    \label{fig:qcqc}
\end{figure}

We can imagine a QC-QC to be made up of $N$ steps. In any given step, say the $n$-th one, which agent acts is controlled via a system $\C{H}^{C_n}$ whose basis states take the form $\ket{\C{K}_{n-1}, k_n}$, where $\C{K}_{n-1} \subseteq \C{N}$ (with $\C{N} = \{1,...,N\}$ being the set of all agent labels) is a subset of cardinality $n-1$ and $k_n \in \C{N} \backslash \C{K}_{n-1}$. The state $\ket{\C{K}_{n-1}, k_n}$ indicates that the agents labelled by the elements of $\C{K}_{n-1}$ have previously acted in some arbitrary order and that agent $k_n$ is to act in the current step. As the control system can be in quantum superposition this introduces quantum uncertainty into the causal order. In between the agent operations the circuit applies an operation which can itself be conditioned on the control system and which sends the target system to the next agent, say $k_{n+1}$, updating the control in the process to $\ket{\C{K}_{n-1} \cup k_n, k_{n+1}}$. The internal operations of the circuit can be denoted by the state of the control and the agent the system is being sent to
\begin{gather}\label{eq:qcqcint}
\begin{aligned}
    V^{\rightarrow k_1}_{\emptyset, \emptyset}&: \C{H}^{P} \rightarrow \C{H}^{A^I_{k_1}} \otimes \C{H}^{\alpha_1}\\
    V^{\rightarrow k_{n+1}}_{\{k_1,...,k_{n-1}\}, k_n}&: \C{H}^{A^O_{k_n}} \otimes \C{H}^{\alpha_n} \rightarrow \C{H}^{A^I_{k_{n+1}}} \otimes \C{H}^{\alpha_{n_1{}}}\\
    V^{\rightarrow F}_{\{k_1,...,k_{N-1}\},k_N}&: \C{H}^{A^O_{k_N}} \otimes \C{H}^{\alpha_N} \rightarrow \C{H}^{F} \otimes \C{H}^{\alpha_F}
\end{aligned}
\end{gather}
Here, the systems $\alpha_n, \alpha_{F}$ are ancillary wires which the QC-QC uses to carry information internally. The above maps must then compose into a valid physical transformation for the QC-QC to be valid itself. Mathematically, this means that the maps
\begin{gather}\label{eq:qcqcint2}
\begin{aligned}
    V_1 &= \sum_{k_1} V^{\rightarrow k_1}_{\emptyset, \emptyset} \otimes \ket{\emptyset, k_1} \\
    V_{n+1} &= \sum_{\substack{\C{K}_{n-1}, k_n, \\k_{n+1}}} V^{\rightarrow k_{n+1}}_{\{k_1,...,k_{n-1}\}, k_n} \otimes \ket{\C{K}_{n-1} \cup k_n, k_{n+1}} \bra{\C{K}_{n-1}, k_n} \\
    V_{N+1} &= \sum_{k_N} V^{\rightarrow F}_{\C{N} \backslash k_N,k_N} \otimes \bra{\C{K}_N \backslash k_N, k_N}\\
\end{aligned}
\end{gather}
define isometries on their effective input space at least, that is on states that can actually be produced by the preceding operations $V_n$ together with arbitrary agent operations. Indeed, any set of maps as in \cref{eq:qcqcint} for which this is the case define a QC-QC and any QC-QC admits such a description.

The resulting supermap is then (up to tracing out the purifying degree of freedom $\alpha_F$)
\begin{equation}\label{eq:qcqcaction}
   (A_1,...,A_N) \mapsto \sum_{(k_1,...,k_N)} V^{\rightarrow F}_{\{k_1,...,k_{N-1}\},k_N} \circ A_{k_N} \circ V^{\rightarrow k_N}_{\{k_1,...,k_{N-2}\}, k_{N-1}} ... A_{k_{n+1}} \circ V^{\rightarrow k_{n+1}}_{\{k_1,...,k_{n-1}\}, k_n} \circ A_{k_n} ... A_{k_1} \circ V^{\rightarrow k_1}_{\emptyset, \emptyset}
\end{equation}
where we sum over all possible orders of the agents $(k_1,...,k_N)$ and $A_k: \C{H}^{A^I_k} \rightarrow \C{H}^{A^O_k}$ is an arbitrary Kraus operator (note that here we use the notation $A_k$ for both the agent and the Kraus operator they apply. It should be clear from context to which of the two we are referring in any given case).

\section{Multi-partite causal box protocols on finite spacetimes}\label{sec:multipartite}

In this section, we will discuss how to model multipartite information processing protocols using causal boxes. We consider $N$ agents $\{A_k\}_{k=1}^N$ implementing a protocol in a spacetime $\C{T}$ modelled as a partially ordered set with finite cardinality $|\C{T}|$. Note that this does not imply a fundamental discretisation of the spacetime, but simply that we consider protocols with a finite number of information-processing steps. Hence, a finite sample of relevant spacetime locations is sufficient to describe them. An element of $\C{T}$ could also represent a small enough spacetime region, instead of a spacetime point, as long as there is a partial order between such regions.\footnote{If the regions are too large, bidirectional influences between two regions becomes possible and we no longer have a partial order but rather a pre-order \cite{VilasiniRennerPRA, VilasiniRennerPRL}.} 

{\bf Agents' operations} To each agent $A_k$, we will associate a set of input positions $\C{T}^I_k \subseteq \C{T}$ and a set of output positions $\C{T}^O_k \subseteq \C{T}$ as well as a set of causal boxes describing the agent's allowed operations $\{\M\}_{x_k}$, where $x_k$ is a classical setting that parametrises the choice of operation. Note that the set of allowed $x_k$ here need not be finite, hence these sets of allowed operations are not necessarily finite, or even countable. In particular, the set $\{\M\}_{x_k}$ could include all valid causal boxes. Each allowed operation is therefore a causal box of the form
\begin{equation}\label{eq:agentop}
    \M: \C{L}(\FI) \rightarrow \C{L}(\FO \otimes (\C{H}^{R_k^{x_k}} \otimes \ket{T}))
\end{equation}
where we conceptualise the wire $R_k^{x_k}$ as carrying the agents' measurement outcome in the future, i.e., $T \succeq t$ for all $t \in \C{T}$. We allow this wire, in particular its dimension, to depend on the choice of operation for generality (indicated by the superscript $x_k$). This form is w.l.o.g, as long as agents cannot act on each other's outcome wires, which is indeed what will be ensured in the framework for capturing that agents act within separate closed labs.  Even if we considered multiple outcome wires with different spacetime positions, we could equivalently post-process them to a single outcome at $T$. We note a related work \cite{Emilien_inprep} where the explicit internal modelling of $\M$, leading to two distinct measurement models for agents in spacetime, is explored. While we do not explore these details here, our model here is general and consistent with both those measurement models.

Furthermore, while we do not model the outcome wire as a Fock space, we could always embed it into one, hence, it still makes sense to call $\M$ a causal box. We can write 
\begin{equation}
    \C{M}^{A_k}_{a_k|x_k} := \bra{a_k, T} \M \ket{a_k, T}
\end{equation}
corresponding to the outcome $a_k$. The map $\C{M}^{A_k}_{a_k|x_k}$ is thus a subnormalised causal box. 

{\bf Multi-partite causal box} The interaction between different agents will be mediated by what we call a multi-partite causal box, which is a causal box of the form
\begin{equation}
    \C{C}: \C{L}(\F{P}{t_P} \otimes \bigotimes_{k =0}^N \FO )\rightarrow \C{L}(\bigotimes_{k=1}^N \FI \otimes \F{F}{t_F})
\end{equation}  
from the outputs of a set of $N$ agents and a global past system $P$, to the inputs of the same $N$ agents together with a global future system $F$. The past and future spaces are taken to carry messages of dimension $d_P$ and $d_F$ with fixed spacetime labels $t_P$ and $t_F$, with $t_P\prec t \prec t_F \prec T$ for all other $t \in \C{T} \backslash \{t_P, t_F, T\}$. Note that this too can ultimately be seen as a causal box with a single input and output wire at all spacetime positions in $\C{T}$ by using \cref{eq:wireisostate} to append vacuum and then the wire isomorphism \cref{eq:wireiso}. Notice that a multi-partite process box is analogous to the channel representation of a process matrix which is a channel from the outputs of all agents and the global past to the inputs of all agents and the global future.

\begin{restatable}[Multi-partite causal box protocols]{defi}{CBP}\label{CBP}
    A multi-partite causal box protocol is a tuple $\mathfrak{P}:=(\C{C},(\{\M\}_{x_k})_{k=1}^N,\{\C{M}^P_{x_p}\}_{x_p},\{\C{M}^F_{x_f}\}_{x_f})$ specified by a multi-partite causal box $\C{C}$ together with a set of causal boxes $\{\M\}_{x_k}$ as defined in \cref{eq:agentop} for each agent $\{A_k\}_{k=1}^N$ along with sets of causal boxes $\{\C{M}^P_{x_P}\}_{x_P}$ and $\{\C{M}^F_{x_F}\}_{x_F}$ for the global past and future agents, where the maps $\C{M}^P_{x_P}$ have a trivial input space along with output space $P$ and the result space $R_P^{x_P}$, while for the maps of the global future we have
\begin{equation}
    \C{M}^F_{x_F}: \C{L}(\C{F}^{F, t_F}) \rightarrow \C{L}(\C{H}^{R_F^{x_F}} \otimes \ket{T}).
\end{equation} 
\end{restatable}
In the following, we will denote the past and future agents as the $k=0^{th}$ and $k=N+1^{th}$ agents and simply write $\PBP$ for the protocol involving $N+2$ agents, $\{A_0,...,A_{N+1}\}$. 

{\bf Composition} The composition of all the maps in a multi-partite causal box protocol $\mathfrak{P}$ (obtained by connecting wires with matching system and spacetime labels via loop composition) defines a state on the result wires $\C{L}(\bigotimes_{k=0}^{N+1} \C{H}^{R_k^{x_k}} \otimes \ket{T})$ 
\begin{equation}
    \CM{\C{C}}{\M} \in \C{L}(\bigotimes_{k=0}^{N+1} \C{H}^{R_k^{x_k}} \otimes \ket{T})
\end{equation}
or a probability, when using the subnormalised causal boxes $\C{M}^{A_k}_{a_k|x_k}$  for agents: 
\begin{equation}
    P(a_0,...,a_{N+1}|x_0,...,x_{N+1})=\CM{\C{C}}{ \C{M}^{A_k}_{a_k|x_k}} 
\end{equation}
The causal box framework guarantees that arbitrary composition of causal boxes via connecting systems with such matching labels yields a valid causal box, which implies that the above must be a valid probability distribution for any multi-partite causal box and agents' operations defined on corresponding spaces. 

\begin{figure}
    \centering
    \includegraphics[width=0.5\linewidth]{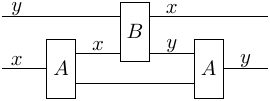}
    \caption{A trivial violation of a causal inequality. The label on each wire indicates the variable that wire carries at that point. Alice and Bob receive an input bit $x$ respectively $y$. Alice sends her bit to Bob who swaps the two bits $x, y$. Finally, Alice acts a second time and outputs the bit she received from Bob.}
    \label{fig:trivvio}
\end{figure}

Despite the structural similarities in the constructions so far, arbitrary causal box protocols should not be thought of as processes as the following example shows.

\begin{restatable}[Trivial causal inequality violation]{example}{trivvio}\label{ex:trivvio}
    We give an example of a bipartite causal box protocol (illustrated in \cref{fig:trivvio}) which maximally violates the causal inequality obtained from the so-called guess your neighbour's input game (GYNI) \cite{Branciard_2015}. In this game, Alice receives a bit $x$, outputs a bit $a$, Bob receives a bit $y$, outputs a bit $b$ and the goal for them is that $x=b$ and $y=a$. This game gives rise to a causal inequality under certain assumptions, more precisely, if the agents act within closed labs (acting once with inputs preceding outputs) and choose operations freely, then for any causally ordered strategy of Alice and Bob, their winning probability is at most $\frac{1}{2}$. 

    We will define the action of the agents and the causal box only on states which can actually occur in the protocol for simplicity. We assume that Alice can only receive an input at $t=4$ while Bob can only receive an input at $t=2$ and that all states are classical. The agents' operations are then
    \begin{gather}\label{eq:trivvio}
    \begin{aligned}
        \C{M}^{A}_{x} (\ket{y, t=4} \bra{y, t=4}^{A^I}) &= \ket{x,t=1} \bra{x, t=1}^{A^O} \otimes \ket{y, t=5}\bra{y, t=5}^{R_A} \\
        \C{M}^{B}_{y} (\ket{x, t=2} \bra{x, t=2}^{B^I}) &= \ket{y,t=3} \bra{y, t=3}^{B^O} \otimes \ket{x, t=5}\bra{x, t=5}^{R_B}.
    \end{aligned}
    \end{gather}
    The causal box acts as
    \begin{equation}
        \C{C} (\ket{x, t=1} \bra{x,t=1}^{A^O} \otimes \ket{y,t=3} \bra{y, t=3}^{B^O}) = \ket{x, t=2} \bra{x,t=2}^{B^I} \otimes \ket{y,t=4} \bra{y, t=4}^{A^I}.
    \end{equation}
    Composing these yields
    \begin{equation}
        \lc{\C{C} \circ (\C{M}^{A}_{x} \otimes \C{M}^{B}_{y})} = \ket{y, t=5}\bra{y, t=5}^{R_A} \otimes \ket{x, t=5}\bra{x, t=5}^{R_B}
    \end{equation}
    and we see that Alice and Bob win the GYNI game with certainty. 

    This causal inequality violation is, however, not a witness of indefinite causal order. The protocol is perfectly compatible with a well-defined causal order, simply not one where each agent acts once and only once with their inputs happening before their outputs, i.e., this violates the closed labs assumption. Indeed, from \cref{eq:trivvio}, we see that Alice sends out a message at $t=1$ before she receives one at $t=4$.
\end{restatable}

Multipartite causal box protocols model multipartite quantum experiments in spacetime. Two experimenters can make different choices in their experimental design, for example, they might choose different physical systems, different measurement apparatuses, or even just when and where they carry out their experiment. Still, if the observations we care about are the same in both, we would still think of their experiments as being equivalent in terms of their relevant behaviour. We formalise this in the definition below. This will in particular become relevant when we connect causal box protocols to QC-QCs. There may be many ways in which an abstract QC-QC can be physically realised in spacetime, each such realisation corresponding to some causal box protocol with behavioural equivalence among the different realisations of the same QC-QC.

\begin{defi}[Behavioural equivalence of causal boxes]\label{def:pb_eq}
Consider two $N$-partite causal box protocols, $\PBP$ and $\PBPp$. We say that $\mathfrak{P}$ and $\mathfrak{P}'$ are behaviourally equivalent if for all $k\in\{0,...,N+1\}$ and for each agent operation $\M$, there exists an agent operation $\Mp$, and vice versa, such that
\begin{equation}
    \CM{\C{C}}{\M} = \CM{\C{C}'}{\Mp}.
\end{equation}

\end{defi}

\section{Defining process boxes compositionally}\label{sec:pb}
\subsection{Definition}\label{sec:pbdef}

As we have seen in the previous section, general causal boxes can violate causal inequalities trivially via violations of the closed labs assumption of the process matrix/higher-order frameworks. Hence, we cannot think of arbitrary causal box protocols as ``processes'' in the same sense of the higher-order framework. The goal of this section is to define the subset of causal box protocols which do in fact deserve to be called processes, by formalising the closed labs assumption within the context of such spacetime protocols. An initial definition has previously been developed in \cite{Vilasini_2020} which first proposed the idea of \emph{process boxes}. Here, we will build on that work, while fully formalising and generalising the definition of process boxes in the process.

An important aspect of processes is that each agent acts exactly once, with an implicit local order: inputs are received before producing relevant outputs. Ignoring this assumption leads to trivial causal inequality violation as seen in \cref{ex:trivvio}. We first formalise the idea of an agent acting once as there being only a single message on their input and output wire. For this purpose, let us first define the 1-message subspace and the associated projector onto that space. 

\begin{defi}[1-message subspace supported in a spacetime region]\label{def:oneprojectors}
 
Consider a spacetime $\C{T}$. Let $A_k$ be an agent with non-trivial input/output space and corresponding Fock space $\C{F}^{A^{I/O}_k, \C{T}}$. For any subset $\C{T}' \subseteq \C{T}$, the \emph{1-message subspace supported in $\C{T}'$} is the subspace of $\C{F}^{A^{I/O}_k, \C{T}}$ consisting of states with
exactly one message located somewhere in $\C{T}'$ and vacuum at all other spacetime locations in $\C{T}$. This space can be written as
\begin{equation}
    \C{H}^{A^{I/O}_k,\C{T}'} = \C{H}^{A^{I/O}_k} \otimes \C{H}^{\C{T}'} \subsetneq\FIOall
\end{equation}
which is naturally embedded in the full Fock space by identifying
\begin{equation}
\ket{\psi , t} \mapsto \ket{\psi,t} \otimes \ket{\Omega, \C{T}\backslash\{t\}} .
\end{equation} 
  The projector onto this space from the full Fock space on $\C{T}$ is thus well-defined and can be denoted as
    \begin{equation}\label{eq:oneprojectors}
        \C{P}^{I/O, \C{T}'}_k: \C{L}(\FIOall) \rightarrow \C{L}(\C{H}^{A^{I/O}_k, \C{T}'}).
    \end{equation}
    If $\C{T}' = \C{T}$, we will also denote this projector as $\C{P}^{I/O}_k$. If $A_k$ has a trivial in/output, then the corresponding projector $\C{P}^{I/O,\C{T}'}_k$ is defined to be trivial for all $\C{T}'$ (i.e., simply acts as the identity on the trivial space).
\end{defi}
Note that $P^{I/O,\C{T}'}_k$ acts on the full Fock space over $\C{T}$. It projects onto states with exactly one message somewhere in $\C{T}'$ and vacuum on all other spacetime locations in $\C{T}$. These projectors should not be interpreted as acting independently on subspaces associated with individual locations. In particular, an expression like $\C{P}^{I, t}_k \otimes \C{P}^{I, t'}_k$ is not well-defined.

In general, we should not think of projectors on the Fock space as physical operations as they can allow backwards in time signalling (an explicit example is given in \cref{ex:sigproj}). However, in the case of the projectors in \cref{eq:oneprojectors} it turns out that we can interpret them as subnormalised pseudo-causal boxes, which we define as follows: A pseudo-causal box is essentially a causal box but processing can be instantaneous (this relaxes the condition that $\chi(\C{T}')\subsetneq \C{T}'$ for valid causality functions on bottom-closed sets $\C{T}'\subseteq \C{T}$ to allow $\chi(\C{T}')\subseteq \C{T}'$). A formal definition is given in \cref{defi:pseudocb} in \cref{app:cb}. There, we also show that composing either all the input or all the output wires of pseudo-causal boxes with a causal box yields a causal box (\cref{lemma:itscb}). This means that we can think of pseudo-causal boxes as physical operations where the processing time is allowed to go to 0 (in contrast to causal boxes where this must be strictly non-zero).

\begin{restatable}{lemma}{projcb}\label{lemma:projcb}
    The projectors $\C{P}^{I, \C{T}'}_k$ and $\C{P}^{I, \C{T}'}_k$ are subnormalised pseudo-causal boxes for all $k$ and all $\C{T}' \subseteq \C{T}$.
\end{restatable}

{\bf Physical intuition on 1-message projectors and subspaces} We now discuss an equivalent way to understand these projectors, which arises through the wire isomorphism. This serves to provide further physical intuition and to show consistency with related literature on causal/process boxes, in particular the initial process box proposal of \cite{Vilasini_2020} as well as our previous work \cite{salzger2024mappingindefinitecausalorder} and a related work involving one of us \cite{Emilien_inprep}. Moreover, this is also relevant for connecting to the routed circuit formalism (and its sectorial constraints) \cite{vanrietvelde2021routed} and the analysis of the fine-grained quantum switch in \cite{vilasini2025}.

Note that the 1-message space for any system $S$ of dimension $d_S$ associated with spacetime locations $\C{T}'$, is by definition equivalent to a direct sum of 1-message spaces of $S$ associated to each location $t\in\C{T}'$
\begin{equation}
\label{eq: 1msgspace_dirsum}
     \C{H}^{S, \C{T}'}=\C{H}^S\otimes \mathbb{C}^{|\C{T}'|}=\bigoplus_{t\in\C{T}'}    \C{H}^{S} \otimes \ket{t}=\bigoplus_{t\in\C{T}'}    \C{H}^{S, t}  =\text{span}\Big(\{\ket{i,t}\}_{i\in \{0,...,d_S-1\},t\in\C{T}'}\Big)  
\end{equation}
Moreover, this is isomorphic to a subspace of a tensor product of spaces associated to each location $t\in\C{T}'$, where each space in the product can contain a vacuum (zero-message) or a single non-vacuum (1-message) state
\begin{align}
    \begin{split}
    \label{eq: 1msgspace_span_tensorprod}
    \C{H}^{S, \C{T}'}\cong \text{span}\Big(\{\ket{i,t}\ket{\Omega,\C{T}'\backslash \{t\}}\}_{i\in \{0,...,d_S-1\},t\in\C{T}'}\Big)\subseteq \bigotimes_{t\in\C{T}'} \Big(  (\C{H}^S \oplus \ket{\Omega})\otimes \ket{t}\Big).
    \end{split}
\end{align}
Finally, by the trivial embedding in \cref{eq:wireisostate}, we also have for any $\C{T}'\subseteq \C{T}$,
\begin{equation}
    \text{span}\Big(\{\ket{i,t}\ket{\Omega,\C{T}'\backslash \{t\}}\}_{i\in \{0,...,d_S-1\},t\in\C{T}'}\Big)\cong \text{span}\Big(\{\ket{i,t}\ket{\Omega,\C{T}'\backslash \{t\}}\}_{i\in \{0,...,d_S-1\},t\in\C{T}'}\Big) \otimes \ket{\Omega,\C{T}\backslash \C{T}'}^S\subseteq \C{F}^{S,\C{T}}
\end{equation}
Therefore note that in particular, we can equivalently express the 1-message projector for any spacetime location $\C{T}'=\{t\}$ as the projector on the subspace  $\text{span}\Big(\{\ket{i,t}\ket{\Omega,\C{T}\backslash \{t\}}\}_{i\in \{0,...,d_S-1\}}\Big)$.

\begin{restatable}[Signalling measurements]{example}{sigproj}\label{ex:sigproj}
    Consider the Fock space $\C{F}(\mathbb{C}^2 \otimes \C{H}^{\C{T}})$ where $\C{T} = \{1,2\}$ and the projector $\C{P}_{\leq 1} = P_{\leq 1}(\cdot)P_{\leq 1}$, where $P_{\leq1}$ projects onto the subspace containing 0 or 1 messages, $(\mathbb{C}^2 \oplus \ket{\Omega}) \otimes \C{H}^{\C{T}}$. This projector allows for backwards-in-time signalling, no matter how we choose other orthogonal projectors $\{\C{P}_i\}_i$ to obtain a complete measurement
    \begin{equation}
        \text{tr}_{t=2} \circ (\C{P}_{\leq 1} + \sum_i \C{P}_i) \neq \text{tr}_{t=2} \circ (\C{P}_{\leq 1} + \sum_i \C{P}_i) \circ \text{tr}_{t=2},
    \end{equation}
    that is, the choice of input state at $t=2$ can influence the output state at the earlier time $t=1$. To see this, consider the two states $(\ket{\Omega, t=1}+ \ket{0, t=1})\ket{\Omega, t=2}/\sqrt{2}$ and $(\ket{\Omega, t=1}+ \ket{0, t=1}) \ket{0,t=2} /\sqrt{2}$. Tracing out the $t=2$ time stamp yields the same state at $t=1$, $(\ket{\Omega, t=1}+ \ket{0, t=1})/\sqrt{2}$, and so applying the measurement to these two states and then tracing out the $t=2$ time stamp should yield the same state. The first state lies in the image of $\C{P}_{\leq 1}$, hence, the measurement acts trivially on it and after tracing out the $t=2$ time stamp, we obtain the pure state $(\ket{\Omega, t=1}+ \ket{0, t=1})/\sqrt{2}$. For the second state, we find 
    \begin{gather}
    \begin{aligned}
        \text{tr}_{tr=2} \circ (\C{P}_{\leq 1} &+ \sum_i \C{P}_i) ((\ket{\Omega, t=1}+ \ket{0, t=1})(\bra{\Omega, t=1}+ \bra{0, t=1}) \otimes \ket{0,t=2}\bra{0, t=2} /2) \\
        &= \frac{1}{2} \ket{\Omega, t=1} \bra{\Omega, t=1}  +\frac{1}{2}\sum_i \text{tr}_{t=2} \circ \C{P}_i (\ket{0, t=1}\bra{0, t=1} \otimes \ket{0,t=2}\bra{0,t=2}).
    \end{aligned}
    \end{gather}
    This is either a mixed state or equal to vacuum at $t=1$. In both cases, it is different from $(\ket{\Omega, t=1}+ \ket{0, t=1})/\sqrt{2}$ and so we observe signalling from $t=2$ to $t=1$, backwards in time.
\end{restatable}

With the projectors of \cref{def:oneprojectors}, we can now formalise the notion of each agent acting exactly once.

\begin{defi}[Acting once condition for a causal box protocol]\label{def:ao}
    Given a causal box protocol $\mathfrak{P}$ we say that $\mathfrak{P}$ satisfies acting once (AO) if for all $x_k$
     \begin{equation}
        \CM{\C{C}}{(\C{P}^O_k \circ \M \circ \C{P}^I_k)} = \CM{\C{C}}{\M}
    \end{equation}
\end{defi}

\begin{figure}
    \centering
    \includegraphics{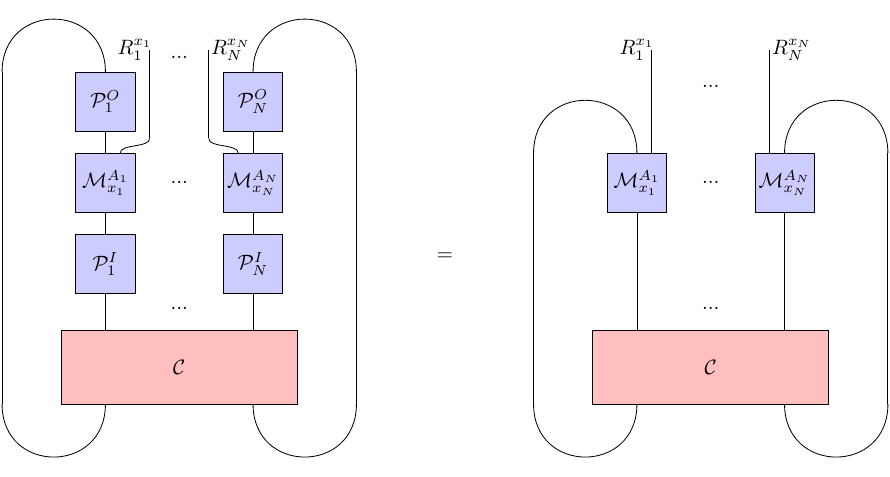}
    \caption{Graphical depiction of the Action Once (AO) condition of \cref{def:ao}. The projectors $\C{P}^I_k, \C{P}^O_k$ can be viewed as physical operations (\cref{lemma:projcb}), corresponding to measuring the number of messages on an agent's input respectively output wire over the entire experimental procedure and getting the outcome ``1 message''. The condition then states that this measurement is non-disturbing, verifying that the agent only received exactly a single message and sent exactly a single message.}
    \label{fig:aocond}
\end{figure}

This condition is illustrated graphically in \cref{fig:aocond}. In other words, in the causal box protocol, it is guaranteed that each agent receives and sends one message at some spacetime location in $\C{T}$ (or more precisely, in $\C{T}^{I/O}_k \subseteq \C{T}$ for agent $A_k$) or a superposition thereof.

Note that the above condition also restricts the global past and future to their respective one-message spaces. In particular, this restricts the global past's operations, which for general causal box protocols are states in the Fock space, to states in the 1-message space. The projectors $\C{P}^{I}_{0}$ and $\C{P}^{O}_{N+1}$ are trivial since the past input and future output spaces are trivial. 

\Cref{lemma:projcb} implies that we can interpret the AO condition in an operational manner. There exists an operational method (corresponding to the pseudo-causal boxes of which the projectors $\C{P}^{I/O}_k$ are a particular outcome\footnote{The pseudo-causal box adds the idealisation that this measurement can be instantaneous, this is merely a simplification that allows to preserve the input/output time stamps of the original causal box protocol.}) to check if a protocol satisfies AO during a run of the protocol. This procedure is moreover non-disturbing if AO is satisfied. 

Intuitively, one would expect that if a causal box protocol satisfies AO, then only checking the number of messages on the input wires, but not the number of messages on the output wires (or vice versa), should not disturb the protocol either. The below lemma shows that this is true. We note that the physicality of the projectors $\C{P}^{I/O}_k$ is an important ingredient here as an analogous result does not hold for projectors which are not subnormalised pseudo-causal boxes (as seen in \cref{ex:sigproj}).

\begin{restatable}{lemma}{checkio}\label{lemma:checkio}
    If a causal box protocol satisfies AO, then for all $x_k$
    \begin{equation}
        \CM{\C{C}}{(\C{P}^O_k \circ \M \circ \C{P}^I_k)} = \CM{\C{C}}{(\M \circ \C{P}^I_k)} = \CM{\C{C}}{(\C{P}^O_k \circ \M)}= \CM{\C{C}}{\M}
    \end{equation}
\end{restatable}
\begin{proof}
This is a consequence of \cref{coro:normalise} with the sequential composition of $\C{C} \circ \bigotimes_{k=0}^{N+1} (\C{P}^O_k \circ \M)$ playing the role of the causal box and $\bigotimes_{k=0}^{N+1} \C{P}^I_k$ playing the role of the projector (and analogously, when switching $I$ and $O$).
\end{proof}
As we have seen in \cref{ex:trivvio}, another way we can trivially violate causal inequalities (using a causally ordered protocol) is by allowing agents to send a non-vacuum message before they receive a non-vacuum message. Indeed, in the higher-order processes framework, one assumes that there is a local causal ordering inside each agent's lab, with their output coming after their input (both of these being non-vacuum by definition), even if the overall causal order between different agents' operations is not specified. 

To unambiguously formalise the condition for causal box protocols, we will consider situations where if an agent $A_k$ receives some (non-vacuum) message at an input location $t\in \C{T}_k^I$ they will output a (non-vacuum) message at a unique later output location $t'\in \C{T}_k^O$, $t'\succ t$, such that no two input locations associate to the same output location. This correspondence between input and output locations is a restriction, although it is physically well motivated and satisfied in particular in all scenarios where each agent's device has a fixed processing time translating inputs at time $t$ to outputs at $t+\Delta$ for $\Delta>0$ \cite{VilasiniRennerPRA, VilasiniRennerPRL}. The relaxation of this restriction is discussed later in \cref{sec:relaxing}, and entails open questions.

We wish to ultimately formalise the above intuition in a manner similar to AO, where it only needs to hold effectively during an actual run of the protocol. As a first step, however, we present a simpler formulation that refers only to the agent operations.

\begin{figure}
    \centering
    \includegraphics[width=0.5\linewidth]{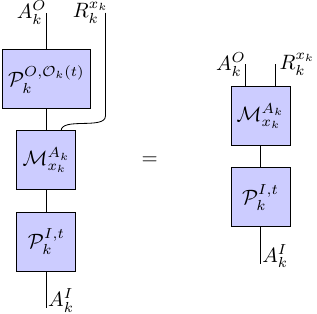}
    \caption{Graphical depiction of the Agent Local Order (ALO) condition of \cref{def:alo}. 
This tells us that whenever $\M$ receives a single message at input position $t$ (and no other non-vacuum inputs in $\C{T}_k^I$) as ensured by the input projector $\C{P}^{I, t}_k$, it is guaranteed to output a single non-vacuum message at a unique later time $\C{O}_k(t)\succ t$ (and no other non-vacuum outputs in $\C{T}_k^O$) as ensured by the output projector $ \C{P}^{O, \C{O}_k(t)}_k$.}
    \label{fig:alocond}
\end{figure}

\begin{restatable}[Agent local order condition for a causal box protocol]{defi}{alo}\label{def:alo}
    For an agent $A_k$ with non-trivial input and output spaces, we say that an agent operation $\M$ satisfies Agent Local Order (ALO) if there exists an injective map $\C{O}_k: \C{T}_k^I \rightarrow \C{T}_k^O$  with $\C{O}_k(t) \succ t$ for all $t \in \C{T}_k^I$ such that for all $t\in \C{T}$
    \begin{equation}\label{eq:alo}
        \C{P}^{O, \C{O}_k(t)}_k \circ \M \circ \C{P}^{I, t}_k = \M \circ \C{P}^{I, t}_k.
    \end{equation}
    We say that the agent $A_k$ with non-trivial input and output spaces satisfies ALO if all their operations satisfy ALO with the same map $\C{O}_k$. We say that a causal box protocol satisfies ALO if all its agents with non-trivial input and output systems satisfy ALO.
\end{restatable}

We also graphically illustrate this definition in \cref{fig:alocond}. 

We now generalise this definition to a condition called Local Order (LO) that involves the whole causal box protocol (more analogous to the Acting Once or AO condition) as opposed to just the agent operations as in Agent Local Order or ALO (see also the graphical depiction of LO in \cref{fig:locond}).

\begin{defi}[Local order condition for a causal box protocol]\label{def:lo}
Given a causal box protocol $\mathfrak{P}$, we say that $\mathfrak{P}$ satisfies local order (LO) if for each agent $A_k$ with non-trivial input and output spaces, there exists an injective map $\C{O}_k: \C{T}_k^I\rightarrow \C{T}_k^O$ with $\C{O}_k(t) \succ t$ for all $t \in \C{T}_k^I$ such that for all $x_k$ 
\begin{equation}
    \CM{\C{C}}{\C{P}_{eff,k}(\M)} = \CM{\C{C}}{(\M \circ \C{P}^I_k)}
\end{equation}
where
\begin{equation}
\label{eq: Peff}
    \C{P}_{eff,k}(\M)(\rho) := \sum_{t,t' \in \C{T}^I_k} P^{O, \C{O}_k(t)}_k \M(P^{I, t}_k \rho P^{I, t'}_k) P^{O, \C{O}_k(t)}
\end{equation}
if the input and output spaces of $A_k$ are non-trivial. 

If the output space is trivial we instead have $\C{P}_{eff, k}(\M) = \M\circ \C{P}_k^I$ and if the input space is trivial, then $\C{P}_{eff, k}(\M) = \M$.
\end{defi}

This effective definition is indeed more general, i.e., a weaker assumption than ALO as shown in the following lemma. 

\begin{restatable}[ALO implies LO]{lemma}{aloimpliesnlo}\label{lemma:aloimpliesnlo}
    If $\M$ satisfies ALO, then
    \begin{equation}
        \M \circ \C{P}^I_k= \C{P}_{eff,k} (\M).
    \end{equation}
    This implies that if $\PBP$ is a causal box protocol satisfying AO and ALO, then $\mathfrak{P}$ is a process box protocol (satisfies AO and LO).
\end{restatable}

The object $\C{P}_{eff,k}$ can be rewritten in terms of an encoding and decoding operation which are subnormalised pseudo-causal boxes connected by a memory which is a causal box, giving it a well-defined operational interpretation (it can be viewed as a higher-order operation acting on the agents' operation as shown in \cref{fig:locond}). Concretely, it can be seen as applying the subnormalised pseudo-causal box $\C{P}^I_{LO,k}: \C{L}(\FI) \rightarrow \C{L}(\FI \otimes \F{\alpha}{\C{T}^I_k})$ which is defined via a single Kraus operator
\begin{equation}
    \sum_{t \in \C{T}^I_k} P^{I, t}_k  \otimes \ket{0,t}^\alpha
\end{equation}
followed by $\M \otimes \sum_{t, t' \in \C{T}^I_k} \ket{\C{O}_k(t)} \bra{t} \cdot \ket{t'}^\alpha\bra{\C{O}_k(t')}^\alpha$ and finally the subnormalised pseudo-causal box $\C{P}^O_{LO,k}: \C{L}(\FO  \otimes \F{\alpha}{\C{T}^I_k}) \rightarrow \C{L}(\FO)$ which is defined via a single Kraus operator
\begin{equation}
    \sum_{t \in \C{T}^O_k} P^{O, t}_k  \otimes \bra{0, t}^{\alpha}.
\end{equation}
Summarising this in one equation, we have
\begin{equation}\label{eq:lodecomp}
    \C{P}_{eff,k}(\M) = \C{P}^I_{LO,k} \circ (\M \otimes \C{O}_k) \circ \C{P}^I_{LO,k}.
\end{equation}
where we abused notation to write $\C{O}_k = \sum_{t, t' \in \C{T}^I_k} \ket{\C{O}_k(t)} \bra{t} \cdot \ket{t'}^\alpha\bra{\C{O}_k(t')}^\alpha$.

That these maps are indeed subnormalised pseudo-causal boxes is shown in \cref{lemma:lopseudocb} in \cref{app:cb}. 

\begin{figure}
    \centering
    \includegraphics[width=\linewidth]{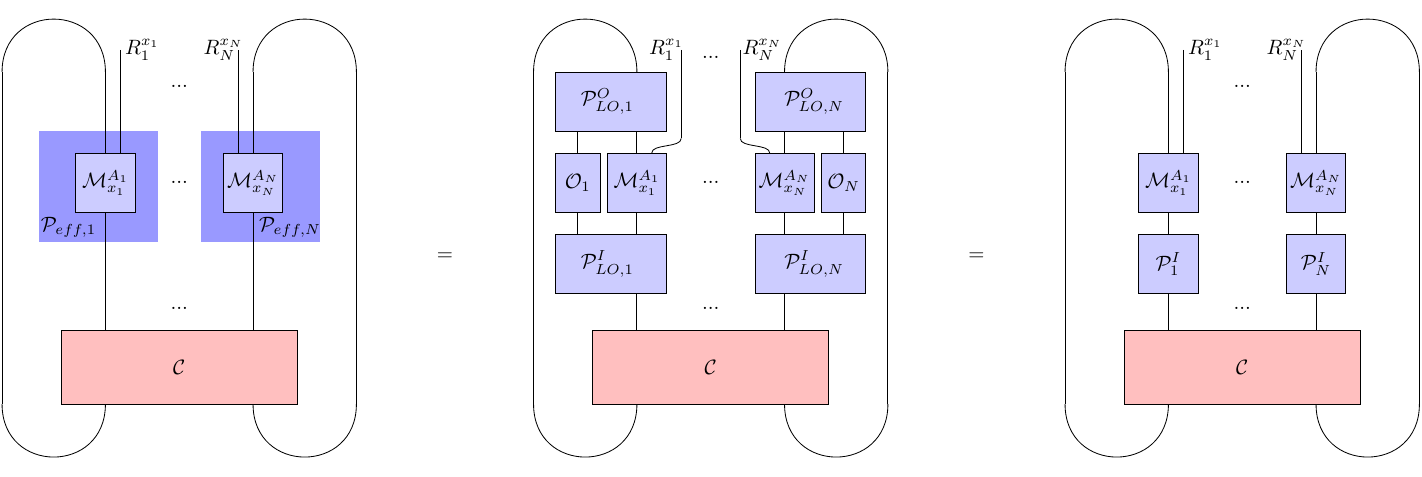}
    \caption{Graphical depiction of the Local Order (LO) condition of \cref{def:lo}. The projectors $\C{P}_{eff,k}$ act on the agents' operations. This operation can be decomposed into a comb-like structure involving two subnormalised pseudo-causal boxes ($\C{P}^{I/O}_{LO,k}$) and a causal box $\C{O}_k$ which simply maps messages at time $t$ to messages at time $\C{O}_k(t)$. It can thus be viewed as a higher-order transformation acting on the agents' operation, as shown in the middle figure. The condition ensures that during the protocol, the agents' operations do not send out a non-vacuum message before receiving a non-vacuum message from the multi-partite causal box $\C{C}$. } 
    \label{fig:locond}
\end{figure}

\begin{defi}[Process box protocols]
    We call a multipartite causal box protocol $\mathfrak{P}$ a process box protocol if $\mathfrak{P}$ satisfies AO and LO.
\end{defi}

Combining \cref{def:lo} and \cref{lemma:checkio}, we immediately find that a process box protocol satisfies
\begin{equation}
    \CM{\C{C}}{\M} = \CM{\C{C}}{\C{P}_{eff,k} (\M)}.
\end{equation}

Here we have defined process box protocols in terms of LO rather than ALO as the former is weaker, and thus allows for a larger class of protocols to qualify as process box protocols. However, in the next section we will show that any protocol satisfying LO is behaviourally equivalent to one satisfying ALO.

{\bf \noindent Effective Choi} A multipartite causal box $\C{C}$ acts on an infinite dimensional Fock space and its Choi representation is therefore generically given by a sesquilinear form \cite{Portmann_2017}. However, when we know that this causal box $\C{C}$ is in fact part of a process box protocol $\mathfrak{P}$, the AO and LO restrictions enable us to write an effective, finite-dimensional Choi representation of $\C{C}$ which can be equivalently used whenever we consider its behaviour under composition with the agents' allowed operations in this protocol. For this, we define the effective Hilbert space of the joint input and output of such a $\C{C}$ which is identified via the AO and LO restrictions
\begin{equation}
    \C{H}_{eff} = \bigotimes_k \bigoplus_{t \in \C{T}} (\C{H}^{A^I_k, t} \otimes \C{H}^{A^O_k, \C{O}_k(t)}) \subsetneq \bigotimes_k \C{H}^{A^I_k, \C{T}} \otimes \C{H}^{A^O_k, \C{T}} \subsetneq \bigotimes_k \FIall \otimes \FOall
\end{equation}
where $\C{H}^{A^O_k, \succ t} = \bigoplus_{t' \succ  t} \C{H}^{A^O_k,t'}$.

For a process box protocol, we can then write (notice that in the below all three expressions are density matrices in $\C{L}(\bigotimes_{k=0}^{N+1} \C{H}^{R^{x_k}_k, T})$ and the equation is thus well-typed) 
     \begin{equation}\label{eq:ceff}
        \CM{\C{C}}{\M} = \CM{\C{C}}{\C{P}_{eff,k}(\M)} = C_{eff} * \bigotimes_{k=0}^{N+1} M^{A_k}_{x_k}|_{\C{L}(\C{H}^{A^I_k, \C{T}})}
    \end{equation}
    where $M^{A_k}_{x_k}|_{\C{L}(\C{H}^{A^I_k, \C{T}})}$ is the Choi-Jamiołkowski matrix of the restriction to the one-message space of $\M$ (equivalently, of $\M \circ \C{P}^I_k$) and we call $C_{eff} \in \C{L}(\C{H}_{eff})$ the effective Choi-Jamiołkowski of $\C{C}$ and it is defined as the Choi-Jamiołkowski matrix of the effective causal box 
    \begin{equation}\label{eq:ceffrho}
        \C{C}_{eff}(\cdot) :=  \sum_{\{t_k, t'_k \in \C{T}^I_k\}_k} \bigotimes_{k=0}^{N+1} P^{I, t_k}_k \C{C}(\bigotimes_{k=0}^{N+1}P^{O, \C{O}_k(t_k)}_k \cdot \bigotimes_{k=0}^{N+1}P^{O, \C{O}_k(t_k')}) \bigotimes_{k=0}^{N+1}P^{I, t_k'}_k.
    \end{equation}
    Note that the effective causal box can be viewed as a map on a finite-dimensional space and hence its Choi-Jamiołkowski matrix is well-defined. \Cref{eq:ceff} can be verified straightforwardly by writing $\CM{\C{C}}{\C{P}_{eff,k}(\M)}$ with the alternative representation for loop composition of \cref{lemma:looprep} and the equivalence of loop composition and the link product in the finite-dimensional setting \cite{salzger2024mappingindefinitecausalorder}.

The conditions AO and LO above are both properties of the CB protocol as a whole (involving both the CB and the agents' operations) and are effective, in the sense, that they only care about what happens when we compose the protocol. This means that agents can have operations that send multiple messages as long as this never occurs in an actual run of the protocol. It also means that if we take a process box protocol and add additional allowed agent operations, the result may no longer be a process box protocol (but still a causal box protocol), as motivated in the example below. This is the reason why we introduced the notion of allowed operations of agents, instead of simply allowing all possible operations.

\begin{restatable}{example}{norestriction}\label{ex:norestriction}
    We consider a causal box protocol with a single agent $A$ with input time stamps $t=2,4$ and output time stamps $t=3,5$. The message spaces on both the input and output side correspond to the Fock space of a qubit. The causal box acts as (we will restrict the definition to the case where there is at most one message on the wire at every time for simplicity)
    \begin{equation}
        V \ket{i, t=3}^{A^O} \ket{j, t=5}^{A^O} = \begin{cases} \ket{0, t=2}^{A^I} \ket{\Omega, t=4}^{A^I} \ket{j}^{\alpha_F}, & i =0\\
        \ket{0, t=2}^{A^I} \ket{1, t=4}^{A^I} \ket{j}^{\alpha_F}, & i =1
        \end{cases}
    \end{equation}
    where $\alpha_F$ is an ancilla. From this definition we see that the causal box always sends a message to the agent at $t=2$ and sends a second message at $t=4$ iff the agent sent back the state $\ket{1}$ at $t=3$. This means that for $\C{C}(\cdot):=V (\cdot) V^\dagger$,  $\mathfrak{P} := \{\C{C}, \{\C{M}^A_x\}_x\}$ can satisfy AO only if for all $x$, the state $\C{M}^{A}_x (\ket{0, t=2}\bra{0, t=2})$ has no overlap with $\ket{1, t=3}$.
\end{restatable}

\subsection{Characterisation results}\label{sec:characterisation}

Our main goal is to relate general process box protocols to QC-QCs. However, there are several aspects of process box protocols which resist an immediate translation into the QC-QC framework. As a result, we might expect process box protocols to be more general in some sense (note that \cite{salzger2024mappingindefinitecausalorder} shows that process box protocols are at least as general as QC-QCs). 

The goal for this section is as follows: we wish to simplify process box protocols as far as possible so that they look as close as possible to the naive picture of a QC-QC in spacetime. For this purpose, we will first discuss in an informal fashion which features of process box protocols might be hard to translate to the QC-QC picture and how they might be represented in a different manner to make this translation. We will then turn this informal discussion into two rigorous statements (\cref{lemma:relabeling,thm:simplifying}) showing that for every process box protocol there exists a behaviourally equivalent process box protocol with properties very close to what we would expect from a naive embedding of a QC-QC into spacetime. These will also serve as useful characterisation results for the process box framework, showing that the structure of general process boxes can be notably simplified while preserving their relevant operational behaviour.

\begin{itemize}
\item Process box protocols incorporate a background space time explicitly, while QC-QCs do not, however, we could imagine a total time ordering on the operations where the internal operation $V_n$ is applied during the $n$-th time step. Even with this idea in mind, the process box protocol setting appears more general as its background spacetime can be a partial order.
\item In the QC-QC framework we can imagine that in each step exactly one agent acts, while in a process box protocol, there could be spacetime positions where no agent does anything (i.e., only vacuum is sent or received) or multiple agents act. This could also be controlled, even dynamically. 
\item Agents with trivial input are only constrained by AO but not LO. It then appears that they can freely choose when to act while agents in the QC-QC framework do not have this freedom.
\item In QC-QCs, we could think of the agents as outputting a message immediately after receiving a message, which corresponds to $\C{O}_k(t)= t+1$ for all agents, while in a process box protocol an agent is able to ``wait'' (corresponding to $\C{O}_k(t) > t+1$). Additionally, this waiting time could itself depend on $t$ and be non-uniform for different input positions.
\item In a process box protocol, an agent's operation can depend on the spacetime position of the incoming message while in a QC-QC the agent applies the same operation regardless of the ``time stamp''. 
\item For process box protocols, we codified the assumptions AO and LO which characterise processes in an effective way, i.e., they only have to hold when we actually compose the multi-partite causal box with the agent operations. On the other hand, in the QC-QC/process matrix framework these assumptions are baked in. Intuitively, it seems that only the 1-message space should matter but this does not necessarily mean that there exists a QC-QC which is well-defined and behaviourally equivalent to the process box protocol. 
\end{itemize}

We begin by discussing the background spacetime. In the causal box framework, the spacetime imposes restrictions on what is an allowed operation, as these need to satisfy the causality condition. For process box protocols, the spacetime imposes an additional constraint in the form of local order. These constraints depend on the order relations between spacetime points and not the arbitrary labels we assign to them. Hence, we can freely relabel our spacetime and the operational predictions should remain the same. This is also what we would expect from special relativity, with the labels corresponding to the frame-dependent spacetime coordinates and the order relations corresponding to the frame-independent lightcone structure.

We can also introduce new order relations even and still obtain a behaviourally equivalent process box protocol. This is because the \textit{lack of} an ordering relation between two spacetime points, i.e., two spacetime points being spacelike separated, is what imposes restrictions on the causal box/process box protocol. 

Indeed, we can even go a step further and make the relabelling dependent on the system as long as some conditions hold. We formalise this in the lemma below.

\begin{restatable}[Invariance of causal boxes under relabelling of spacetime locations]{lemma}{relabeling}
\label{lemma:relabeling}
Let $\PBP$ be an $N$-partite causal box protocol with locations in $\C{T}$. Let $\C{T}'$ be another set of locations, and consider a set of relabelling maps $\{\C{R}_S: \C{T} \rightarrow \C{T}'|S=I_k, O_k, k=1...,N\}$, such that for all $t_1, t_2 \in \C{T}$ with $t_1 \prec t_2$, the following conditions hold
\begin{gather}\label{eq:relabelcondition}
\begin{aligned}
    \C{R}_{I_k}(t_1) &\prec \C{R}_{I_k}(t_2), \C{R}_{O_k}(t_2) \\
    \C{R}_{O_k}(t_1) &\prec \C{R}_{O_k}(t_2), \C{R}_{I_l}(t_2) 
\end{aligned}
\end{gather}
for all $k, l$. We then define a new, relabelled $N$-partite causal box protocol $\PBPp$, as follows
\begin{equation}
\label{eq: relabelledCB}
    \C{C}' \coloneqq \C{R} \circ \C{C} \circ \C{R}^{-1}
\end{equation}
and for all agents $A_k$ and setting choices $x_k$,
\begin{equation}
\label{eq: relabelledlocalop}
\M \mapsto \Mp \coloneqq \C{R} \circ \M \circ \C{R}^{-1}
\end{equation}
where $\C{R}(\ket{\psi, t}^{S}) = \ket{\psi, \C{R}_S(t)}^{S}, \C{R}(\ket{\psi, t_{P/F}}^{P/F}) = \ket{\psi, t'_{P/F}}^{P/F}$ and $\C{R}(\ket{\psi, T}^{R_k}) = \ket{\psi, T'}^{R_k}$ (and linearly extended to mixed states) such that $t'_P \prec t' \prec t'_F \prec T'$ for all $t'\in \C{T}'$ for single messages and $\C{R}$ acts on multi-message states by acting on each message separately. Then, $\mathfrak{P}$ and $\mathfrak{P}'$ are behaviourally equivalent. Further, if $\mathfrak{P}$ is a process box protocol, then so is $\mathfrak{P}'$.
\end{restatable}

We state two weaker forms of this lemma as corollaries.
\begin{restatable}[Agent-independent relabelling]{coro}{relabcoro1}
 Let $\PBP$ be an $N$-partite causal box protocol with locations in $\C{T}$, and $\mathfrak{P}'$ be another such protocol associated with locations in $\C{T}'$ obtained from $\mathfrak{P}$ using \cref{eq: relabelledCB} and \cref{eq: relabelledlocalop} but where the relabelling maps $\C{R}_S$ do not depend on the system $S$. That is there is a single relabelling map $\C{R}_{all}: \C{T}\rightarrow \C{T}'$ that is order-preserving, i.e., $t_1\prec t_2\Rightarrow \C{R}_{all}(t_1)\prec \C{R}_{all}(t_2)$ and $\C{R}_S = \C{R}_{all}$ for all systems $S$. Then $\mathfrak{P}$ and $\mathfrak{P}'$ are behaviourally equivalent, and if the former is a process box protocol, so is the latter.
\end{restatable}

\begin{restatable}[A relabelling lemma with stronger assumptions]{coro}{relabcoro2}
 Let $\PBP$ be an $N$-partite causal box protocol with locations in $\C{T}$, and $\mathfrak{P}'$ be another such protocol associated with locations in $\C{T}'$ obtained from $\mathfrak{P}$ using \cref{eq: relabelledCB} and \cref{eq: relabelledlocalop} for a set of relabellings $\{\C{R}_S: \C{T} \rightarrow \C{T}'|S=I_k, O_k, k=1,...,N\}$ such that for all $t_1, t_2 \in \C{T}$ with $t_1 \prec t_2$, we have that $\C{R}_S(t_1)\prec \C{R}_{S'}(t_2)$ for all systems $S,S'$. Then $\mathfrak{P}$ and $\mathfrak{P}'$ are behaviourally equivalent, and if the former is a process box protocol, so is the latter.
\end{restatable}

We will use this lemma to argue that w.l.o.g. we can assume that the spacetime is totally ordered, bringing us a bit closer to QC-QCs. It also allows us to deal with situations where multiple agents receive or send messages at the same time. This is somewhat similar to QC-PARs (QC-QCs where the operations happen in parallel instead of in sequence, see \cite{Wechs_2021} for more details) but is also more general, as here we could also have superposition between parallel and sequential operations. We give a simple example below.

\begin{figure}
    \centering
    \includegraphics[width=0.8\linewidth]{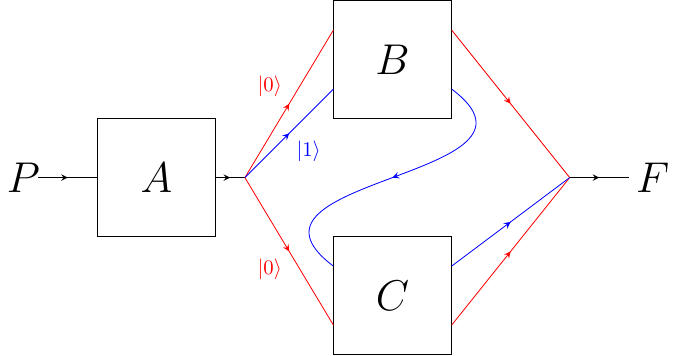}
    \caption{The protocol sends a state $\psi$ to Alice. Depending on her output, either Bob and Charlie act in parallel (red path) or sequentially (blue path).}
    \label{fig:dynamicalparallel}
\end{figure}

\begin{restatable}{example}{dynamicalpar}\label{ex:dynamicalpar}
    We consider a 3-partite process box protocol with trivial global past. The 3-partite causal box is given by the isometry
    \begin{equation}
        V \ket{i, t=2}^{A^O} \ket{j, r}^{B^O} \ket{k, s}^{C^O} = \begin{cases}
            \ket{\psi, t=1}^{A^I} \ket{i, t=3}^{B^I} \ket{i, t=3}^{C^I} \ket{jkrs, t=5}^F, &i=0 \\
            \ket{\psi, t=1}^{A^I} \ket{i, t=3}^{B^I} \ket{j, r+1}^{C^I} \ket{\Omega krs, t=5}^F, &i=1
        \end{cases}
    \end{equation}
    for arbitrary time stamps $r, s \in \{ 2,4\}$ and $\ket{\psi}^{A^I}$ is some particular state and the allowed agent operations are all those which map a one-message state at $t$ to a one-message state at $t+1$. The protocol is depicted schematically in \cref{fig:dynamicalparallel}.
\end{restatable}

In the example, we see that Alice can decide whether Bob and Charlie act in parallel or sequentially. This is not something that explicitly appears in the QC-QC framework. However, it is still possible to cast this protocol so that it looks completely sequential, simply by changing Charlie's (and the global future's) time stamps
\begin{equation}
        V \ket{i, t=2}^{A^O} \ket{j, r}^{B^O} \ket{k, s}^{C^O} = \begin{cases}
            \ket{\psi, t=1}^{A^I} \ket{i, t=3}^{B^I} \ket{i, t=5}^{C^I} \ket{jkr(s-2), t=9}^F, &i=0 \\
            \ket{\psi, t=1}^{A^I} \ket{i, t=3}^{B^I} \ket{j, r+3}^{C^I} \ket{\Omega kr(s-2), t=9}^F, &i=1
        \end{cases}
\end{equation}
and similarly for Charlie's operations.

We will soon see that we can actually take this a step further and find for every process box protocol a behaviourally equivalent one such that each time stamp is associated to a single agent who may send respectively receive during this time step, by using \cref{lemma:relabeling}. We motivate some of these useful properties further before stating them formally in \cref{thm:simplifying}.

The possibility of delays between the agent's input and the agent's output is somewhat limited by the fact that we fix $\C{O}_k(t)$ for the whole protocol. Additionally, since it is then determined that the agent will send a message at $\C{O}_k(t)$ if they receive a message at $t$, one may expect that it one can transform to a behaviourally equivalent scenario where the agent instead outputs at $t+1$ and the causal box essentially pretends it received the message at $\C{O}_k(t) \succeq t+1$. This will allow to fix $\C{O}_k(t)=t+1$ for all agents $A_k$ and input times $t$.

Finally, let us discuss the question of time-dependent agent operations. While this might seem more general than what we can do with a QC-QC, it turns out that whenever we have time-dependent agent operations we can think of this as the time stamp taking over part of the role of the message and we can always shift this back into the message. We illustrate this with a simple example of a time-dependent operation where the protocol essentially encodes a bit in the time stamp. Imagine an agent whose input message space is 1-dimensional. Their input Fock space is then given by $\C{F}(\mathbb{C} \otimes \C{H}^{\C{T}}) \cong \C{F}(\C{H}^{\C{T}})$ where $\C{T}$ is totally ordered. The causal box can still communicate a bit to this agent by using the parity of the time stamp, i.e., sending a message at an even time $t$ if it wants to communicate the bit $0$ and at an odd time $t$ if it wants to communicate the bit $1$. The agent can then output whatever message they want at $t+1$. For example, if they wanted to implement a NOT gate, they would output $\ket{0, t+1}$ if $t$ is odd and $\ket{1, t+1}$ if $t$ is even. It is quite clear that at least in this case time-dependence does not gain us anything as we can achieve the same thing by giving the agent access to a qubit message space, i.e., replacing $\C{F}(\C{H}^{\C{T}})$ with $\C{F}(\mathbb{C}^2 \otimes \C{H}^{\C{T}})$. Whenever the causal box would send a message $\ket{2t+i}$ with $i=0,1$, it instead sends the message $\ket{i, 2t+i}$ and the agent simply implements the time-independent operation $\ket{i, 2t+i} \mapsto \ket{i\oplus 1, 2t+i+1}$ on it (note that the action on the time stamps is simply to add 1, independently of $t$). In general, we can have a more complex interplay between the message space and the time stamp, but the same intuition holds. 

We formalise all these intuitive ideas in the theorem below which tells us that we can always find a behaviourally equivalent process box protocols that looks more similar to QC-QCs with respect to the points discussed before.

\begin{restatable}[Simplifying the set of process boxes]{thm}{simplifying}
\label{thm:simplifying}
Any process box protocol $\PBP$ on a finite spacetime $\C{T}$ 
is behaviorally equivalent to a process box protocol $\PBPp$ with the following properties:

\begin{enumerate}
    \item \textbf{Total order:} The spacetime is the totally ordered set $\C{T}' = \{1,2,...,M\}$ for some $M \in \mathbb{N}$
    \item \textbf{Non-trivial input:} All agents except for the global past have non-trivial input.
    \item \textbf{Time-independent agent operations:} For every agent $A_k$, the action of each local operation $\Mp$ on the 1-message space can be decomposed as 
    \begin{equation}\label{eq:osr}
    \Mp|_{\C{L}(\C{H}^{A^{\prime I}_k, \C{T}^{\prime I}_k})} = \Mp \circ \C{P}^{\prime I}_k = \sum_i A^i_{k, x_k} \cdot A^{i \dagger}_{k, x_k} \otimes T_{+1} \cdot T^\dagger_{+1}
\end{equation}
for some $ A^i_{k, x_k}: \C{H}^{A^{\prime I}_k} \rightarrow \C{H}^{A^{\prime O}_k } \otimes \C{H}^{R_k^{x_k}, T=M} $ and where $T_{+1}\ket{t} = \ket{t+1}$ acts on time stamps by increasing them by one.\footnote{In \cite{salzger2024mappingindefinitecausalorder}, we gave a similar condition for the Kraus operators of an agent's local operation. The Kraus operators of the local operation $\Mp$ have exactly that form.} Note that we added a prime to the one-message projectors on the input spaces of the new protocols to distinguish it from the ones in the original protocol. In particular, $\C{O}_k(t) = t+1$ for all $k$.
    \item \textbf{1-message output:} The image of the causal box $\C{C}'$ is the tensor product of the 1-message input spaces of the agents, i.e., 
    \begin{equation}
        \C{C}' = \bigotimes_{k=1}^{N+1} \C{P}^{\prime I}_k \circ \C{C}'.
    \end{equation}
    \item \textbf{Disjoint input and output times:} The sets $\C{T}^{\prime I/O}_k$ are disjoint for distinct $k$ and $\C{T}^{\prime I}_k$ contains only even times while $\C{T}^{\prime O}_k$ contains only odd times for each $k$. Thus, for each time step there exists exactly one agent that can send or receive a message.
\end{enumerate}

\end{restatable}
The above theorem is strictly stronger than the conjunction of conjecture 1 and lemma 6 from \cite{salzger2023mscthesis}. Using the wire isomorphisms, we can write \cref{eq:osr} also as
\begin{gather}
\begin{aligned}
    \Mp|_{\C{L}(\C{H}^{A^I_k, \C{T}'})}(\cdot)= \sum_i \sum_{t,t' \in \C{T}^I_k}& (A^i_{k, x_k} \cdot A^{i \dagger}_{k, x_k} \otimes \ket{t+1}\bra{t} \cdot \ket{t'}\bra{t'+1})\\
    &\otimes \ket{\Omega, \C{T}\backslash \{t+1\}} \bra{\Omega, \C{T}\backslash \{t\}} \cdot \ket{\Omega, \C{T}\backslash \{t'\}} \bra{\Omega, \C{T}\backslash \{t'+1\}}.
\end{aligned}
\end{gather}

The simplified causal box from the above theorem can always be written in terms of a sequence representation $V_1,...,V_M$ with $M$ isometries, one for each time step, where $V_m$ maps inputs at $t=m$ to outputs at $t=m+1$.

It is worth stating which issues we have yet to resolve. As we can see from property 5 in the theorem above, it is still possible that no agent acts during a particular time step which is not possible in the QC-QC framework. Another question that is still unanswered is about the ``control'' part of quantum circuits with quantum control of causal order. From the description in \cref{thm:simplifying} it is not clear why such operations are necessarily controlled. Both of these questions will be the subject of the next section.

\section{Mapping process boxes to QC-QCs}\label{sec:mapping}
\subsection{Examples and properties}

Before getting to the general result, we will look at an example of a process box protocol and its mapping to a QC-QC. The quantum switch has been implemented in a number of experiments using long coherence times \cite{Goswami_2018,Goswami_2020}. In the below, we model this in a minimal discrete fashion as a process box protocol.

\begin{restatable}[Long-coherence time quantum switch]{example}{coherence}\label{ex:coherence}
    We consider a process box protocol with two agents, which we will refer to as $A_0$ and $A_1$, and a global past (consisting of a target and a control) and future. We define the causal box via its sequence representation, without explicitly writing the action of the causal box for certain cases which can never occur in a run of the protocol. The isometry $V_{m+1}$ acts on states at $t=2m+1$ and sends out states at $t=2m+2$. We will leave the time stamps implicit below.
    \begin{gather}
    \begin{aligned}
        V_1 \ket{\psi}^{P_T} \ket{i}^{P_C} &= a_1 \ket{\psi}^{A^I_i}\ket{\Omega}^{A^I_{i \oplus 1}} \ket{\Omega}^{\alpha} \ket{i}^{\alpha_C} \ket{1}^{\alpha_S} + b_1 \ket{\Omega}^{A^I_0} \ket{\Omega}^{A^I_1}\ket{\psi}^{\alpha} \ket{i}^{\alpha_C} \ket{0}^{\alpha_S} \\
        V_{m+1} \ket{\psi_0}^{A_0^O} \ket{\psi_1}^{A_1^O} \ket{\psi}^{\alpha} \ket{i}^{\alpha_C} \ket{j}^{\alpha_S} &= \begin{cases} a_{m+1} \ket{\psi}^{A^I_i} \ket{\Omega}^\alpha \ket{i}^{\alpha_C} \ket{1}^{\alpha_S} \\+ b_{m+1} \ket{\Omega}^{A^I_i} \ket{\psi}^{\alpha} \ket{i}^{\alpha_C} \ket{0}^{\alpha_S}, &j = 0, \psi_0 = \psi_1 = \Omega \\
        \ket{\Omega}^{A^I_i}\ket{\psi_{i}}^{A^I_{i\oplus1}} \ket{\psi_{i\oplus1}}^{\alpha} \ket{i}^{\alpha_C} \ket{1}^{\alpha_S}, &j=1, \psi= \Omega \\
        \end{cases} \\
        V_{M+1} \ket{\Omega}^{A^O_0} \ket{\Omega}^{A^O_1} \ket{\psi}^{\alpha} \ket{i}^{\alpha_C} \ket{j}^{\alpha_S} &= \ket{\psi}^{F} \ket{ij}^{\alpha_F}
    \end{aligned}
    \end{gather}
    where $|a_m|^2 + |b_m|^2 = 1$ for all $m$ and $b_{M-1} =0$ (this guarantees that if neither agent has received a message up until the time step $t=2(M-1)$, the first agent receives a message during that time step). We take the agent operations to be of the form from property 3 of \cref{thm:simplifying}, i.e., time-independent (note that this implies, together with the action of $V_{m+1}$ on a non-vacuum state, that if one agent receives a message at $t$, the other one receives a message at $t+2$ and sends an output at $t+3$. Hence, both agents output at or before $t=2M+1$. In particular, $V_{M+1}$ only ever receives vacuum from the agents). This process box protocol is essentially the quantum switch. If we plug in $\ket{0}$ or $\ket{1}$ for the control at $P_2$ we can even think of it as a quantum circuit with fixed causal order (a QC-FO). This latter case shows that we can have time non-localisation without anything that could be called indefinite causal order.
\end{restatable}

While \cref{thm:simplifying} gives us a form for process box protocols which look very similar to QC-QCs, a key aspect that is still missing is the control. This is not just a technical question, but also one of physical significance. Can every process in a classical spacetime be viewed as a controlled one and if so, why? 

Using the sequence representation, we can show that we can always identify subsystems of ancillas that we can interpret as a control system, given process box protocols of the form of \cref{thm:simplifying}. 

\begin{restatable}[Adding the control]{lemma}{addcontrol}
\label{lemma:addcontrol}
Let $\PBP$ be a process box protocol that satisfies the conditions set out in \cref{thm:simplifying}. Denote with $I(m)$ the unique agent which can receive a message at $t=2m+2$ and with $O(m)$ the unique agent which can send a message at $t=2m+1$. Then there exists a sequence representation $\bar{V}_1,...,\bar{V}_{M+1}$ for some $M \in \mathbb{N}$ of $\C{C}$ such that for $m \leq M$,
\begin{equation}\label{eq:addcontrol0}
     \bar{V}_{m+1} = \sum_{\C{K}: |\C{K}| \leq m} (\bar{V}^{\rightarrow I(m)}_{\C{K},{m+1}} \otimes \ket{\C{K} \cup I(m)} \bra{\C{K}} +\bar{V}^{\rightarrow \Omega}_{\C{K},{m+1}} \otimes \ket{\C{K}} \bra{\C{K}})
\end{equation}
where the second tensor factor of each term is a control system which is composed between the different $\bar{V}_{m+1}$, $\C{K} \subseteq \C{N}$ and
\begin{equation}
    \bar{V}_{M+1} = \sum_{\C{K}} \bar{V}^{\rightarrow F}_{\C{K},M} \otimes \ket{\C{K} \cup F} \bra{\C{K}}
\end{equation}
where 
    \begin{gather}
    \begin{aligned}
    \bar{V}^{\rightarrow I(m)}_{\C{K},{m+1}}: &\F{A^O_{O(m)}}{t=2m+1} \otimes \C{H}^{\alpha_m} \rightarrow \C{H}^{A^I_{I(m)}, t=2m+2} \otimes \C{H}^{\alpha_{m+1}}\\
    \bar{V}^{\rightarrow \Omega}_{\C{K},{m+1}}: &\F{A^O_{O(m)}}{t=2m+1} \otimes \C{H}^{\alpha_m} \rightarrow \ket{\Omega}^{A^I_{I(m)}, t=2m+2} \otimes \C{H}^{\alpha_{m+1}}.
    \end{aligned}
    \end{gather}
\end{restatable}

We will discuss how we can interpret this in \cref{sec:conclusion}.

\subsection{General mapping results}

In this section, we present our main result, namely that process box protocols are behaviourally equivalent to QC-QCs. We first give a formal definition of what this means. For this purpose, we will first put QC-QCs on the same footing as causal/process box protocols.

\begin{restatable}[QC-QC protocols]{defi}{qcqcp}\label{def:qcqcp}
    A QC-QC protocol $\QP$
    is specified by a QC-QC $\C{Q}$ together with a set of CPTP maps $\{\M: \C{L}(\C{H}^{A^I_k}) \rightarrow \C{L}(\C{H}^{A^O_k} \otimes \C{H}^{R_k^{x_k}})\}_{x_k}$ for each agent $\{A_k\}_{k=0}^{N+1}$ such that the input space of the global past $\C{H}^{A^I_0}$ and the output space of the global future $\C{H}^{A^O_{N+1}}$ are trivial.
\end{restatable}

The supermap view of a QC-QC specifies a default QC-QC protocol whereby for each agent, all quantum operations between their input and output spaces are allowed. The above definition also permits protocols involving a QC-QC and some subset of allowed operations for each agent, i.e., the same QC-QC can be part of different QC-QC protocols depending on the set of agent operations being considered. This shares an analogous structure to causal and process box protocols defined before. 

We can then define the notion of behavioural equivalence between causal box protocols and QC-QC protocols.

\begin{restatable}[Behavioural equivalence of QC-QC and causal box protocols]{defi}{qcqcsim}\label{def:qcqcsim}
    Consider an $N$-partite causal box protocol, $\PBP$ and an $N$-partite QC-QC protocol $\QPp$. We say that $\mathfrak{Q}$ is behaviourally equivalent to $\mathfrak{P}$ if  for all $k\in\{0,...,N+1\}$ and for each agent operation $\M$, there exists an agent operation $\Mp$, and vice versa, such that 
\begin{equation}
    \CM{\C{C}}{\M} = \CM{\C{Q}}{\Mp}
\end{equation}
\end{restatable}

Note that this definition is consistent with our definition of behavioural equivalence of causal box protocols. Given behaviourally equivalent causal box protocols $\mathfrak{P}, \mathfrak{P}'$ and a QC-QC protocol $\mathfrak{Q}$, $\mathfrak{Q}$ is behaviourally equivalent to $\mathfrak{P}$ iff $\mathfrak{Q}$ is behaviourally  equivalent to $\mathfrak{P}'$. This immediately follows from \cref{def:pb_eq,def:qcqcsim}.

We then have our main result below.
\begin{restatable}[All process box protocols map to QC-QCs]{thm}{pbtoqcqc}
\label{theorem:pbtoqcqc}
For every process box protocol $\mathfrak{P}$, there exists a QC-QC protocol which is behaviourally equivalent to $\mathfrak{P}$.
\end{restatable}

In \cite{salzger2024mappingindefinitecausalorder}, we considered the reverse direction of this mapping, i.e., we constructed for any QC-QC a causal box extension, which means that we can map a QC-QC to a causal box and each possible agent operation in the QC-QC picture to a corresponding operation in the causal box picture such that the two scenarios reproduce the same operational behaviour when the AO and LO conditions are satisfied. We note that the AO and LO conditions in \cite{salzger2024mappingindefinitecausalorder} were formulated in a more restricted manner, but it is immediate that they imply that our weaker formulations of these are satisfied. It thus follows from this previous work, that any QC-QC protocol can be mapped to a behaviourally equivalent process box protocol. We thus obtain the following corollary through the combination of the above theorem and the results of this previous work.

\begin{restatable}[Two-way mapping between PBs and QC-QCs]{coro}{bothmappings}
\label{cor:bothmaps}
  For every process box protocol there exists a corresponding QC-QC protocol which is behaviourally equivalent and for every  QC-QC protocol, there exists a corresponding process box protocol that is behaviourally equivalent.
\end{restatable}

Furthermore, the PB protocols in the image of the  mapping of our previous work satisfy simplification properties, which are somewhat different from some of the properties obtained in \cref{thm:simplifying}. For example, the spacetime labels of such a PB protocol, which is $N$-partite, take values in $\C{T}=\{1,2,...,2N+2\}$, and each agent (besides the global past and future) has the same set of input and output time stamps $\C{T}^I=\{2,4,...,2N\}$ and $\C{T}^O=\{3,5,...,2N+1\}$ with the global past (trivial input) agent outputting at $t=1$ and the global future (trivial output) agent receiving input at $t=2N+2$. Moreover, in such PB protocols, it is still ensured that at each time step, only one agent acts on a non-vacuum message (see e.g., Sec 4.5 and Lemma 2 in \cite{salzger2024mappingindefinitecausalorder}) and that agent operations are time-independent (analogous to \cref{thm:simplifying}). This fact gives us another implication of our results (of the current and previous paper) which provides, an alternative simplification of process boxes: for any process box protocol $\PBP$, there exists a behaviourally equivalent process box protocol $\PBPp$ that respects all properties satisfied by the process box protocols induced by the causal box extensions of QC-QCs constructed in \cite{salzger2024mappingindefinitecausalorder}.

Before stating this final simplification, we observe two points and set out some notation. Notice that a multi-partite causal box, $  \C{C}: \C{L}(\bigotimes_{k=0}^{N} \FO  ) \rightarrow \C{L}(\bigotimes_{k=1}^{N+1} \FI)$ can be equivalently viewed as a map $  \C{C}: \C{L}(\bigotimes_{k=0}^{N} \bigotimes_{t\in \C{T}^{O}_k}\F{A^O_k}{t}  ) \rightarrow \C{L}(\bigotimes_{k=1}^{N+1} \bigotimes_{t\in \C{T}^I_k}\F{A^I_k}{t} )$ using the wire isomorphisms $\FIO\cong \bigotimes_{t\in \C{T}^{I/O}_k}\F{A^{I/O}_k}{t}$ for all $k$. Now, we denote by $\C{P}^{I/O,t}$ the projector onto the one message space at a fixed time $t$, over the inputs/outputs of all agents in $\{A_1,...,A_N\}$ (i.e., those with non-trivial in and outputs). Specifically, $ \C{P}^{I/O, t}$ then projects onto the following subspace of $ \C{L}(\bigotimes_{k=1}^{N}\C{F}^{I/O, t}_k)$.
    \begin{equation}
  \text{Span}\{\ket{i,t}^{A^{I/O}_j}\bigotimes_{k=1,k\neq j}^N\ket{\Omega,t}^{A^{I/O}_k}\}_{i\in \{0,...,d-1\},j\in \{1,...,N\}}.
    \end{equation}
     That is, the projector ensures that exactly one such agent receives/sends a non-vacuum message at $t$.
Then, $\C{P}^O = \bigotimes_{t'\in \C{T}^{O}_k} \C{P}^{O, t'}$ is a projector on the input space of $\C{C}$ while $\C{P}^I = \bigotimes_{t\in \C{T}^{I}_k} \C{P}^{I, t}$ is a projector on it's output space.

\begin{restatable}[Another simplification of process box protocols]{coro}{finalsimp}
\label{coro:finalsimp}
Let $\PBP$ be a process box protocol. Then, $\mathfrak{P}$ is behaviorally equivalent to a process box protocol $\PBPp$ satisfying the following properties:
\begin{enumerate}
    \item $\C{T}' = \{1,2,...,2N+3\}$, $A_0$ outputs at $t=1$ with certainty, $A_{N+1}$ receives an input at $t=2N+2$ with certainty, and for all $k\in \{1,...,N\}$, the sets of input and output time stamps are the same  $\C{T}'^I = \{2, 4,...,2N\}$ and $\C{T}'^O = \{3,...,2N+1\}$, with $\C{O}_k(t)=t+1$. That is, the process box $\C{C}'$ of the primed protocol acts as
    \begin{equation}
        \C{C}': \C{L}(\C{F}^{P,t=1}\otimes \bigotimes_{k=1}^{N}\F{A^{\prime O}_k}{\C{T}^{\prime O}}) \rightarrow \C{L}(\bigotimes_{k=1}^{N} \F{A^{\prime I}_k}{ \C{T}^{\prime I}} \otimes \F{F}{t=2N+2})
    \end{equation}
    \item For every agent $A_k$ with non-trivial input, the action of each local operation $\Mp$ on the 1-message space can be decomposed as 
    \begin{equation}
    \Mp|_{\C{L}(\C{H}^{A^I_k, \C{T}'})} = \Mp \circ \C{P}^I_k = \sum_i A^i_{k, x_k} \cdot A^{i \dagger}_{k, x_k} \otimes T_{+1} \cdot T^\dagger_{+1}
\end{equation}
for some $ A^i_{k, x_k}: \C{H}^{A^{\prime I}_k} \rightarrow \C{H}^{A^{\prime O}_k } \otimes \C{H}^{R_k^{x_k}, T=2N+3} $ and where $T_{+1}\ket{t} = \ket{t+1}$ acts on time stamps by increasing them by one.
    \item The causal box $\C{C}'$ maps a state with $m$ messages at $t$ to a state with $m$ messages at $t+1$. In particular,
    \begin{equation}\label{eq:oneatatime}
        \CM{\C{P}^O \circ \C{C}' \circ \C{P}^I}{\Mp} = \CM{\C{C}'}{\Mp}.
    \end{equation}
\end{enumerate}
\end{restatable}

Note that property 2 in the above is the same as property 2 in \cref{thm:simplifying}.

\section{Ruling out non-causal processes from physical principles}\label{sec:ruling}

{\bf Deriving QC-QCs from physical principles in spacetime} Our results combined with the results from \cite{VilasiniRennerPRA, VilasiniRennerPRL} provide a top-down approach to recover QC-QCs by considering all possible information processing protocols in a fixed spacetime and restricting them systematically via physical principles coming from relativistic causality and the set-up assumptions of the process matrix framework. In contrast, the original QC-QC paper \cite{Wechs_2021} used a bottom-up approach, which constructively defines QC-QCs are via generalised quantum circuits leading to a subclass of process where ICO arises through quantum controlled superpositions of orders. The question of whether QC-QCs are the most general  processes physically realisable in a spacetime in a meaningful way (e.g., through table-top experiments) has therefore remained open. Our top-down approach and results resolve this question by formalising ``physically realisable'' and ``meaningful way'' in a rigorous manner, with the former related to relativistic principles in a spacetime and the latter to respecting the set-up assumptions of the process matrix framework, even at a fine-grained level.

More explicitly, in \cite{VilasiniRennerPRA}, it is shown that causal boxes are the most general quantum information processing protocols that can be achieved in a fixed, acyclic spacetime, based on the assumptions that

\begin{itemize}
    \item (A1) relativistic causality is respected, even when considering all physically possible interventions which agents (treated as effectively classical\footnote{In contrast to Wigner's Friend type scenarios where agents are modelled as quantum systems and part of the quantum process being described.}, e.g., our experimental colleagues) can perform in spacetime.
    \item (A2) physical protocols only have a finite number of relevant information processing steps 
\end{itemize}

(A1) and (A2) alone do not single out QC-QCs or even situations describable as valid processes between a given set of agents, as for instance they allow for two agents Alice and Bob to trivially and maximally violate the GYNI causal inequality which is impossible with process matrices (\cref{ex:trivvio}, and also noted in \cite{Vilasini_2020}). In this work, building upon the initial ideas proposed in \cite{Vilasini_2020}, we have argued that one must additionally impose the assumptions AO (Acting Once, \cref{def:ao}) and LO (Local order, \cref{def:lo}), which capture the set-up assumptions (especially that of closed labs) of the process matrix framework at the spacetime level. Together, AO and LO can be understood as imposing a spatiotemporal (and thus more ``fine-grained'') version of the closed labs (CL) assumption of the process matrix framework. Originally, CL is formulated at the coarse-grained level of the process alone (where each agent interacts with the process via a single quantum input and quantum output space which contains no spacetime information). In the fine-grained treatment here, the formulation of CL explicitly accounts for the spatiotemporal structure involved in the inputs and outputs. This is because physical protocols in spacetime can allow agents to exchange messages at superpositions of different spacetime locations and a rigorous formulation of CL at the fine-grained level of such spacetime protocols was lacking prior to the process box framework \cite{Vilasini_2020,salzger2023mscthesis,salzger2024mappingindefinitecausalorder}.

Therefore, our results together with those of \cite{VilasiniRennerPRA, VilasiniRennerPRL} imply the following conclusion.
\begin{center}
The most general quantum information-processing protocols realisable in any acyclic spacetime  respecting the physical assumptions (A1), (A2), AO and LO (or alternatively (A1),(A2) and spatiotemporal closed labs) are behaviourally equivalent to QC-QCs. This in turn provides a physical interpretation for QC-QCs: QC-QCs are precisely those processes which are compatible with relativistic causality in a spacetime and the setup assumptions of the process matrix framework being respected in that spacetime (up to natural assumptions regarding finite information processing steps).    
\end{center}

\begin{figure}
    \centering
    \includegraphics{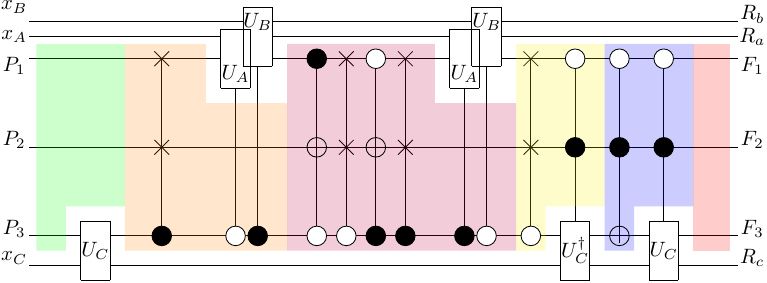}
    \caption{The circuit implementation of the Lugano process as proposed by \cite{Wechs_2023}. In the corresponding causal box protocol, the coloured regions correspond to the causal box while the uncoloured regions correspond to the agent operations. Each coloured region corresponds to one time step or equivalently an isometry in the sequence representation of the causal box (cf. \cref{eq:luganoseq}). }
    \label{fig:lugano}
\end{figure}

{\bf Why the Lugano process is ruled out, comparison to time-delocalised subsystems} We illustrate the implications of our results further with a well-known non-causal (and hence not a QC-QC) process, the Lugano process \cite{Baumeler_2014} which has attracted much research interest. In \cite{Wechs_2023}, it was shown that this process can be realised on so-called time-delocalised subsystems \cite{Oreshkov_2019}. Here we show how there is no mathematical contradiction between their result and ours which rules out such processes, i.e., the two works and results are consistent. However, our work additionally implies that realisation with time-delocalised subsystems (although interesting in its own right and potentially for beyond classical spacetime/agent scenarios) is insufficient to guarantee that there exists a physical realisation through an experiment in a classical background spacetime, respecting the set-up assumptions of the process matrix framework. In the case of non-causal processes like the Lugano process, our results imply that such an experiment in fact cannot exist.

To see the above-mentioned connection, we recall that \cite{Wechs_2023} provided a circuit depiction of their implementation of the Lugano process on time-delocalised subsystems, which is reproduced in \cref{fig:lugano}. Clearly, this respects the assumptions (A1) and (A2) by virtue of being a regular quantum circuit with finitely many operations whose order is consistent with the direction of time i.e., it can be described as a valid causal box protocol, as also explicitly modelled in \cref{ex:lugano}. However, our results imply that it must necessarily violate spatiotemporal closed labs whenever one considers agents' operations for which the process violates a causal inequality, i.e., at least one of our AO (Acting Once) or LO (Local Order) assumptions must be violated in such a case. In fact, in \cref{ex:lugano}, we show that this causal box protocol violates AO (while still respecting LO) whenever it is capable of violating a causal inequality.
\begin{restatable}[Time-delocalised Lugano process \cite{Wechs_2023} as a causal box protocol]{example}{lugano}\label{ex:lugano}
    Each agent's input/output space corresponds to a qubit. Alice's and Bob's input time stamps are $\C{T}^{I}_{A/B} = \{4, 6\}$, their output time stamps are $\C{T}^O_{A/B} = \{5,7\}$, while Charlie's input time stamps are $\C{T}^O_C = \{2, 8, 10\}$ and his output time stamps are $\C{T}^O_C = \{3, 9, 11\}$. For the agent operation, we assume they can be written as
    \begin{gather}
    \begin{aligned}
        \bar{U}_A &= (U_A \otimes \ket{t=5}\bra{t=4}) \otimes (U_A \otimes \ket{t=7}\bra{t=6}) \\
        \bar{U}_B &= (U_B \otimes \ket{t=5}\bra{t=4}) \otimes (U_B \otimes \ket{t=7}\bra{t=6}) \\
        \bar{U}_C &= (U_C \otimes \ket{t=3}\bra{t=2}) \otimes (U^\dagger_C \otimes \ket{t=9}\bra{t=8}) \otimes (U_C \otimes \ket{t=11}\bra{t=10})
    \end{aligned}
    \end{gather}
    where each $U_i, i = A, B, C$ is a unitary on the Fock space which preserves the number of messages (in particular, it maps vacuum to vacuum). The causal box is represented via the following sequence representation where we keep the time stamps implicit and $X^O, Y^O \in \{A^O, B^O\}$ with $X^O \neq Y^O$ 
    \begin{gather}\label{eq:luganoseq}
    \begin{aligned}
        V_1 &= \mathbb{1}^{P_3 \rightarrow C^I} \otimes \mathbb{1}^{P_1 P_2\rightarrow \alpha_1 \alpha_2} \\
        V_2 \ket{c}^{C^O} \ket{\alpha_1}^{\alpha_1} \ket{\alpha_2}^{\alpha_2} &= \begin{cases} \ket{\alpha_1}^{A^I} \ket{\alpha_2}^{\alpha} \ket{0}^{\alpha_C}, &c=0 \\
        \ket{\alpha_2}^{B^I} \ket{\alpha_1}^{\alpha} \ket{1}^{\alpha_C}, &c=1\end{cases} \\
        V_3 \ket{x}^{X^O} \ket{\alpha}^{\alpha} \ket{c}^{\alpha_C} &= \ket{\alpha \oplus x \oplus c}^{Y^I} \ket{x}^{\alpha} \ket{c}^{\alpha_C} \\
        V_4 \ket{y}^{Y^O} \ket{\alpha}^{\alpha} \ket{c}^{\alpha_C}&= \begin{cases} \ket{c}^{C^I} \ket{y}^{\alpha_1} \ket{\alpha}^{\alpha_2} \ket{c}^{\alpha_C}, &c=\alpha \neq y \\
        \ket{\Omega}^{C^I} \ket{y}^{\alpha_1} \ket{\alpha}^{\alpha_2} \ket{c}^{\alpha_C}, &else \end{cases} \\
        V_5 \ket{c'}^{C^O} \ket{y}^{\alpha_1} \ket{\alpha}^{\alpha_2} \ket{c}^{\alpha_C} &= \begin{cases} \ket{c'\oplus 1}^{C^I} \ket{y}^{\alpha_1} \ket{\alpha}^{\alpha_2} \ket{c}^{\alpha_C}, &c=\alpha \neq y \\
        \ket{\Omega}^{C^I} \ket{y}^{\alpha_1} \ket{\alpha}^{\alpha_2} \ket{c}^{\alpha_C}, &else \end{cases} \\
        V_6 \ket{c'}^{C^O} \ket{y}^{\alpha_1} \ket{\alpha}^{\alpha_2} \ket{c}^{\alpha_C} &= \begin{cases} \ket{c'}^{F_3} \ket{\alpha}^{F_1} \ket{y}^{F_2}, &c=\alpha \neq y \\
        \ket{c}^{F_3} \ket{y}^{F_1} \ket{\alpha}^{F_2}, &else \end{cases} \\
    \end{aligned}
    \end{gather}
    Let us now discuss what happens if we check AO by applying our projector $\C{P}^I_C$ to this protocol. From the above we see that we always have non-vacuum on the wire $C^I$ at $t=2$ (given that $P_3$ sends a non-vacuum state). Hence, only the branches where $C^I$ contains the vacuum at all other times survive. This corresponds to thebottom cases of $V_4, V_5$. As the output on $\alpha$ of $V_3$ is equal to $x$, we have that $\alpha = x$, this is the case when $c \neq x$ or $c = y$. Combining this with the fact that the value of $c$ controls which of $A^O$, $B^O$ is $X^O$ respectively $Y^O$, we find that the bottom case occurs iff $a=1$ or $b=0$ where $a$ is the output on $A^O$ and $b$ the output on $B^O$, which in \cref{fig:lugano} corresponds to third wire from the top (note that $A$ and $B$ only ever output a single message). 

    What we find then is the following: Either the allowed operations of $A$ and $B$ allow them to output $a=0$ and $b=1$ and $C$ acts multiple times, violating AO. The causal inequality violation observed in this protocol could then be interpreted as a trivial one analogous to the one discussed in \cref{ex:trivvio} (further justification for this interpretation given at the end of \cref{sec:relaxing} which analyses different types of violations of AO and LO). Alternatively, the allowed operations of $A$ force the output $a=1$ or the allowed operations of $B$ force the output $b=0$ and AO is respected. However, in this case, the causal box protocol can be viewed as $C$ always acting first followed by a classical/quantum switch involving $A$ and $B$ controlled on the output of $C$. This is a QC-QC and in particular does not violate any causal inequalities. 
\end{restatable}

At this point, it is important to compare how the set-up assumptions of the process matrix framework, in particular the closed labs (CL) assumption, were formalised here (via our AO and LO conditions) versus in \cite{Wechs_2023}. 

Ref. \cite{Wechs_2023} formulates CL at the coarse-grained level of the cyclic causal model \cite{Barrett_2021} associated with a process, which is associated with a directed graph involving the inputs and outputs of the agents (both the quantum ones as well as classical settings and outcomes) as nodes. Specifically the CL assumption in \cite{Wechs_2023} requires that in the (cyclic) causal model associated to the process, the causal influence of an agent $A_k$'s setting $x_k$ on the systems of any other agent must be ``screened off'' by the output $A_k^O$ and similarly, if the outcome $a_k$ of $A_k$ is causally influenced by any other agents' system, this influence must be screened off by the input $A_k^I$. This was motivated as capturing that the input and output systems $A_k^I$ and $A_k^O$ are the only means by which $A_k$'s ``closed lab'' can interact (be influenced or influence) the outside world, and there is only one such in and output per agent.
We observe that this condition is satisfied by any process matrix by definition and hence any implementation of the Lugano process, in particular the one in \cref{fig:lugano}, necessarily satisfies this condition. The reason is that the causal model for any process matrix would include the following four causal influences for each agent $A_k$, $x_k\rightarrow A_k^O$, $x_k\rightarrow a_k$, $A_k^{I}\rightarrow A_k^O$ and $A_k^{I}\rightarrow a_k$ and any causal connections between different agents only occurs via $A_k^I$ and $A_k^O$ as all such connections are described by the process matrix which is a map from all agents' output systems to all agents' input systems.\footnote{In fact, such a condition is also satisfied by all causal box protocols, even by those such as \cref{ex:trivvio} which do not correspond to valid processes. This is because all causal box protocols have a valid causal model (given by the sequence decomposition of all the causal boxes involved, which is guaranteed to exist). This causal model is fine-grained (as defined in \cite{VilasiniRennerPRA, VilasiniRennerPRL}) and also considers the time stamps of the input and output systems. In causal box protocols, we again have by construction, that for each agent $A_k$, the largest set of possible causal influences are: $x_k\rightarrow  A_k^{O,t'}$ (for all output times/positions $t'$), $x_k\rightarrow a_k$, $A_k^{I,t}\rightarrow A_k^{O,t'}$ and $A_k^{I,t}\rightarrow a_k$ (for all input times/positions $t$ and output times/positions $t'$).
Even allowing all of these influences to exist, we see that the only way $x_k$ could influence the systems of other agents is via one of the set of output systems $A_k^O:=\{A_k^{O,t'}\}_{t'\in \C{T}_k^O}$ and similarly the only way $a_k$ can be influenced by a system of another agent is through the inputs $A_k^I:=\{A_k^{I,t}\}_{t\in \C{T}_k^I}$. Again, this is because the causal box $\C{C}$ of such a protocol is the only thing that allows causal connections between different agents and it is defined only on these systems in and output systems).}

On the other hand, our formalisation of CL via AO and LO are at the fine-grained level of the spacetime implementation. Hence, one implementation of the same process could respect this version of CL while another implementation of the same process might violate it. Indeed, we have seen that violating this formulation of CL then allows for trivial violations of causal inequalities. This is analogous to the fact that a given quantum model that violates Bell inequalities can have its behaviour reproduced in a physical experiment that satisfies the set-up assumptions needed for ``interesting/loophole free'' violations of the inequality or through an experiment that violates these setup assumptions. The latter case includes for instance classical strategies that violate Bell inequalities via communication, and we must therefore impose the set-up assumptions at the fine-grained, spacetime level (as opposed to just the abstract quantum circuit level) via relevant spacelike separations in order to have loophole free experimental tests of Bell inequality violations. 

Violations of our AO and LO conditions can in principle be operationally verified to confirm violation of the closed labs assumption in a physical realisation of a process in spacetime. This is because we have defined them via projectors which are subnormalised pseudo-causal boxes\footnote{That they are pseudo-causal boxes captures an idealisation made for simplicity that the associated measurements are taken to be instantaneous, and does not restrict their physicality.} (\cref{lemma:projcb} and  \cref{lemma:lopseudocb}). In other words, we can think of these subnormalised pseudo-causal boxes as particular outcomes of a physical measurement which occurs with certainty and without disturbing the overall protocol iff the AO respectively LO condition is satisfied. For instance, the projectors of the AO condition allow us to verify whether the input/output of an agent acts on the one message space or not, in a coherent way that does not reveal the spacetime location of the one message.

With the above in mind, our results can be interpreted as suggesting that loophole free experimental tests of causal inequality violations are impossible in a fixed acyclic spacetime. This includes Minkowski as well as curved (globally hyperbolic) spacetimes which are non-dynamical.  Further, we have clarified how this conclusion is mathematically fully consistent with the claim of \cite{Wechs_2023} that causal inequality violating processes can be realised on time-delocalised subsystems, even though it might appear to be conflicting on the surface.

\section{Relaxing the constraints on process box protocols}\label{sec:relaxing}

As highlighted in the previous section, our AO (Acting Once) and LO (Local Order) definitions capture the set-up assumptions inherent in the process matrix framework, at the level of physical quantum protocols in spacetime. We have shown that their violation can result in a trivial ability for two agents to violate a causal inequality maximally with a causally ordered strategy (\cref{ex:trivvio}), and when they are satisfied, our main result shows that we recover QC-QCs, which are known not to violate any such causal inequalities. However, there remains the question of whether there is scope to slightly relax AO and LO and still recover QC-QCs while avoiding trivial violations such as \cref{ex:trivvio}, and similarly whether the process matrix set-up assumptions can be relaxed while still ensuring that violating causal inequalities under those relaxed assumptions remains ``non-trivial'' or ``loophole free''. This question is discussed in further detail below.

For both AO and LO there are two ways we can violate them. For AO, we have the following two options:
\begin{itemize}
    \item Agents are not fully active (i.e., do not receive and/or send a message)
    \item Agents act more than once
\end{itemize}
On the other hand, a violation of LO could consist of 
\begin{itemize}
    \item An agent sends a message only after receiving one but without the strict one-to-one correspondence between input and output positions enforced by the functions $\C{O}_k(t)$
    \item An agent sends a message before receiving one
\end{itemize}

In the case of both AO and LO, we shall call violations via the second option ``strong'' violations and violations via the first option ``weak'' violations. Strong violations of either AO or LO clearly also allow for trivial causal inequality violations as discussed in \cref{ex:trivvio} and thus to protocols that are no longer processes. However, one can consider allowing weak violations of AO and LO, and thus generalise our definition of process box protocols. In the case of such weak violations of spatiotemporal closed labs, we are not aware of any examples that can exploit these to violate causal inequalities.

We provide an example of a causal box protocol which violates AO and LO in the aforementioned weak ways but respects the strong aspects of AO and LO. We will see that this protocol can still be mapped to a behaviourally equivalent QC-QC protocol. 

\begin{restatable}[Weak violations of AO and LO]{example}{nolo}\label{ex:nolo}
    We consider a bipartite causal box protocol with agents Alice and Bob. The global past and future systems split into a target and a control, $P = P_T P_C$ and $F= F_T F_C$. The causal box acts as follows (where once again, we only specify the action on states which can actually occur)
    \begin{gather}
    \begin{aligned}
        V &\ket{\psi, t=1}^{P_T} \ket{0, t=1}^{P_C} \ket{\psi_A, t=5}^{A^O} = \ket{\psi, t=2}^{A^I} \ket{\psi_A, t=6}^{F_T} \ket{0, t=6}^{F_C} \\
        V &\ket{\psi, t=1}^{P_T} \ket{1, t=1}^{P_C} \ket{\psi_A, t=5}^{A^O} \ket{\psi_B, t=3}^{B^O} = \ket{\psi, t=2}^{B^I} \ket{\psi_B, t=4}^{A^I} \ket{\psi_A, t=6}^{F_T} \ket{1, t=6}^{F_C}.
    \end{aligned}
    \end{gather}
    The allowed operations for Alice are all causal boxes which only output at $t=5$ and separately preserve the vacuum and the one-message space. The allowed operations for Bob are all causal boxes which only output at $t=3$ and separately preserve the vacuum and the one-message space.

    This causal box protocol violates both AO and LO. On the one hand, it is possible for Bob to not act at all if the control is in the state $\ket{0}$. On the other hand, there is no injective relationship $\mathcal{O}_A$ between Alice's input and output time stamps. Regardless of the input time, she always outputs at $t=5$. The latter, in particular, allows Alice to apply a measurement which acts on both input times simultaneously. For example, the following Kraus operators define an allowed operation, where for simplicity, we assume that Alice's input message space is fixed to the state $\ket{0}$ 
    \begin{equation}\label{eq:nolo2}
        \sqrt{2} K_{\pm} = \ket{0, t=5}(\bra{0, t=2} \pm \frac{1}{\sqrt{2}} \bra{0, t=4}) + \frac{1}{\sqrt{2}} \ket{1, t=5} \bra{0,t=4}.
    \end{equation}
    We can, however, turn this causal box protocol into a behaviourally equivalent process box protocol satisfying LO and AO via the identifications (see also \cref{fig:figures} for a visualisation of both protocols)
    \begin{gather}
    \begin{aligned}
        \mathbb{C}^{d_{A^I}} \otimes (\ket{t=2} \oplus \ket{t=4}) &\cong \mathbb{C}^{2d_{A^I}} \otimes \ket{t=4} \\
        \mathbb{C}^{d_{A^O}} \otimes \ket{t=5} &\cong \mathbb{C}^{d_{A^O}} \otimes \ket{t=5} \\
        \ket{\Omega} \oplus (\mathbb{C}^{d_{B^I}} \otimes \ket{t=2}) &\cong \mathbb{C}^{d_{B^I}+1} \otimes \ket{t=2} \\
        \ket{\Omega} \oplus (\mathbb{C}^{d_{B^O}} \otimes \ket{t=3}) &\cong \mathbb{C}^{d_{B^O}+1} \otimes \ket{t=3}
    \end{aligned}
    \end{gather}
    where $d_{A^{I/O}}, d_{B^{I/O}}$ are respectively the dimensions of Alice and Bob's input/output spaces. In words, we treat the input time stamp of Alice as part of the message (thus, doubling the message space) and treat vacuum on Bob's wire as a message. The time-crossing operator in \cref{eq:nolo2}, then instead simply acts on superpositions of message states at the same time.

    Since this causal box protocol is behaviourally equivalent to a process box protocol it is thus also behaviourally equivalent to a QC-QC protocol by \cref{theorem:pbtoqcqc}. 
\end{restatable}

\begin{figure}
\centering
\begin{subfigure}{0.4\textwidth}
    \includegraphics[width=\textwidth]{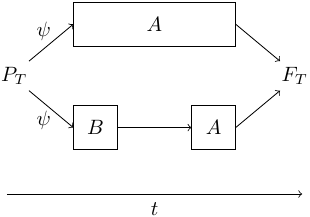}
    \caption{Depending on the value of a control system, either Alice or Bob receives a state $\psi$ at $t=2$. If Bob received the state, then his output is sent to Alice but if Alice receives the state, Bob does not act. Alice always outputs at the same time regardless of when she receives an input. This causal box protocol violates both AO (as Bob sometimes outputs no messages) and LO (as Alice outputs at the same time for two different input times).}
    \label{fig:first}
\end{subfigure}
\hfill
\begin{subfigure}{0.4\textwidth}
    \includegraphics[width=\textwidth]{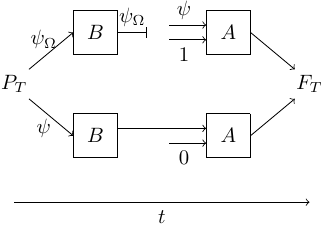}
    \caption{We promote the vacuum to a message (denoted with the state $\psi_\Omega$) so that Bob always acts. Alice now always receives an input at the same time. The information previously carried by the input time is now encoded in an additional qubit. This protocol satisfies AO and LO and is thus behaviourally equivalent to a QC-QC protocol.}
    \label{fig:second}
\end{subfigure}
        
\caption{The two protocols are behaviourally equivalent.}
\label{fig:figures}
\end{figure}

This leaves us with the following open questions: can we cleanly formalise these weak vs strong violations of AO and LO and show that these provide the exact boundary between violations of closed labs that are ``benign'' vs those that lead to ``problematic loopholes''? Note that in particular for AO, \cref{ex:sigproj} shows that there exists no pseudo-causal box which checks that there are 0 \textit{or} 1 messages in a spacetime region without disturbing the protocol. This is in contrast to checking the strictly 1-message case as in AO, and suggests subtleties in the analogous formalisation of the weak relaxation of AO. 

Moreover, can we precisely characterise when a strong violation of AO and/or LO leads to causal inequality violations akin to \cref{ex:trivvio}? Further, can every causal box protocol that only violates AO or LO in a weak sense be mapped to a QC-QC? The example above which involves weak violation does map to a QC-QC, but it also violates LO and AO in a relatively straightforward way. For one, we know that if Bob does not receive a message at $t=2$ he will not receive one at a later time. For another, we can think of this as a fixed order process where Bob always acts first (but potentially trivially on the vacuum state) and Alice acts second. These factors make it easy to find a behaviourally equivalent process box protocol which satisfies AO and LO, but it is unclear how to do this in general, when it could for example be dynamically determined whether an agent acts or not. If there exist protocols which violate AO and/or LO only weakly and cannot be mapped to QC-QCs, are any of them non-causal? We leave these questions for future work.

Let us briefly return to the Lugano process discussed in \cref{ex:lugano}. A question that naturally arises in light of the discussion of this section is whether the implementation of the Lugano process as a causal box protocol discussed in the previous section violates AO in a weak sense akin to \cref{ex:nolo}  or in a strong sense akin to \cref{ex:trivvio}. We see that it is in fact a case of a strong violation of AO. The causal box protocol for the Lugano process does satisfy LO (and also the stronger ALO condition of \cref{def:alo}), but it violates the AO condition. If this causal box protocol were to only weakly violates AO, it would mean that there is a party which does not always act on a non-vacuum message. Looking at \cref{eq:luganoseq}, we clearly see that this is not the case. Charlie always receives a message from $V_1$, while Alice and Bob always receive a message from $V_2$ or $V_3$ depending on Charlie's output. On the other hand, we saw in \cref{ex:lugano} that Charlie receives non-vacuum messages at two distinct times for certain operations of Alice and Bob. It thus follows that the circuit realisation of the Lugano process proposed in \cite{Wechs_2023}, when modelled as a causal box protocol does violate AO in the strong sense.

\section{Discussion and outlook}\label{sec:conclusion}

The main contributions and conclusions of this work are summarised in \cref{sec:intro} and \cref{sec:ruling}, and we will therefore not repeat them here. Instead, we discuss and clarify two important points in light of our results, and then outline some future directions.

{\bf On the role of the control vs time information} A natural question that arises from our result is why the principles of relativistic causality and closed labs in a spacetime single out QC-QCs, or, in particular, controlled superpositions of acyclic (possibly dynamical) causal orders. One physical intuition behind our result is that the information contained in relevant spacetime labels acts as a control for the causal order of agents' operations. As an illustrative example, think of a possible spacetime embedding of a scenario with two orders in which two agents can perform operations on a single message: we can describe the spacetime positions of when Alice and Bob receive their respective (non-vacuum) messages as a state $\ket{t_A, t_B}$ with $t_A \prec t_B$ if Alice acts first and $t_B \prec t_A$ if Bob acts first (note that superpositions are also possible but do not affect the argument at hand). States compatible with Alice acting first are then necessarily orthogonal to states compatible with Bob acting first, as $t_A\neq t_B$ and different spacetime labels are perfectly distinguishable in the causal box framework. This therefore suggests that the set of $N$ spacetime positions associated with when and where each of the $N$ agents ``acts once'' on their respective non-vacuum message encodes the same information as the control degree of freedom in the QC-QC framework (both of these can be in superposition and entangled with other systems in the respective frameworks). Notice that this argument crucially relies not just on relativistic causality but also on the fact that an agent acts exactly on one non-vacuum message, even if this is at a superposition of different spacetime locations (as imposed by the spatiotemporal closed labs condition, or AO and LO). This idea generalises to more than two agents, and also to situations where Alice's choice of operation (via her setting $x_A$) could dynamically determine the temporal order in which future agents Bob and Charlie may act, and could be formalised further using the concept of the effective Choi-Jamiołkowski matrix of a process box (further arguments and results concretely linking the control in QC-QCs to ``time of action'' in the PB framework, and showing a correspondence between control vs time decoherence in the two frameworks, are part of upcoming work \cite{Emilien_inprep}).

{\bf On the notion of indefinite causal order in spacetime contexts}
The concept of ``indefinite causal order (ICO)'' introduced in the process matrix framework has generated growing interest but has also been a subject of considerable debate in spatiotemporal contexts, in particular whether or not existing quantum switch experiments \cite{Procopio_2015, Rubino_2017,Goswami_2018,Wei_2019, Rubino_2022, Guo_2020, Goswami_2020, Taddei_2021} performed in Minkowski spacetime constitute a ``genuine implementation'' or ``simulation'' of ICO \cite{Chiribella2013,Oreshkov_2019, Paunkovi__2020, Felce_2022,Ormrod_2023, VilasiniRennerPRA, VilasiniRennerPRL, de_la_Hamette_2025, Kabel_2025,  vilasini2025}. Our formal results and the physical conclusions drawn in this work
stand independently of which side one takes on this debate (also noting that ``implementation'' vs ``simulation'' here do not have formal, agreed upon definitions), and provide a clear connection between coarse-grained and fine-grained operational views discussed below (both of which are formally and operationally well-defined) on such experiments.

While abstract objects in the QC-QC framework (such as the quantum switch) do exhibit ICO in the sense that they are causally non-separable, the causal box framework (and by extension process box protocols), which describes physical quantum protocols and experiments in spacetime, only describes objects with a definite acyclic causal order (both in the operational and spatiotemporal sense \cite{VilasiniRennerPRA, VilasiniRennerPRL}). 
Nevertheless, we have shown that there is a two-way mapping between process boxes and QC-QCs that preserves relevant operational behaviour, and ensures that each agent acts exactly once on a message. In these spacetime regimes, the process box description provides a fine-grained model for the abstract QC-QC, where the acyclic fine-grained causal structure can become cyclic under coarse-graining (see also \cite{salzger2024mappingindefinitecausalorder} for an explicit link between our results on mapping QC-QCs and causal boxes, and the concept of fine-graining introduced in \cite{VilasiniRennerPRA, VilasiniRennerPRL}). In light of this, we clarify that when we talk about realisations of ICO process matrices or QC-QCs in spacetime, this does not imply that such a realisation also exhibits ``ICO''. Indeed, the concept of ICO is defined in the spacetime agnostic abstract framework while the fine-grained models that account for the spacetime information in their operational description do not exhibit ICO. Thus, one should bear in mind that behavioural equivalence does not imply an equivalence between the operational causal structures of QC-QCs and process boxes as the former can admit indefinite causal order while the latter always has a fixed acyclic causal order explanation. Our results further clarify the operational relationships between these two perspectives in the case of QC-QCs, while shedding light on coarse and fine-grained causal structures of quantum processes.

{\bf Future directions} Below we outline some interesting future directions.

{\it Relaxing our closed labs assumption} As we have discussed in detail in \cref{sec:relaxing}, there are two clear routes for relaxing our spatiotemporal closed labs (i.e. AO + LO) assumption corresponding to weak vs strong violations of AO and/or LO. Based on the examples we have analysed for such cases, some natural questions are: can process box protocols which only weakly violate AO and/or LO still be mapped to behaviourally equivalent QC-QC protocols, and if not, do they still produce causal correlations? For process box protocols which strongly violate at least one of AO or LO, can we clearly characterise the conditions under which this can lead to trivial causal inequality violations (analogous to violating Bell inequalities using communication) vs situations where these still produce causal correlations, and can we quantify ``the resource'' required for such causal simulations of non-causal correlations, analogously to the research program which studies the classical communication cost required to simulate quantum violations of Bell inequalities \cite{TonerBacon2003,Brassard1999,Massar2001,Pironio2003,Brunner2005}? Finally, disregarding AO and LO altogether, and only imposing relativistic causality, what is the largest class of cyclic quantum causal structures that admit a fine-grained acyclic causal structure in the form of a causal box protocol, without involving any post-selection?\footnote{If we allow post-selection, then indeed all such cyclic structures or postselected CTCs can be realised experimentally, by definition of such CTCs.}

{\it Relaxing the finiteness assumption} Another possible relaxation concerns our assumption of finiteness of spacetime. Our results do not depend strongly on the set of relevant spacetime locations being finite. However, showing this rigorously would come with additional technical challenges as it is generally not possible to represent arbitrary causal boxes as just a single CPTP map when $\C{T}$ is not a finite set. Additionally, a naive extension of some of our techniques to this case might require infinite-dimensional message spaces (even prior to considering the Fock spaces of such message spaces) as some of our proofs rely on transferring information encoded in the spacetime label to the message itself. Hence, we would either have to rework our proofs to avoid this or explore alternative methods for such proofs. 

{\it Realisable processes beyond fixed spacetime and role of reference systems} While the fixed spacetime assumption captures precisely the regime of current quantum experiments when described from the perspective of the classical experimenters who perform them, there are
interesting theoretical quantum gravity scenarios \cite{Zych_2019,Castro_Ruiz_2020,Paunkovi__2020,M_ller_2021} that extend beyond this regime, where such a background spacetime
cannot be safely assumed. For instance, this includes the gravitational quantum switch \cite{Zych_2019}, a thought-experiment where a superposition of gravitating masses leads to a superposition of spacetime geometries and thus of causal orders. This motivates the question: which set of higher-order processes are realisable in the regime of such superpositions of acyclic spacetimes, or in more general models of quantum spacetime beyond such linearised superpositions? A related question is what happens if we replace our common background spacetime $\C{T}$ shared by all agents with a model where each agent uses some physical systems as references for space and time (e.g., rods and clocks) relative to which they describe ``when and where'' their operations occur?

Although our present results focused on fixed spacetimes, the general order-theoretic modelling of the spacetime as a partial order and the techniques backing the physical interpretation of our result (discussed at the start of this section) suggest some intuition that might extend beyond this case: if one can ensure distinct ordering of spacetime labels associated to different orders of operations and these labels are perfectly distinguishable, then one might expect a similar structure as QC-QCs to emerge, even if these labels do not arise from a shared background spacetime. This intuition is also supported by the analysis of \cite{schmitt2023operational} which suggests that in the Page-Wootters framework \cite{page1983evolution, Giovannetti_2015} for describing time with quantum clocks (as opposed to a background parameter), under physical constraints on distinguishability of the clock states and forward-ticking clocks, one recovers a similar structure of realisable processes as in the fixed spacetime case here, i.e., QC-QCs. 

In the broader research area, there exist several approaches for formulating quantum reference frames for space and time \cite{Castro_Ruiz_2020, Giacomini_2019, Hoehn_2020, Hoehn_2021, Loveridge_2019, castro2025relative, de_la_Hamette_2020, de_la_Hamette_2025, Kabel_2025}, abstract frameworks for incorporating relational notions of time into the study of processes \cite{Baumann_2022, Guerin_2018, wechs2024, apadula2026}, as well as a recent proposal \cite{vilasini2025} to describe operational events and their localisation relative to a Lab which includes reference degrees of freedom (distinct from QRFs). More generally, building on these ideas, it would be an interesting research program to study whether the QC-QC structure can be derived from physical principles and axioms constraining the reference degrees of freedom used for space and time, without assuming a background spacetime or a particular framework/theoretical model (e.g., linear superpositions of metrics, Page-Wootters framework) for such references or reference frames. 

\bigskip
{\bf Acknowledgements} We thank Cyril Branciard and Emilien de Bank for useful discussions and feedback on parts of this draft. VV thanks Lidia del Rio and Renato Renner for interesting discussions on early versions of the process box framework. MS is supported by National Science Centre, Poland (Preludium project, Classicality and compositionality of general probabilistic theories in infinite dimensions, project no.2024/53/N/ST2/01192 and Opus
project, Categorical Foundations of the Non-Classicality of Nature, project no. 2021/41/B/ST2/03149). VV's research at ETH was supported by an ETH Postdoctoral Fellowship. VV also acknowledges support from the PEPR integrated project EPiQ ANR-22-PETQ-0007 as part of Plan France 2030. 
\printbibliography

\appendix

\section{Appendix}

To keep equations compact, we will sometimes apply CP(TP) maps to states in the Hilbert space. This should be understood as applying the CP(TP) map to the corresponding density matrix, e.g., for the CPTP map given by the partial trace
\begin{equation}
    \text{tr}_{B} \ket{\psi}^{AB} := \text{tr}_{B} \ket{\psi}\bra{\psi}^{AB}.
\end{equation}
Analogously, we will sometimes apply linear operators on Hilbert spaces to density matrices. This should be understood as
\begin{equation}
    V \rho := V \rho V^\dagger.
\end{equation}

\subsection{Useful properties of causal boxes}\label{app:cb}

In this section, we prove a number of useful properties of causal boxes. We will make frequent use of them to prove our main results. The techniques developed here might also be useful in other settings using causal boxes.

\looprep*
\begin{proof}
    It is sufficient to prove the statement for normalised causal boxes as the analogous representation for a subnormalised causal box can be obtained from the one of the underlying normalised causal box.
    
    As $\chi$ is a causality function for $\C{C}$, we can find a sequence representation $V_1,...,V_M$ with $V_1$ preparing a pure normalised state on $\C{F}^{YZ, \C{T}_{1}} \otimes \C{H}^{\alpha_1}$, $V_m: \C{F}^{XY, \C{T}_{m-1}} \otimes \C{H}^{\alpha_{m-1}} \rightarrow \C{F}^{YZ, \C{T}_{m}} \otimes \C{H}^{\alpha_m}$ for $1<m<M$ and $V_M: \C{F}^{XY, \C{T}_{M-1}} \otimes \C{H}^{\alpha_{M-1}} \rightarrow \C{H}^{\alpha_M}$ with $\C{T} = \bigsqcup_{m=1}^{M-1} \C{T}_m$. The map $V := V_{M} \circ_{\alpha_{M-1}} V_{M-1} \circ_{\alpha_{M-2}} ... \circ_{\alpha_1} V_1$, where $\circ_{\alpha_m}$ denotes composition along the system $\C{H}^{\alpha_m}$ only, is an isometry and purification of $\C{C}$. 

    We now define
    \begin{equation}
        V' = \sum_{i} \bra{i}^Y V \ket{i}^Y.
    \end{equation}
    This is an isometry as we now show. Using that the basis in the above is a product basis in terms of the decomposition $\C{F}^{Y, \C{T}} \cong \bigotimes_{m=1}^M \C{F}^{Y, \C{T}_m}$, we find
    \begin{gather}
    \begin{aligned}
        V' &= \sum_{i_1,...,i_M} \bigotimes_{m=1}^{M-1} \bra{i_m}^{Y, \C{T}_m}  V_{M} \circ_{\alpha_{M-1}}  ... \circ_{\alpha_1} V_1 \bigotimes_{m'=1}^M \ket{i_m'}^{Y, \C{T}_m'} \\
        &= \sum_{i_1,...,i_M}  V_{M} \circ_{\alpha_{M-1}} \ket{i_{M-1}}\bra{i_{M-1}}^{Y, \C{T}_M} ... \circ_{\alpha_1} \ket{i_1}\bra{i_1}^{Y, \C{T}_1} V_1 \\
        &= V_{M} \circ_{Y\alpha_{M-1}}  ... \circ_{Y\alpha_1} V_1
    \end{aligned}
    \end{gather}
    where $\circ_{Y \alpha_m}$ denotes sequential composition along $Y$ and $\alpha_m$. Here, we used that $\sum_i \ket{i}\bra{i}^{Y, \C{T}_m}$ is the identity on $\C{F}^{Y, \C{T}_m}$. We thus find that $V'$ can be written as a composition of isometries and it is thus also an isometry. 

    Note that
    \begin{gather}
    \begin{aligned}
        \C{C}'(\cdot) &:= \text{tr}_{\alpha_M} V' \cdot V'^{\dagger} \\
        &= \text{tr}_{\alpha_M} \sum_{ij} \bra{i}^Y V \ket{i}^{Y} \cdot \bra{j}^Y V^{\dagger} \ket{j}^Y \\
        &= \sum_{ij} \bra{i}^Y \text{tr}_{\alpha_M} (V \ket{i}^{Y} \cdot \bra{j}^Y V^{\dagger}) \ket{j}^Y \\
        &= \sum_{ij} \bra{i}^Y \C{C}( \cdot \otimes \ket{i} \bra{j}^Y) \ket{j}^Y.
    \end{aligned}
    \end{gather}
    In particular, $\C{C}'$ is CPTP as it has a purification $V'$. We claim now that $\C{C}' = \C{C}^{\hookrightarrow}$. To show this, let us calculate the Choi-Jamiołkowski representation of $\C{C}'$ (which is well-defined as $\C{C}'$ is CPTP)
    \begin{gather}
    \begin{aligned}
        R_{\C{C}'}(\psi^X \otimes \psi^Z, \varphi^X \otimes \varphi^Z) &= \bra{\psi^Z} \sum_{ij} (\bra{i}^Y \C{C}(\ket{\bar{\psi}^X} \bra{\bar{\varphi}^X} \otimes \ket{i} \bra{j}^Y) \ket{j}^Y) \ket{\varphi^Z} \\
        &= \sum_{ij} \bra{\psi^Z} \bra{i}^Y \C{C}(\ket{\bar{\psi}^X} \bra{\bar{\varphi}^X} \otimes \ket{i} \bra{j}^Y) \ket{j}^Y \ket{\varphi^Z} \\
        &= R_{\C{C}^{\hookrightarrow}}(\psi^X \otimes \psi^Z, \varphi^X \otimes \varphi^Z)
    \end{aligned}
    \end{gather}
    where in the last line we used the definition of $R_{\C{C}^{\hookrightarrow}}$. Note that we can pull the bra and ket into the infinite sum in the second line as the sum converges (because it is equal by definition to $\C{C}'(\ket{\bar{\psi}^X} \bra{\bar{\varphi}^X})$ which in turn is well-defined as $\C{C}'$ is CPTP) and $\rho \mapsto  \bra{\psi^Z} \rho \ket{\varphi^Z}$ is a continuous map (it is norm contractive by Cauchy-Schwarz). 

    Since the Choi-Jamiołkowski representations of $\C{C}'$ and $\C{C}^{\hookrightarrow}$ are the same, these maps must also be the same.
    \end{proof}

\begin{restatable}[Pseudo-causal boxes]{defi}{pseudocb}\label{defi:pseudocb}
    A pseudo-causal box is a system with an input wire $X$ and an output wire $Y$, together with a CPTP map
\begin{equation}
    \C{C}: \mathcal{L}(\F{X}{\C{T}}) \rightarrow \mathcal{L}(\F{Y}{\C{T}})
\end{equation}
that fulfils for all $\C{T}' \subseteq \C{T}$ which are closed from the bottom
\begin{equation}\label{eq:pseudo-causality}
        \text{tr}_{\C{T}\backslash \C{T}'} \circ \C{C} = \text{tr}_{\C{T}\backslash \C{T}'} \circ \C{C} \circ \text{tr}_{\C{T}\backslash \C{T}'}
    \end{equation}
We call \cref{eq:pseudo-causality} the pseudo-causality condition and it should be interpreted in light of \cref{eq:wireisostate} so that the output space of $\text{tr}_{\C{T} \backslash \C{T}'}$ matches the input space of $\C{C}$.

We call $\C{C}: \mathcal{L}(\F{X}{\C{T}}) \rightarrow \mathcal{L}(\F{Y}{\C{T}})$ a subnormalised pseudo-causal box if there exists a (normalised) pseudo-causal box $\C{C}':\C{L}(\F{X}{\C{T}}) \rightarrow \mathcal{L}(\F{YR}{\C{T}})$ and a product state $\ket{\psi}^R = \bigotimes_{t \in \C{T}} \ket{\psi_t, t}^R$ such that 
\begin{equation}
    \C{C}(\cdot) = \bra{\psi} \C{C}'(\cdot) \ket{\psi}.
\end{equation}
\end{restatable}

\projcb*
\begin{proof}
    We define the isometries $V^{I/O, \C{T}'}_{k, count}$ which act as
    \begin{equation}\label{eq:msgcount}
        V^{I/O, \C{T}'}_{k, count} \ket{\psi_n}^{A^{I/O, \C{T}'}_k} \otimes \ket{\phi_m}^{A^{I/O, \C{T} \backslash \C{T}'}_k} = \ket{\psi_n}^{A^{I/O, \C{T}'}_k} \otimes \ket{\phi_m}^{A^{I/O, \C{T} \backslash \C{T}'}_k} \otimes \ket{n, m, T}^{C_k}
    \end{equation}
    where $\ket{\psi_n} \in \vee^n \C{H}^{A^{I/O}_k, \C{T}'}$ is an $n$-message state and $\ket{\phi_m}\in \vee^m \C{H}^{A^{I/O}_k, \C{T} \backslash \C{T}'}$ is an $m$-message state.

    This is a pseudo-causal box. To see this let $\C{T}'' \subsetneq \C{T}$ be a downwards-closed non-empty set (the pseudo-causality condition for $\C{T}'' = \C{T}$ is trivial) and define
    \begin{gather}
    \begin{aligned}
        \C{T}_1 &= \C{T}' \cap \C{T}'' \\
        \C{T}_2 &= (\C{T}\backslash\C{T}') \cap \C{T}'' \\
        \C{T}_3 &= \C{T}' \cap (\C{T}\backslash \C{T}'') \\
        \C{T}_4 &= (\C{T}\backslash\C{T}') \cap (\C{T}\backslash\C{T}'') \\
    \end{aligned}
    \end{gather}
    Notice that $\C{T}_i\cap \C{T}_j=\emptyset$ for $i\neq j$ and $\C{T} = \bigsqcup_{i=1}^4 \C{T}_i$. Due to linearity, it is sufficient to prove pseudo-causality for all pure states. Let
    \begin{equation}
        \ket{\psi} = \sum_{\{n_i\}_{i=1}^4}  \bigotimes_{i=1}^4 \ket{\psi^i_{n_i}}^{A^{I/O, \C{T}_i}_k}
    \end{equation} where $\ket{\psi^i_{n_i}}^{A^{I/O, \C{T}_i}_k} \in \vee^{n_i}\C{H}^{I/O, \C{T}_i}_k$ is an $n_i$-message state (not necessarily normalised) over the set of positions $\C{T}_i$. Then, 
    \begin{gather}
    \begin{aligned}
        \text{tr}_{\C{T} \backslash \C{T}''} V^{I/O, \C{T}'}_{k, count} \ket{\psi} =& \sum_{\{n_i\}_{i=1}^4, \{m_i\}_{i=1}^4} \ket{\psi_{n_1}^1}\bra{\psi_{m_1}^1} \otimes \ket{\psi_{n_2}^2} \bra{\psi_{m_2}^2} \underbrace{\text{tr}(\ket{\psi_{n_3}^3} \bra{\psi_{m_3}^3})}_{\delta_{n_3, m_3} \braket{\psi_{m_3}^3|\psi_{n_3}^3}} \underbrace{\text{tr}(\ket{\psi_{n_4}^4} \bra{\psi_{m_4}^4})}_{\delta_{n_4, m_4} \braket{\psi_{m_4}^4|\psi_{n_4}^4}}\\
        &\underbrace{\text{tr}(\ket{n_1+n_3, n_2+n_4, T}\bra{m_1+m_3, m_2+m_4, T}}_{\delta_{n_1+n_3, m_1+m_3}\delta_{n_2+n_3,m_2+m_4}}) \\
        =& \sum_{\{n_i\}_{i=1}^4, \{m_i\}_{i=1}^4} \ket{\psi_{n_1}^1}\bra{\psi_{m_1}^1} \otimes \ket{\psi_{n_2}^2} \bra{\psi_{m_2}^2} \underbrace{\text{tr}(\ket{\psi_{n_3}^3} \bra{\psi_{m_3}^3})}_{\delta_{n_3, m_3} \braket{\psi_{m_3}^3|\psi_{n_3}^3}} \underbrace{\text{tr}(\ket{\psi_{n_4}^4} \bra{\psi_{m_4}^4})}_{\delta_{n_4, m_4} \braket{\psi_{m_4}^4|\psi_{n_4}^4}}\\
        &\underbrace{\text{tr}(\ket{n_1, n_2, T}\bra{m_1, m_2, T}}_{\delta_{n_1, m_1}\delta_{n_2,m_2}}) \\
        =& \text{tr}_{\C{T} \backslash \C{T}''} V^{I/O, \C{T}'}_{k, count} \text{tr}_{\C{T} \backslash \C{T}''} \ket{\psi}
    \end{aligned}
    \end{gather}
    where we used that since $\C{T}''$ is a proper subset of $\C{T}$ and bottom-closed it cannot contain the maximal element $T$ as the smallest bottom-closed set that contains $T$ is $\C{T}$, that $\delta_{n_i+n_j, m_i+m_j} \delta_{n_j,m_j} = \delta_{n_i,m_i} \delta_{n_j,m_j}$ and that $\text{tr}_{\C{T}\backslash \C{T}''} = \text{tr}_{\C{T}_3 \cup \C{T}_4}$. 

    As $V^{I/O, \C{T}'}_{k, count}$ is also clearly an isometry, it is a normalised pseudo-causal box. Finally, notice that 
    \begin{equation}
        \C{P}^{I,O, \C{T}'}_k = \bra{1,0, T}V^{I/O, \C{T}'}_{k, count} \cdot V^{I/O, \C{T}' \dagger}_{k, count} \ket{1,0,T},
    \end{equation}
    i.e., $\C{P}^{I,O, \C{T}'}_k$ corresponds to an outcome of $V^{I/O, \C{T}'}_{k, count}$, making it a subnormalised pseudo-causal box.
\end{proof}

Since \cref{eq:pseudo-causality} is strictly weaker than the causality condition of causal boxes, any causal box is also a pseudo-causal box. Pseudo-causal boxes are closed under sequential and parallel composition, but not under arbitrary loop composition (e.g., the identity on $\F{X}{\C{T}}$ is a pseudo-causal box but looping the output back to the input diverges). Moreover, sequential composition with a causal box yields a causal box if all input respectively output wires of the pseudo-causal box are composed.

\begin{restatable}[Closure properties of pseudo-causal boxes]{lemma}{itscb}\label{lemma:itscb}
    Let $\C{C}: \C{L}(\F{X}{\C{T}}) \rightarrow \C{L}(\F{Y}{\C{T}}), \C{C}':\C{L}(\F{X'}{\C{T}}) \rightarrow \C{L}(\F{Y'}{\C{T}})$ be subnormalised pseudo-causal boxes. Then,
    \begin{equation}
        \C{C}' \otimes \C{C}
    \end{equation}
    is a subnormalised pseudo-causal box and if $Y = X'$ so is
    \begin{equation}
        \C{C}' \circ \C{C}.
    \end{equation}
    If $\C{C}$ or $\C{C}'$ is subnormalised causal box, then $\C{C}' \circ \C{C}$ is a subnormalised causal box. If $\C{C}$ and $\C{C}'$ are both normalised, then so are their sequential and parallel compositions.
\end{restatable}
\begin{proof}
    If any of the maps are subnormalised, we can instead work with the underlying normalised map. Hence, it is sufficient to prove the statement for normalised (pseudo)-causal boxes.

Since all maps are CPTP, their compositions are also CPTP. Hence, we only need to check that the composed maps respect (pseudo)-causality (\cref{eq:cbreq} or respectively \cref{eq:pseudo-causality}). Let $\C{T}' \subseteq \C{T}$ be an arbitrary bottom-closed set. For the parallel composition, note
\begin{gather}
\begin{aligned}
     \text{tr}_{\C{T}\backslash \C{T}'} \circ (\C{C}' \otimes \C{C}) &= (\text{tr}_{\C{T}\backslash \C{T}'} \circ \C{C}') \otimes (\text{tr}_{\C{T}\backslash \C{T}'} \circ \C{C}) \\
     &=(\text{tr}_{\C{T}\backslash \C{T}'} \circ \C{C}' \circ \text{tr}_{\C{T}\backslash \C{T}'}) \otimes (\text{tr}_{\C{T}\backslash \C{T}'} \circ \C{C} \circ \text{tr}_{\C{T}\backslash \C{T}'}) \\
     &= \text{tr}_{\C{T}\backslash \C{T}'} \circ (\C{C}' \otimes \C{C}) \circ \text{tr}_{\C{T}\backslash \C{T}'} 
\end{aligned}
\end{gather}
where we used the pseudo-causality of $\C{C}$ and $\C{C}'$ and that the trace factorises over subsystems. 

For sequential composition, let $\chi$ ($\chi'$) be a causality function for $\C{C}$ ($\C{C}'$) if it is a causal box and identity if it is a pseudo-causal box. Then,
\begin{gather}
\begin{aligned}
    \text{tr}_{\C{T}\backslash \C{T}'} \circ \C{C}' \circ \C{C} &= \text{tr}_{\C{T}\backslash \C{T}'} \circ \C{C}' \circ \text{tr}_{\C{T}\backslash \chi'(\C{T}')} \circ \C{C} \\
     &= \text{tr}_{\C{T}\backslash \C{T}'} \circ \C{C}' \circ \text{tr}_{\C{T}\backslash \chi'(\C{T}')} \circ \C{C} \circ  \text{tr}_{\C{T}\backslash \chi( \chi'(\C{T}'))}\\
     &= \text{tr}_{\C{T}\backslash \C{T}'} \circ \C{C}' \circ \C{C} \circ  \text{tr}_{\C{T}\backslash \chi( \chi'(\C{T}'))}.
        \end{aligned}
\end{gather}
Here, we used (pseudo)-causality of $\C{C}$ and $\C{C}'$ in each line. If $\chi, \chi'$ are both trivial (i.e., both $\C{C}, \C{C}'$ are pseudo-causal boxes), this is the pseudo-causality condition. If either of them is an actual causality function, then this is the causality condition and $\C{C} \circ \C{C}'$ is a causal box.
\end{proof}

It is often more convenient to work with pure maps instead of CPTP maps. While the Stinespring dilation guarantees the existence of purifications of (pseudo)-causal boxes it is a priori unclear whether these are (pseudo)-causal boxes themselves (we can say that an isometry $V$ is a (pseudo)-causal box if the CPTP map $V \cdot V^\dagger$ is a (pseudo)-causal box). The lemma below shows that we can always find purifications where this is true.

\begin{restatable}[Purifications of causal boxes]{lemma}{purification}\label{lemma:purification}
    Let $\C{C}: \C{L}(\F{X}{\C{T}}) \rightarrow \C{L}(\F{Y}{\C{T}})$ be (pseudo-)causal box. Let $V: \F{X}{\C{T}} \rightarrow \F{Y}{\C{T}} \otimes \C{F}^{\alpha, T}$ be a purification of $\C{C}$ such that $T \succ  t$ for all $t \in \C{T}$. Then, $V$ is a (pseudo-)causal box. Moreover, such purifications exist for any choice of the finite-dimensional Hilbert space $\C{H}^\alpha$, in particular, when $\C{H}^\alpha = \mathbb{C}$ is trivial.
\end{restatable}
\begin{proof}
    Since $\C{C}$ is a CPTP map it has a purification $V': \F{X}{\C{T}} \rightarrow \F{Y}{\C{T}} \otimes \C{H}^{\alpha'}$. Due to the existence of minimal purifications where the cardinality of the purifying system is at most that of the input Hilbert space, we can choose $\C{H}^{\alpha'}$ to be an infinite-dimensional separable Hilbert space. As all such spaces are isomorphic, for any finite-dimensional $\C{H}^{\alpha}$, there exists a unitary isomorphism $U: \C{H}^{\alpha'} \rightarrow \C{F}^{\alpha, T}$ and $V:= U \circ V'$ is again a purification of $\C{C}$. Moreover, it is a (pseudo-)causal box. To see this, let $\chi$ be a causality function for $\C{C}$ if $\C{C}$ is a causal box or identity if $\C{C}$ is a pseudo-causal box and let $\C{T}' \subseteq \C{T}$ be a bottom-closed set (note that $\C{T}'$ is then also a bottom-closed set in $\C{T} \cup T$),
    \begin{gather}
    \begin{aligned}
        \text{tr}_{(\C{T} \cup T) \backslash \C{T}'} V &= \text{tr}_{\C{T} \backslash \C{T}'} \circ \text{tr}_T V \\
        &= \text{tr}_{\C{T} \backslash \C{T}'} \circ \text{tr}_\alpha V \\
        &= \text{tr}_{\C{T} \backslash \C{T}'} \circ \C{C} \\
        &= \text{tr}_{\C{T} \backslash \C{T}'} \circ \C{C} \circ \text{tr}_{\C{T}\backslash \chi(\C{T}')} \\
        &= \text{tr}_{(\C{T} \cup T) \backslash \C{T}'} \circ V \circ \text{tr}_{(\C{T} \cup T)\backslash \chi(\C{T}')}
    \end{aligned}
    \end{gather}
    where we used (pseudo)-causality of $\C{C}$, that $V$ is a purification of $\C{C}$ and that there is no input time $T$, so we can freely trace over it.
\end{proof}

\begin{restatable}[A pseudo-causal box into the one-message space]{lemma}{vone}\label{lemma:vone}
   Let $\C{H}^X$ be a finite-dimensional Hilbert space and $\C{T}$ a finite totally ordered spacetime. There exists a pseudo-causal box $\C{C}_{one}: \C{L}(\F{X}{\C{T}}) \rightarrow \C{L}(\F{X}{\C{T}})$ which acts as identity on the one-message space and whose image is the one-message space.  
\end{restatable}
\begin{proof}
    W.l.o.g., we assume $\C{T} = \{1,...,M\}$. Let $\C{T}' := \{M+1,...,2M\}$. We define $V_{one}: \C{F}^{X, \C{T}} \rightarrow \C{F}^{Y, \C{T} \cup \C{T}'} \otimes \C{H}^{\alpha}$ via its action on the spanning set of states $\{\ket{\psi_n, t} \otimes \bigotimes_{t'> t} \ket{\phi_{t'}, t'}|t\in \C{T}, \ket{\psi_n} \in \vee^n \C{H}^X, n \in \mathbb{N}\backslash \{0\}\}$ as follows
    \begin{equation}
        V_{one} \ket{\psi_n, t} \otimes \bigotimes_{t'> t} \ket{\phi_{t'}, t'} = \begin{cases}
            \ket{\psi_1, t} \otimes \bigotimes_{t'> t} \ket{\phi_{t'}, t+M} \otimes \ket{1}^\alpha, &n=1\\
            \ket{0, M}  \otimes \ket{\psi_n, t+M} \otimes \bigotimes_{t'> t} \ket{\phi_{t'}, t' +M}\otimes \ket{1}^\alpha, &n> 1\\
            \ket{0, M} \otimes \ket{\Omega, \C{T}'} \otimes \ket{0}^\alpha, & n=0
        \end{cases}
    \end{equation}
    This is an isometry. The three cases in the definition above are the action of $V_{one}$ on orthogonal subspaces. Restricted to each of these, it clearly acts as an isometry. Hence, we only need to check that it keeps these subspaces orthogonal. This is clear when comparing the case $n=0$ to the other two due to the system $\alpha$. Considering the case $n=1$, we see that if $t < M$, then the output has vacuum at $M$, and is thus orthogonal to any state obtained in the $n>1$ case. If $t=M$, then the input state must be $\ket{\psi_1, t=M}$ and the output has vacuum on $\C{T}'$ while in the $n>1$ case we always have non-vacuum on $\C{T}'$. 

    This is a pseudo-causal box (once we trace out $\alpha$). For some $r \in \C{T}$, we can explicitly rewrite
    \begin{equation}
        V_{one} = V_2 \circ V_1 
    \end{equation}
    where $V_1$ is simply the restriction of $V_{one}$ to time steps less or equal $r$ and $V_2: \C{F}^{X, \C{T}^{>r}} \otimes \C{H}^{\alpha} \rightarrow \C{F}^{X, \C{T}^{>r} \cup \C{T}'} \otimes \C{H}^{\alpha}$ (where $\C{T}^{>r} = \{r+1,...,M\}$) acts as
    \begin{equation}
        V_2 (\bigotimes_{t'>r} \ket{\psi_{t'}, t'}^{X} \otimes \ket{i}^{\alpha}) = \begin{cases}
            V_{one}|_{\C{F}^{X, \C{T}^{>r}}}\bigotimes_{t'>r} \ket{\psi_{t'}, t'}, &i=0 \\
            \bigotimes_{t'>r} \ket{\psi_{t'}, t'+M} \otimes \ket{1}^\alpha,& i=1
        \end{cases}
    \end{equation}
    Both $V_1$ and $V_2$ are isometries. For $V_1$ this is obvious as it is a restriction of an isometry. For $V_2$, notice first that the two cases in the definition correspond to the action of $V_2$ on orthogonal subspaces and the map is an isometry when restricted to either of them. The map $V_2$ maps these orthogonal subspaces to orthogonal subspaces, as in the $i=0$ case we have non-vacuum on the output in $\C{T}$ while in the $i=1$ case, we have vacuum on the output in $\C{T}$.
    
    Hence, we have
    \begin{equation}
        \text{tr}_{>r} \text{tr}_{\alpha} V_{one} = (\text{tr} V_2) \circ \text{tr}_{>r}V_1 = \text{tr}_{\alpha} \text{tr}_{>r}V_1 \circ \text{tr}_{>r}
    \end{equation}
    where we used that $V_2$ is an isometry, i.e., in particular trace-preserving. As this can be done for any $r \in \C{T}$, $\text{tr}_{\alpha} V_{one}$ defines a pseudo-causal box.

    The pseudo-causal box $\C{C}_{one} := \text{tr}_\alpha \text{tr}_{\C{T}'} V_{one}$ now satisfies the conditions in the lemma statement as can be seen straightforwardly from the definition of $V_{one}$. 
\end{proof}

\begin{restatable}[Extending maps to causal boxes]{lemma}{extension}\label{lemma:extension}
    Let $\C{C}: \C{L}(\C{H}^{X,\C{T}}) \rightarrow \C{L}(\C{H}^{Y, \C{T}})$ be a CPTP map such that $\C{H}^X, \C{H}^Y$ are finite-dimensional Hilbert spaces and $\C{T}$ is totally ordered and finite and it holds that
    \begin{equation}\label{eq:extendone1}
        \C{C}\ket{\psi, t} \in \C{L}(\C{H}^{Y, \geq t})
    \end{equation}
    for all $t \in \C{T}$ and all $\ket{\psi} \in \C{H}^X$. Then, there exists a pseudo-causal box $\C{C}': \C{L}(\F{X}{\C{T}}) \rightarrow\C{L}(\F{Y}{\C{T}})$ such that
    \begin{equation}\label{eq:extendone2}
        \C{C}'|_{\C{L}(\C{H}^{X, \C{T}})} = \C{C}
    \end{equation}
    and the image of $\C{C}'$ lies in the 1-message space, i.e.,
    \begin{equation}\label{eq:extendone3}
        \text{Im}(\C{C}') \subseteq \C{L}(\C{H}^{Y, \C{T}}).
    \end{equation}
    If 
    \begin{equation}\label{eq:extendone4}
        \C{C}\ket{\psi, t} \in \C{L}(\C{H}^{Y, > t})
    \end{equation}
    for all $t \in \C{T}$ and all $\ket{\psi} \in \C{H}^X$, then $\C{C}'$ can be chosen to be a causal box.
\end{restatable}
\begin{proof}
    We define $\C{C}' = \C{C} \circ \C{C}_{one}$ where $\C{C}_{one}: \C{L}(\F{X}{\C{T}}) \rightarrow \C{L}(\F{X}{\C{T}}) $ is the pseudo-causal box whose image is the one-message space and which acts as identity on the one-message space (the existence of such pseudo-causal boxes is shown in \cref{lemma:vone}). This is well-defined as the image of $\C{C}_{one}$ is the one-message space, satisfies \cref{eq:extendone2} as $\C{C}_{one}$ acts as identity on the one-message space and \cref{eq:extendone3} as $\C{C}$ maps into the one-message space. 

    Hence, we only need that this is a (pseudo-)causal box. Clearly, it is CPTP as the composition of CPTP maps is CPTP again. We now check the (pseudo)-causality condition. Let $V: \C{H}^{X, \C{T}} \rightarrow \C{H}^{Y, \C{T}} \otimes \C{H}^{\alpha}$ be a purification of $\C{C}$. For $r \in \C{T}$, we can write (up to restricting the RHS to the one-message space)
    \begin{equation}
        V = V_2 \circ V_1
    \end{equation}
    where $V_1 = (\ket{\Omega, < r}\bra{\Omega, < r}^{X} \otimes \ket{\Omega}^{\alpha})\oplus V|_{\C{H}^{X, < r}}$ and $V_2 = (\ket{\Omega, \geq r}\bra{\Omega, \geq r}^{X} \otimes \mathbb{1}^{\alpha}) \oplus V|_{\C{H}^{X, \geq r}}$. Note that the image of $V_2$ lies in $(\C{H}^{Y, \geq r} \oplus \ket{\Omega, \geq r}^{Y})\otimes \C{H}^{\alpha}$ due to \cref{eq:extendone1}. These are isometries as they are direct sums of isometries with orthogonal images.

    Hence, we have
    \begin{equation}\label{eq:extendone5}
        \text{tr}_{Y, \geq r} \text{tr}_{\alpha} V = (\text{tr} V_2) \circ \text{tr}_{Y, \geq r}V_1 = (\text{tr}_{X, \geq r} \otimes \text{tr}_{\alpha}) \circ  \text{tr}_{Y, \geq r}V_1 =  \text{tr}_{\alpha} \text{tr}_{Y, \geq r}V_1 \circ \text{tr}_{X, \geq r} = \text{tr}_{Y, \geq r} \text{tr}_{\alpha} V \circ \text{tr}_{X, \geq r}
    \end{equation}
    where we used that $V_2$ is an isometry, i.e., in particular trace-preserving and in the last equality we are allowed to commute $\text{tr}_{X, \geq r}$ past $V_1$ as $V_1$ has no such output wire to trace out. 
    
    We then have
    \begin{equation}\label{eq:extendone6}
        \text{tr}_{>r} \circ \C{C}' = \text{tr}_{>r} \circ \C{C} \circ \C{C}_{one} = \text{tr}_{>r} \circ \C{C} \circ \text{tr}_{>r} \circ \C{C}_{one} = \text{tr}_{>r} \circ \C{C} \circ \C{C}_{one} \circ \text{tr}_{>r}.
    \end{equation}
    where we used that $\C{C}_{one}$ is a pseudo-causal box. If \cref{eq:extendone4} holds, we can show causality by noting that the image of $V_2$ lies in $(\C{H}^{Y, > r} \oplus \ket{\Omega, > r}^{Y})\otimes \C{H}^{\alpha}$ which then leads to an analogous equation to \cref{eq:extendone5}, namely that $\text{tr}_{Y, > r} \text{tr}_{\alpha} V = (\text{tr} V_2) \circ \text{tr}_{Y, > r}V_1 = (\text{tr}_{X, \geq r} \otimes \text{tr}_{\alpha}) \circ  \text{tr}_{Y, > r}V_1 =  \text{tr}_{\alpha} \text{tr}_{Y, > r}V_1 \circ \text{tr}_{X, \geq r} = \text{tr}_{Y, > r} \text{tr}_{\alpha} V \circ \text{tr}_{X, \geq r}$. 
\end{proof}

\begin{restatable}[Causality functions remove maximal elements]{lemma}{chikill}\label{lemma:chikill}
    Let $\chi$ be a causality function on a finite spacetime $\C{T}$. For all bottom-closed $\C{T}' \subseteq \C{T}$ and all $t \in \C{T}'$ which are maximal in $\C{T}'$, it holds that $t \not \in \chi(\C{T}')$.
\end{restatable}
\begin{proof}
    We can write $\C{T}' = \C{T}' \backslash \{t\} \cup \C{T}^{\preceq t}$ where $\C{T}^{\preceq t}= \{t' \in \C{T}|t' \preceq t\}$ (which is bottom-closed by definition). The set $\C{T}' \backslash \{t\}$ is also bottom-closed as $\C{T}'$ is bottom-closed and $t$ is maximal in $\C{T}'$. Then, since $\chi$ respects unions of bottom-closed sets
    \begin{equation}
        \chi(\C{T}') = \chi(\C{T}' \backslash \{t\}) \cup \chi(\C{T}^{\preceq t})
    \end{equation}
    We have that $t \not \in \chi(\C{T}' \backslash \{t\}) \subsetneq \C{T}' \backslash \{t\}$. But we also have that $\chi(\C{T}^{\preceq t}) \subsetneq \C{T}^{\preceq t}$ is bottom-closed but $\C{T}^{\preceq t}$ is the smallest bottom-closed set that contains $t$. Hence, $t \not \in \chi(\C{T}^{\preceq t})$ and thus $t \not \in \chi(\C{T}')$.
\end{proof}

\begin{restatable}[Certain implies non-disturbing]{lemma}{nondist}\label{lemma:nondist}
    Let $\C{C}: \C{L}(\C{F}^{X, \C{T}}) \rightarrow \C{L}(\C{F}^{XY,\C{T}})$ be a normalised or subnormalised causal box and $\C{P}:\C{L}(\C{F}^{X, \C{T}}) \rightarrow \C{L}(\C{F}^{X,\C{T}} \otimes \C{H}^{R, T})$ a normalised pseudo-causal box such that
    \begin{equation}
        \C{P} = \sum_{n =0}^\infty \C{P}_n \otimes \ket{n,T}\bra{n,T}^R
    \end{equation}
    where $\C{P}_n$ is a projector for each $n$. 
    If
    \begin{equation}
        \lc{\C{C} \circ \C{P}_n} = 0
    \end{equation}
    for all $n \neq 0$, then
    \begin{equation}
        \lc{\C{C} \circ \C{P}_0} = \C{C}^{\hookrightarrow}.
    \end{equation}
\end{restatable}
\begin{proof}
    Consider a purification of $\C{C}$, $V: \C{F}^{X, \C{T}} \rightarrow \C{F}^{XY,\C{T}} \otimes \C{F}^{\alpha, T}$ where $T \succ t$ for all $t \in \C{T}$. This can be chosen to be a causal box due to \cref{lemma:purification} for any choice of finite-dimensional $\C{H}^{\alpha}$. For $n \neq 0$, we have (writing $\C{P}_n = P_n \cdot P_n$)
    \begin{gather}
    \begin{aligned}
        0 &= \lc{\C{C} \circ \C{P}_n} \\
        &= \lc{\text{tr}_{\alpha} V \circ \C{P}_n} \\
        &= \text{tr}_{\alpha} \circ \lc{V P_n \cdot P_n V^\dagger}
    \end{aligned}
    \end{gather}
    where we used associativity of the link product in the last line. 
    
    As $\text{tr}_{\alpha}$ is completely positive this implies that $\lc{V P_n \cdot P_n V^\dagger}=0$. Calculating the Choi-Jamiołkowski representation of the latter, we find for any $\ket{\psi}, \ket{\varphi} \in \F{Y}{\C{T}} \otimes \F{\alpha}{T} \otimes \F{R}{T}$ and any choice of basis $\ket{i}$ of $\F{X}{\C{T}}$
    \begin{gather}
    \begin{aligned}
        0 &= \sum_{ij} \bra{\psi} \bra{i} V P_n \ket{i} \bra{j} P_n V^\dagger \ket{j} \ket{\varphi} \\
        &= \sum_{i} \bra{\psi} \bra{i} V P_n \ket{i} \sum_j \bra{j} P_n V^\dagger \ket{j} \ket{\varphi}.
    \end{aligned}
    \end{gather}
    This can only be true if $\sum_{i} \bra{\psi} \bra{i} V P_n \ket{i} = 0$ for all $\ket{\psi} \in \F{Y}{\C{T}} \otimes \F{\alpha}{T} \otimes \F{R}{T}$. We then find using that $\sum_n P_n$ is a decomposition of unity (this follows from the fact that $\C{P}$ is normalised)
    \begin{gather}
    \begin{aligned}
        R_{V^{\hookrightarrow}}(\ket{\psi}, \ket{\varphi}) &= \sum_{ij}  \bra{\psi} \bra{i} V \ket{i} \bra{j} V^\dagger \ket{j} \ket{\varphi} \\
        &= \sum_{ij} \sum_{nm} \bra{\psi} \bra{i} V P_n \ket{i} \bra{j} P_m V^\dagger \ket{j} \ket{\varphi} \\
        &= \sum_{nm} \sum_{ij} \bra{\psi} \bra{i} V P_n \ket{i} \bra{j} P_m V^\dagger \ket{j} \ket{\varphi} \\
        &= \sum_{ij} \bra{\psi} \bra{i} V P_0 \ket{i} \bra{j} P_0 V^\dagger \ket{j} \ket{\varphi} \\
        &= R_{(V \circ P_0)^{\hookrightarrow}}(\ket{\psi}, \ket{\varphi}).
    \end{aligned}
    \end{gather}
    Here, we used that the sum converges absolutely and to the same value for any choice of basis. Hence, we can choose the basis to be one such that $P_n \ket{i} \in \{0, 1\}$ in which case the exchange of sums corresponds simply to a different order of summation which is allowed due to absolute convergence. 

    We thus find $V^{\hookrightarrow} = (V \circ P_0)^{\hookrightarrow}$ and the statement follows after tracing out the purifying system.
    \end{proof}

\begin{restatable}{coro}{normalise}\label{coro:normalise}
    Let $\C{C}: \C{L}(\C{F}^{X, \C{T}}) \rightarrow \C{L}(\C{F}^{XY,\C{T}})$ be a normalised causal box and $\C{P}:\C{L}(\C{F}^{X, \C{T}}) \rightarrow \C{L}(\C{F}^{X,\C{T}} \otimes \C{H}^{R, T})$ a normalised pseudo-causal box such that
    \begin{equation}
        \C{P} = \sum_{n =0}^\infty \C{P}_n \otimes \ket{n,T}\bra{n,T}^R
    \end{equation}
    where $\C{P}_n$ is a projector for each $n$. If
    \begin{equation}
        \text{tr}(\lc{\C{C} \circ \C{P}_0}) = 1
    \end{equation}
    then 
    \begin{equation}
        \lc{\C{C} \circ \C{P}_0} = \C{C}^{\hookrightarrow}.
    \end{equation}
\end{restatable}
\begin{proof}
    Since $\C{P}$ is normalised, we have that 
    \begin{gather}
    \begin{aligned}
        1 &= \text{tr}(\lc{\C{C} \circ \C{P}}) \\
        &= \sum_{n=0}^\infty \text{tr} (\lc{\C{C} \circ \C{P}_n}) \\
        &= \text{tr}(\lc{\C{C} \circ \C{P}_0}) + \sum_{n=1}^\infty \text{tr} (\lc{\C{C} \circ \C{P}_n}) \\
        &= 1 + \sum_{n=1}^\infty \text{tr} (\lc{\C{C} \circ \C{P}_n})
    \end{aligned}
    \end{gather}
    Since $\C{P}_n$ is a subnormalised pseudo-causal box, $ \text{tr} (\lc{\C{C} \circ \C{P}_n})$ is non-negative. Thus, the above equation implies that
    \begin{equation}
        \text{tr} (\lc{\C{C} \circ \C{P}_n}) = 0 
    \end{equation}
    for $n \neq 0$ which is only possible if
    \begin{equation}
         \lc{\C{C} \circ \C{P}_n} = 0
    \end{equation}
    as $\lc{\C{C} \circ \C{P}_n}$ must be a (subnormalised) state. Hence, the conditions of \cref{lemma:nondist} are satisfied and we arrive at the conclusion.
\end{proof}

\subsection{Useful properties of process box protocols}\label{app:pbp}
\begin{restatable}[Normalisation condition for AO]{coro}{aonormalise}\label{coro:aonormalise}
    Let $\PBP$ be a causal box protocol. Then, it satisfies AO if for all $k, x_k$
    \begin{equation}
        \text{tr} (\CM{\C{C}}{(\C{P}^O_k \circ \M \circ \C{P}^I_k)}) =  1.
    \end{equation}
\end{restatable}
\begin{proof}
This follows immediately from \cref{coro:normalise}. We choose $\C{C} \otimes \bigotimes_{k=0}^{N+1} \M$ as the normalised causal box and $\C{V}^I_{count} \otimes \C{V}^O_{count}$ as the normalised pseudo-causal box, where 
\begin{equation}
\C{V}^{I/O}_{count} := \bigotimes_{k=0}^{N+1} \C{V}^{I/O, \C{T}}_{k,count}
\end{equation}
and $\C{V}^{I/O, \C{T}}_{k,count} = V^{I/O, \C{T}}_{k,count} \cdot V^{I/O, \C{T} \dagger}_{k,count}$ (see \cref{eq:msgcount} for the definition). The condition in \cref{coro:normalise} then holds for the outcome $\bigotimes_{k=0}^{N+1} (\C{P}^I_k \otimes \C{P}^O_k)$ by assumption. Thus, we have that
\begin{equation}
    \lc{(\bigotimes_{k=0}^{N+1} (\C{P}^I_k \otimes \C{P}^O_k))\circ (\C{C} \otimes \bigotimes_{k=0}^{N+1} \M)} = \CM{\C{C}}{(\C{P}^O_k \circ \M \circ \C{P}^I_k)} = \CM{\C{C}}{\M}.
\end{equation}
Since this holds for all $k, x_k$, the causal box protocol satisfies AO.
\end{proof}

\begin{restatable}[Comb decomposition of LO]{lemma}{lopseudocb}\label{lemma:lopseudocb}
    The maps 
    \begin{equation}\label{eq:loin}
    \sum_{t \in \C{T}^I_k} P^{I, t}_k \otimes \ket{0, t}^\alpha
\end{equation}
and
\begin{equation}\label{eq:loout}
    \sum_{t \in \C{T}^O_k} P^{O, t}_k \otimes \bra{0, t}^\alpha.
\end{equation}
where $\C{H}^{\alpha} = \mathbb{C}$ are subnormalised pseudo-causal boxes. From this it follows that $\C{P}_{eff,k}(\M)$ (related to the above maps via \cref{eq:lodecomp}) is a subnormalised causal box.
\end{restatable}
\begin{proof}
    The map in \cref{eq:loin} is equal to
    \begin{equation}
        P^I_k \circ \bigotimes_{t\in \C{T}^I_k} ((\mathbb{1}^{A^I_k,t} - \ket{\Omega, t}\bra{\Omega, t}^{A^I_k})   \otimes \ket{0, t}^\alpha + \ket{\Omega, t}\bra{\Omega, t}^{A^I_k} \otimes \ket{\Omega, t}^\alpha)
    \end{equation}
    where $\mathbb{1}^{A^I_k, t}$ is the identity on $\C{F}^{A^I_k, t}$ and $P^I_k$ projects $\FI$ onto its one-message space $\C{H}^{A^I_k, \C{T}^I_k}$. To see this, notice that $P^I_k = \sum_{t\in \C{T}^I_k} P^{I, t}_k$ and that 
    \begin{gather}
    \begin{aligned}
        P^{I, t}_k &\circ \bigotimes_{t'\in \C{T}^I_k} ((\mathbb{1}^{A^I_k,t'} - \ket{\Omega, t'}\bra{\Omega, t'}^{A^I_k})   \otimes \ket{0, t'}^\alpha + \ket{\Omega, t'}\bra{\Omega, t'}^{A^I_k} \otimes \ket{\Omega, t'}^\alpha) \\
        &=P^{I, t}_k \circ ((\mathbb{1}^{A^I_k,t} - \ket{\Omega, t}\bra{\Omega, t}^{A^I_k})   \otimes \ket{0, t}^\alpha \otimes \bigotimes_{t' \neq t} \ket{\Omega, t}\bra{\Omega, t}^{A^I_k} \otimes \ket{\Omega, t}^\alpha) \\
        &= P^{I, t}_k \circ (\mathbb{1}^{A^I_k,t}    \otimes \ket{0, t}^\alpha \otimes \bigotimes_{t' \neq t} \ket{\Omega, t}\bra{\Omega, t}^{A^I_k} \otimes \ket{\Omega, t}^\alpha)\\
        &= P^{I, t}_k \otimes \ket{0, t}^\alpha
    \end{aligned}
    \end{gather}
    where we used the wire isomorphism in the last equality.
    
    Note that $P^I_k$ is a subnormalised pseudo-causal box due to \cref{lemma:projcb} while $\bigotimes_{t\in \C{T}^I_k} ((\mathbb{1}^{A^I_k,t} - \ket{\Omega, t}\bra{\Omega, t}^{A^I_k})   \otimes \ket{0, t}^\alpha + \ket{\Omega, t}\bra{\Omega, t}^{A^I_k} \otimes \ket{\Omega, t}^\alpha)$ is a normalised pseudo-causal box, hence, by \cref{lemma:itscb} their composition is again a subnormalised pseudo-causal box.
    
    On the other hand, \cref{eq:loout} is equal to (where $P_k^O= \sum_{t\in \C{T}^O_k} P^{O, t}_k$)
    \begin{equation}\label{eq:loout2}
        \bra{\Omega, \C{T}}^{\alpha \beta} V_{comp, k} \circ P^O_k
    \end{equation}
    where $V_{comp,k}$ acts locally at each time step $t$ as (where $\ket{\psi_n, t}$ is an $n$-message state)  
    \begin{equation}
        V_{comp,k} (\ket{\psi_n, t}^{A^O_k} \bigotimes_m \ket{0, t}^\alpha) = \begin{cases} \ket{\psi_n, t}^{A^O_k} \bigotimes_{m-n} \ket{0, t}^\alpha \otimes \ket{\Omega,t}^{\beta}, &m \geq n \\
        \ket{\psi_n, t}^{A^O_k} \bigotimes_{n-m} \ket{0, t}^\alpha \otimes \ket{0,t}^{\beta}, &m < n\end{cases}%
    \end{equation}
    In the above, the zero-fold tensor product on $\alpha$ corresponds to the vacuum state $\ket{\Omega,t}^\alpha$, hence this state is obtained whenever $m=n$.
    This is an isometry and a pseudo-causal box as the output at $t$ only depends on the input at $t$. Hence, $\bra{\Omega, \C{T}}^{\alpha \beta} V_{comp, k}$ is a subnormalised pseudo-causal box by definition and \cref{eq:loout2} is also a subnormalised pseudo-causal box by \cref{lemma:projcb,lemma:itscb}.

    To see that \cref{eq:loout,eq:loout2} are indeed equal note that $P^O_k$ projects onto the subspace where there is exactly one message on $A^O_k$ anywhere in $\C{T}^O_k$ and that $\bra{\Omega, \C{T}}^{\alpha \beta} V_{comp, k}$ projects onto the subspace where there is an equal number of messages on $A^O_k$ and on $\alpha$ at every time step. Thus, both together project onto the subspace where there is exactly one message on $A^O_k$ and $\alpha$ at the same $t$, which is precisely what \cref{eq:loout} does.

    Hence, $\C{P}_{eff, k}(\M)$ is a subnormalised causal box by \cref{lemma:itscb} as the sequential composition of a subnormalised pseudo-causal box, a causal box and then again a subnormalised pseudo-causal box.
\end{proof}

\aloimpliesnlo*

\begin{proof}
    Using that $\M$ satisfies ALO, we immediately find  
    \begin{equation}
        \C{P}^{O, r}_k \circ \M \circ \C{P}^{I, t}_k = \C{P}^{O, r}_k \circ \C{P}^{O, \C{O}_k(t)}_k \circ  \M \circ \C{P}^{I, t}_k= \delta_{r, \C{O}_k(t)} \C{P}^{O, \C{O}_k(t)}_k \circ  \M \circ \C{P}^{I, t}_k = \delta_{r, \C{O}_k(t)} \M \circ \C{P}^{I, t}_k.
    \end{equation}
   
    Let now $V_k: \FI \rightarrow \FO \otimes \C{H}^{\alpha}$ be a purification of $\M$. Then, due to the above and the positivity of the trace we have
    \begin{equation}\label{eq:purealo}
        P^{O, r}_k \circ V_k \circ P^{I, t}_k = \delta_{r, \C{O}_k(t)} V_k \circ P^{I, t}_k.
    \end{equation}
    Note that this implies the weaker statement
    \begin{equation}
        V_k \circ P^I_k = P^O_k \circ V_k \circ P^I_k.
    \end{equation}
    Plugging in $P^{I/O}_k = \sum_{t\in \C{T}^{I/O}_k} P^{I/O, t}_k$ on the RHS,
    \begin{gather}
    \begin{aligned}
        V_k \circ P^I_k &= \sum_{t \in \C{T}^I_k, t' \in \C{T}^O_k} P^{O, t'}_k \circ V_k \circ P^{I, t}_k \\
        &= \sum_{t \in \C{T}^I_k} P^{O, \C{O}_k(t)}_k \circ V_k \circ P^{I, t}_k \\
        &= P_{eff,k}(V_k)
    \end{aligned}
    \end{gather}
    where we used \cref{eq:purealo}. Tracing out the purifying system $\alpha$, we find $\M \circ \C{P}^I_k = \C{P}_{eff,k}(\M)$.
\end{proof}

\begin{restatable}[Process box protocols are behaviourally equivalent to protocols satisfying ALO]{lemma}{nlotoalo}\label{lemma:nlotoalo}
    Let $\PBP$ be a process box protocol on a totally ordered spacetime. Then, there exists a process box protocol $\mathfrak{P}' = (\C{C}, (\{\Mp\}_{x_k})_{k=0}^{N+1})$ which satisfies ALO and is behaviourally equivalent to $\mathfrak{P}$. 
\end{restatable}
\begin{proof}
    We define the subnormalised causal box 
    \begin{equation}
        \C{M}_{x_k, sub}^{\prime A_k} = \C{V}_{LO,k} \circ (\M \otimes \C{O}_k) \circ \C{P}^{I}_{LO, k}.
    \end{equation}
    where $\C{P}^I_{LO,k}$ is defined as in \cref{eq:loin}, $\C{O}_k = \sum_{t, t' \in \C{T}^I_k} \ket{\C{O}_k(t)}^\alpha \bra{t}^\alpha \cdot \ket{t'}^\alpha\bra{\C{O}_k(t')}^\alpha$ (note that the maps $\mathcal{O}_k: \C{T}^I_k \rightarrow \C{T}^O_k$ are specified by the given protocol $\mathfrak{P}$ (involved in its LO condition)) and $\C{V}_{LO,k}(\cdot) = V_{LO,k} \cdot V_{LO,k}^\dagger$ where $V_{LO,k}: \FO \otimes \C{F}(\ket{0}^\alpha \otimes \C{H}^{\C{T}^O_k}) \rightarrow \FO \otimes \F{\beta}{\C{T}^O_k} \otimes  \C{F}(\ket{0}^\gamma \otimes \C{H}^{\C{T}^O_k})$ (with $\C{H}^\beta \cong \C{H}^{A^O_k}$) acts separately on each time stamp as (where $\ket{\psi_n, t}$ is an $n$-message state)
    \begin{equation}\label{eq:vlok}
        V_{LO,k} \ket{\psi_n, t}^{A^O_k} \bigotimes_m \ket{0, t}^\alpha = \begin{cases}
            \ket{\psi_n, t}^{A^O_k} \ket{\Omega, t}^{\beta}\ket{\Omega, t}^\gamma, &m=n \\
            (\bigotimes_m \ket{0, t}^{A^O_k})  \ket{\psi_n, t}^\beta \ket{0, t}^\gamma, & n\neq m.
        \end{cases}
    \end{equation}
    This is an isometry and respects pseudo-causality as the output at $t$ depends only on the input at $t$. Hence, $V_{LO,k}$ is a pseudo-causal box. Thus, $\C{M}_{x_k, sub}^{\prime A_k}$ is indeed a subnormalised causal box by \cref{lemma:itscb}. 

    As $\C{M}_{x_k, sub}^{\prime A_k}$ is a subnormalised causal box, there exists by definition (\cref{def:subcb}) a normalised causal box $\C{M}_{x_k, norm}^{\prime A_k}$ such that $\C{M}_{x_k, sub}^{\prime A_k}(\cdot) = \bra{\Omega, \C{T}}^R \C{M}_{x_k, norm}^{\prime A_k}(\cdot) \ket{\Omega, \C{T}}^R$. We now define 
    \begin{equation}
        \Mp = \text{tr}_{R \beta \gamma} \circ \C{M}_{x_k, norm}^{\prime A_k}
    \end{equation}
    As $\C{M}_{x_k, sub}^{\prime A_k}$ is obtained through sequential composition of a map ($\C{P}^I_{LO,k}$) which is CPTP when restricted to the one-message space followed by other CPTP maps, it is also be CPTP when restricted to the one-message space, and so we must have that $\Mp \circ \C{P}^I_k = \text{tr}_{\beta \gamma} \circ \C{M}_{x_k, sub}^{\prime A_k}$ for $\Mp$ to be CPTP.
    
    The map $\Mp$ satisfies ALO as $\C{O}_k \circ \C{P}^{I}_{LO, k} \circ \C{P}^{I, t}_k = \C{P}^{I, t}_k \otimes \C{O}_k(\ket{0, t}\bra{0, t}^{\alpha}) =\C{P}^{I, t}_k \otimes \ket{0, \C{O}_k(t)}\bra{0, \C{O}_k(t)}^{\alpha}$ and from \cref{eq:vlok}, we see that $V_{LO,k}$ maps the subspace $\FO \otimes \ket{0, \C{O}_k(t)}^\alpha$ to $\C{H}^{A^O_k, \C{O}_k(t)} \otimes \C{F}^{\beta, \C{T}^O_k} \otimes (\ket{\Omega, \C{O}_k(t)}^\gamma \oplus \ket{0, \C{O}_k(t)}^\gamma)$. The projector $\C{P}^{O, \C{O}_k(t)}_k$ acts trivially on this space.

    To check AO and behavioural equivalence, note first that $\bra{\Omega, \C{T}}^{\beta\gamma} P^O_k \circ V_{LO,k} = P^{O}_{LO,k}$ (cf. \cref{eq:loout2}). Thus,
    \begin{gather}
    \begin{aligned}
        \CM{\C{C}}{\text{tr}_{\beta \gamma R} &(\ket{\Omega, \C{T}}\bra{\Omega, \C{T}}^{\beta\gamma R} \circ  \C{P}^O_k \circ \C{M}_{x_k, norm}^{\prime A_k} \circ \C{P}^I_k)} \\
        &= \CM{\C{C}}{\text{tr}_{\beta \gamma} (\ket{\Omega, \C{T}}\bra{\Omega, \C{T}}^{\beta\gamma} \circ \C{P}^O_k \circ \C{M}_{x_k, sub}^{\prime A_k} \circ \C{P}^I_k)} \\
        &= \CM{\C{C}}{\C{P}^O_{LO,k} \circ (\M \otimes \C{O}_k) \circ \C{P}^{I}_{LO, k}} \\
        &= \CM{\C{C}}{\C{P}_{eff,k}(\M)} \\
        &= \CM{\C{C}}{\M}
    \end{aligned}
    \end{gather}
    where we used the decomposition of $\C{P}_{eff,k}$ from \cref{eq:lodecomp} and that $\mathfrak{P}$ satisfies LO. As the last line has trace 1, we can drop the projectors $\ket{\Omega, \C{T}}\bra{\Omega, \C{T}}^{\beta\gamma R}, \C{P}^{O}_k, \C{P}^I_k$ by \cref{coro:normalise}, giving us both AO and behavioural equivalence,
    \begin{equation}
        \CM{\C{C}}{ \C{P}^O_k \circ \C{M}_{x_k}^{\prime A_k} \circ \C{P}^I_k } = \CM{\C{C}}{\M} = \CM{\C{C}}{\Mp}
    \end{equation}
\end{proof}

\begin{restatable}[1-message output]{lemma}{cbone}\label{lemma:cbone}
    Let $\PBP$ be a process box protocol on a totally ordered spacetime. Then, there exists a causal box $\C{C}'$ such that the process box protocol $\PBPchoose{\mathfrak{P}'}{\C{C}'}{\M}$ is behaviourally equivalent to $\mathfrak{P}$ and
    \begin{equation}\label{eq:cbonecond}
        \bigotimes_{k=0}^{N+1} \C{P}^I_k  \circ \C{C}' = \C{C}'.
    \end{equation}
\end{restatable}
\begin{proof}
    We define 
    \begin{equation}
        \C{C}' := \bigotimes_{k=1}^{N+1} \C{C}^I_{one, k} \circ \C{C}
    \end{equation}
    where $\C{C}^I_{one, k}: \C{L}(\FI) \rightarrow \C{L}(\FI)$ is a pseudo-causal box which acts as identity on the one-message space and whose image is the one-message space (the existence of such pseudo-causal boxes is shown in \cref{lemma:vone}). As each $\C{C}^I_{one, k}$ is a normalised pseudo-causal box, $\C{C}'$ is a normalised causal box by \cref{lemma:itscb}. Further, \cref{eq:cbonecond} is satisfied as the image of $\C{C}^I_{one,k}$ is the one-message input space of agent $k$. The resulting protocol $\mathfrak{P}'$ is a process box protocol. For this notice that $\bigotimes_{k=1}^{N+1} (\C{C}^I_{one, k} \circ \C{P}^{I}_k) \circ \C{C} =\bigotimes_{k=1}^{N+1} \C{P}^{I}_k\circ\C{C}$ as $\C{C}^I_{one, k}$ is the identity on the one-message space and thus
    \begin{gather}
    \begin{aligned}
        \CM{\bigotimes_{k=1}^{N+1} \C{C}_{one, k} \circ \C{P}^{I}_k \circ \C{C}}{\M} &= \CM{\C{C}}{\M} \\
        \CM{\bigotimes_{k=1}^{N+1} \C{C}_{one, k} \circ \C{P}^{I}_k \circ \C{C}}{(\C{P}^{O}_k \circ \M \circ \C{P}^I_k)} &= \CM{\C{C}}{\M} \\
        \CM{\bigotimes_{k=1}^{N+1} \C{C}_{one, k} \circ \C{P}^{I}_k \circ \C{C}}{\C{P}_{eff,k}( \M )} &= \CM{\C{C}}{\M}
    \end{aligned}
    \end{gather}
    using AO (together with \cref{lemma:checkio}) of $\mathfrak{P}$ in the first two lines and LO of $\mathfrak{P}$ and the fact that $\C{P}_{eff,k}(\M) \circ \C{P}^I_k = \C{P}_{eff,k}(\M)$ in the last line. As the RHS is a normalised causal box and thus has trace 1, by \cref{coro:normalise}, the above equalities still hold true if we remove $\C{P}^I_k$ on the LHS, yielding respectively behavioural equivalence, AO and LO. 
\end{proof}

\subsection{Proofs of all results}\label{app:proofs}

\relabeling*
\begin{proof}
    Relabelling the systems which are composed cannot change the result of the composition. Hence,
    \begin{equation}
        \CM{\C{C}}{\M} = \CM{\C{C}'}{\Mp}.
    \end{equation}
    We now need to show that $\PBPp$ is in fact a causal box protocol (i.e., $\C{C}'$ and $\Mp$ for all $k, x_k$ are causal boxes) and that if $\mathfrak{P}$ is a process box protocol, then so is $\mathfrak{P}'$. 

    Let us first check that $\Mp$ is a causal box. Let $\chi_k$ be a causality function for $\M$. We now define a causality function $\chi'_k$ for $\Mp$ defined on any bottom-closed subsets $\tilde{\C{T}} \subseteq \C{T}'$ as
    \begin{equation}
        \chi'_k: \tilde{\C{T}} \mapsto bc(\C{R}_{I_k} \circ \chi_k \circ \C{R}^{-1}_{O_k}(\tilde{\C{T}} \cap \text{Im}(\C{R}_{O_k})))
    \end{equation}
    where $bc(\cdot)$ refers to the bottom-closure of the set in the argument and $\C{R}^{-1}_{O_k}(\tilde{\C{T}} \cap \text{Im}(\C{R}_{O_k}))$ refers to the preimage of $\tilde{\C{T}} \cap \text{Im}(\C{R}_{O_k})$ under $\C{R}_{O_k}$. As $\chi_k$ only acts on bottom-closed sets, we have to show that $\C{R}^{-1}_{O_k}(\tilde{\C{T}} \cap \text{Im}(\C{R}_{O_k})$ is bottom-closed for all bottom-closed $\tilde{\C{T}} \subseteq \C{T}'$. Let $t_1 \prec t_2$ with $t_2 \in \C{R}^{-1}_{O_k}(\tilde{\C{T}} \cap \text{Im}(\C{R}_{O_k}))$, then, due to \cref{eq:relabelcondition}, it follows that $\C{R}_{O_k}(t_1) \prec \C{R}_{O_k}(t_2) \in \tilde{\C{T}} \cap \text{Im}(\C{R}_{O_k})$. As $\tilde{\C{T}}$ is a bottom-closed set, we thus have $\C{R}_{O_k}(t_1) \in \tilde{\C{T}} \cap \text{Im}(\C{R}_{O_k})$ and hence $t_1 \in \C{R}^{-1}_{O_k}(\tilde{\C{T}} \cap \text{Im}(\C{R}_{O_k}))$. Thus, $\chi'_k$ is well-defined. 

    We now check that $\chi'_k$ is a valid causality function. It commutes with taking the union as $bc, \C{R}_{I_k}, \chi_k, \C{R}^{-1}_{O_k}$ all do, hence their composition does as well. It also respects inclusions for the same reason. Finally, it maps any bottom-closed set to a proper subset of itself. As $\tilde{\C{T}}, \chi'_k(\tilde{\C{T}})$ are both bottom-closed it suffices to show that for a maximal element $t \in \chi'_k(\tilde{\C{T}})$, there exists $t' \in \tilde{\C{T}}$ with $t' \succ t$. As $t$ is maximal, there exists $t'' \in \chi_k ( \C{R}^{-1}_{O_k}(\tilde{\C{T}} \cap \text{Im}(\C{R}_{O_k})))$ such that $t = \C{R}_{I_k}(t'')$. Hence, there exists $\tilde{t} \in \C{R}^{-1}_{O_k}(\tilde{\C{T}} \cap \text{Im}(\C{R}_{O_k}))$ with $\tilde{t}\succ t''$ due to \cref{lemma:chikill} applied to the causality function $\chi_k$. Due to the first line of \cref{eq:relabelcondition}, it then follows $\C{R}_{O_k}(\tilde{t}) \succ \C{R}_{I_k}(t'') = t$ and thus $\chi'_k(\tilde{\C{T}}) \subsetneq \tilde{\C{T}}$. Hence, $\chi'_k$ is a valid causality function.

    We now check that $\Mp$ satisfies causality with the causality function $\chi'_k$. For this, notice first that 
    \begin{gather}
    \begin{aligned}
        \text{tr}_{\C{T}' \backslash \tilde{\C{T}}} \circ \C{R}_{O_k} &= \C{R}_{O_k} \circ\text{tr}_{\C{T} \backslash \C{R}^{-1}_{O_k} (\tilde{\C{T}} \cap \text{Im}(\C{R}_{O_k}))} \\
        \C{R}_{I_k}^{-1} \circ \text{tr}_{\C{T}' \backslash \chi'_k(\tilde{\C{T}})}  &=  \C{R}_{I_k}^{-1} \circ \text{tr}_{\C{T}' \backslash bc(\C{R}_{I_k} \circ \chi_k \circ \C{R}^{-1}_{O_k}(\tilde{\C{T}} \cap \text{Im}(\C{R}_{O_k})))}\\ 
        &=\text{tr}_{\C{T} \backslash \C{R}_{I_k}^{-1}(bc(\C{R}_{I_k} \circ \chi_k \circ \C{R}^{-1}_{O_k}(\tilde{\C{T}} \cap \text{Im}(\C{R}_{O_k})))\cap\text{Im}(\C{R}_{I_k}))} \circ  \C{R}_{I_k}^{-1}\\ 
        &=\text{tr}_{\C{T} \backslash \chi_k(\C{R}^{-1}_{O_k} (\tilde{\C{T}} \cap \text{Im}(\C{R}_{O_k})))} \circ  \C{R}_{I_k}^{-1}.
    \end{aligned}
    \end{gather}
    Using this, we find
    \begin{gather}\label{eq:relabelling1}
    \begin{aligned}
        \text{tr}_{\C{T}'\backslash \tilde{\C{T}}} \circ \Mp \circ \text{tr}_{\C{T}'\backslash\chi(\tilde{\C{T}})}        &= \text{tr}_{\C{T}'\backslash \tilde{\C{T}}} \circ \C{R}_{O_k} \circ \M \circ \C{R}^{-1}_{I_k} \circ \text{tr}_{\C{T}'\backslash\chi(\tilde{\C{T}})} \\
        &= \C{R}_{O_k} \circ\text{tr}_{\C{T} \backslash \C{R}^{-1}_{O_k} (\tilde{\C{T}} \cap \text{Im}(\C{R}_{O_k}))} \circ \M \circ \text{tr}_{\C{T} \backslash \chi_k(\C{R}^{-1}_{O_k} (\tilde{\C{T}} \cap \text{Im}(\C{R}_{O_k})))}\circ  \C{R}_{I_k}^{-1} \\
        &= \C{R}_{O_k} \circ \text{tr}_{\C{T} \backslash \C{R}^{-1}_{O_k} (\tilde{\C{T}} \cap \text{Im}(\C{R}_{O_k}))} \circ \M \circ  \C{R}_{I_k}^{-1} \\
        &= \text{tr}_{\C{T}'\backslash \tilde{\C{T}}} \circ \C{R}_{O_k} \circ \M \circ  \C{R}_{I_k}^{-1} \\
        &= \text{tr}_{\C{T}'\backslash \tilde{\C{T}}} \circ \Mp.
    \end{aligned}
    \end{gather}
    where in the third line we used that $\M$ satisfies causality for $\chi_k$. 

    Let us now check that $\C{C}'$ is a causal box. Let $\chi$ be a causality function for $\C{C}$. We define a causality function $\chi'$ for $\C{C}'$ via
    \begin{equation}
        \chi': \tilde{\C{T}} \mapsto bc(\bigcup_{k=0}^{N+1}\C{R}_{O_k} \circ \chi \circ \C{R}^{-1}_{I_k}(\tilde{\C{T}} \cap \text{Im}(\C{R}_{I_k})).
    \end{equation}
    The argument that this is a valid causality functions is analogous to the argument we used for $\chi'_k$ (now using the second line of \cref{eq:relabelcondition} for the proof that $\chi'$ maps bottom-closed sets to proper subsets of themselves). 

    We then have analogously as before
    \begin{gather}
    \begin{aligned}
        \text{tr}_{\C{T}' \backslash \tilde{\C{T}}} \circ \C{R} &= \C{R} \circ \text{tr}_{\C{T} \backslash (\bigcup_{k=0}^{N+1} \C{R}^{-1}_{I_k} (\tilde{\C{T}} \cap \text{Im}(\C{R}_{I_k}))} \\
        \C{R}^{-1} \circ \text{tr}_{\C{T}' \backslash \chi'(\tilde{\C{T}})}  &= \text{tr}_{\C{T} \backslash \chi(\bigcup_{k=0}^{N+1}\C{R}^{-1}_{I_k} (\tilde{\C{T}} \cap \text{Im}(\C{R}_{I_k})))} \circ \C{R}^{-1}.
    \end{aligned}
    \end{gather}
    where we additionally used that $\chi$ respects (finite) unions. From this it follows that $\C{C}'$ respects causality for $\chi'$ by an analogous calculation as in \cref{eq:relabelling1}. Hence, $\mathfrak{P}'$ is a causal box protocol.
    
    If $\mathfrak{P}$ is a process box protocol, then $\mathfrak{P}'$ must satisfy AO and LO (for the relabelled $\C{O}'_k(\C{R}_{I_k}(t)) = \C{R}_{O_k}(\C{O}_k(t))$). For this notice that
    \begin{gather}
    \begin{aligned}
        \C{R} \circ \C{P}^{I/O}_k &= \C{P}^{\prime I/O}_k \circ \C{R} \\
        \C{R} \circ P^{I/O, t}_k &= P^{\prime I/O, \C{R}_{I_k/O_k}(t)}_k \circ \C{R}
    \end{aligned}
    \end{gather}
    where we added a prime to the projectors on the relabelled Fock spaces to distinguish them from the projectors on the original spaces. Hence, with the definition $\C{O}'_k(\C{R}_{I_k}(t)) = \C{R}_{{O}_k}(\C{O}_k(t))$ and using again that relabelling composed systems does not change the overall composition, we find
    \begin{gather}
    \begin{aligned}
        \CM{\C{C}'}{\C{P}^{\prime O}_k \circ \Mp \circ \C{P}^{\prime I}_k} &= \CM{\C{R} \circ \C{C} \circ \C{R}^{-1}}{\C{P}^{\prime O}_k \circ \C{R} \circ \M \circ \C{R}^{-1} \circ \C{P}^{\prime I}_k} \\
        &= \CM{\C{R} \circ \C{C} \circ \C{R}^{-1}}{\C{R} \circ \C{P}^O_k \circ  \M \circ \C{P}^{I}_k \circ \C{R}^{-1} } \\
        &= \CM{\C{C}}{\M} \\
        &= \CM{\C{C}'}{\Mp}
    \end{aligned}
    \end{gather}
    where we used that $\C{R}$ is a unitary and that $\mathfrak{P}$ satisfies AO in the second to last line and behavioural equivalence between the two protocols in the last line. Thus, the new protocol satisfies AO. Further, we find
    \begin{gather}
    \begin{aligned}
         & \CM{\C{R} \circ \C{C} \circ \C{R}^{-1}}{\sum_{t,t' \in \C{T}^I_k}P^{\prime O, \C{O}'_k(\C{R}_{I_k}(t))}_k( \C{R} \circ \M \circ \C{R}^{-1}(P^{\prime I, \C{R}_{I_k}(t)}_k \cdot P^{\prime I, \C{R}_{I_k}(t')}_k)) P^{\prime O, \C{O}'_k(\C{R}_{I_k}(t'))}_k} \\
        &=  \CM{\C{R} \circ \C{C} \circ \C{R}^{-1}}{\sum_{t,t' \in \C{T}^I_k}\C{R}\circ(P^{O, \C{O}_k(t)}_k\M (P^{I, t}_k \C{R}^{-1}(\cdot) P^{I, t'}_k) P^{O, \C{O}_k(t')}_k)} \\
        &= \CM{\C{C}}{\M}  \\
        &= \CM{\C{C}'}{\C{M}^{\prime A_k}_{x_k}}
    \end{aligned}
    \end{gather}
    where again we used that $\C{R}$ is a unitary and that $\mathfrak{P}$ satisfies LO in the second to last line and behavioural equivalence between the two protocols in the last line. Thus, (since $\mathfrak{P}'$ also already satisfies AO) the new protocol satisfies LO. Thus, it is a process box protocol. 
\end{proof}

\simplifying*

\begin{proof}
We will prove the theorem iteratively by showing that if a process box protocol satisfies the first $n=0,1,2,3,4$ properties, then there exists a behaviourally equivalent process box protocol which satisfies the first $n+1$ properties.

\textbf{Property 1:} Let $\PBP$ be an arbitrary process box protocol on a spacetime $\C{T}$. Let $\C{T}' = \{1,...,M\}$ with $M = |\C{T}|$. Let $\mathcal{R}: \C{T} \rightarrow \C{T'}$ be an isomorphic map such that for all $t_1, t_2 \in \C{T}$ with $t_1 \prec t_2$ it holds that $\mathcal{R}(t_1) \prec \mathcal{R}(t_2)$ (it is always possible to find such a map by the order-extension theorem \cite{szpilrajn1930extension}). By \cref{lemma:relabeling}, the protocol $\mathfrak{P}'$ which is obtained from $\mathfrak{P}$ as in \cref{lemma:relabeling} via the relabelling $\mathcal{R}$ is then behaviourally equivalent to $\mathfrak{P}$ and also a process box protocol.

\textbf{Properties 2\&3:} This is shown in \cref{lemma:tiao}.

\textbf{Property 4:} According to \cref{lemma:cbone}, we can always find a behaviourally equivalent process box protocol where property 4 holds. As the agent operations and the spacetime remain the same, properties 1, 2 and 3 still hold.

\textbf{Property 5:} This follows from \cref{lemma:relabeling} with the relabeling (where $k=1,...,N$) 
\begin{gather}
    \begin{aligned}
        &\C{R}_P(1) = 1 \\
        &\C{R}_{I_k}(m) = 2(m-1)N + 2k \\
        &\C{R}_{O_k}(m) =  2(m-2)N + 2k +1 \\
        &\C{R}_F(M) = 2(M-1)N + 2.
    \end{aligned}
    \end{gather}
    This satisfies the condition of \cref{lemma:relabeling} as $\C{R}_{I_k}, \C{R}_{O_k}$ are obviously strictly order-preserving and for any $k, l$ and $m$
    \begin{gather}
    \begin{aligned}
        \C{R}_{I_k}(m) &= 2(m-1)N + 2k < 2(m-1)N + 2k +1 = \C{R}_{O_k}(m+1) \\
        \C{R}_{O_k}(m) &= 2(m-2)N + 2k +1 < 2mN + 2l = \C{R}_{I_l}(m+1)
    \end{aligned}
    \end{gather}
    where in the bottom line we used that $2k \leq 2N$ and $l \geq 1$. Note that we can choose $\C{T}'$ to be the union of images of $\C{R}_{I_k/O_k}$ together with time stamps for the global past, future and the result wire so that every time stamp $t \in \C{T}'$ is indeed an element of some $\C{T}^{\prime I/O}_k$. 
    Further, properties 1, 2, 3 and 4 of \cref{thm:simplifying} still apply. The new spacetime is still totally ordered, the global past is still the only agent with trivial input and we have that $\C{R}_{I_k}(m) = 2(m-1)N + 2k$ while $\C{R}_{O_k}(m+1) = 2(m-1)N+2k+1$. Thus, $\C{O}_k(t) = t+1$ and a relabelling does not affect the form of the agent operations. A relabelling does not change the number of messages, hence, the image of the causal box is still the 1-message space. Finally, property 5 is satisfied as can be readily seen from the definition of $\C{R}$. For any agent the (relabelled) input time stamps are manifestly even, while the relabelled output time stamps are manifestly odd (which also implies there can be no overlap between input and output time stamps). Further, if $t$ is an input time stamp for two agents $k \neq l$, then $t= 2(m-1)N+ 2k = 2(m' -1)N+2l$ for some $m \neq m'$, which implies $N(m-m') = l-k$ and thus $|l-k| \geq N$. However, $0 < l, k \leq N$ and so this cannot hold. An analogous argument applies to the output time stamps.
\end{proof}

\begin{restatable}[Time-independent agent operations]{lemma}{tiao}
\label{lemma:tiao}
Let $\PBP$ be a process box protocol on a totally ordered spacetime $\C{T}$. Then, there exists a process box protocol $\PBPp$ which is behaviourally equivalent to $\mathfrak{P}$ which satisfies properties 1, 2 and 3 of \cref{thm:simplifying}.
\end{restatable}
\begin{proof}
    W.l.o.g., we assume that the first $n$ agents' input is trivial and the inputs of agents $n+1$ to $N$ are non-trivial. 
    
    By \cref{lemma:nlotoalo}, we can assume that $\M$, for $k>n$, satisfies ALO (see \cref{def:alo}). Since $\M$, and thus in particular its restriction to the one-message space, is a CPTP map, we can write a Kraus decomposition
    \begin{equation}
        \M|_{\C{L}(\C{H}^{A^I_k, \C{T}})} = \sum_i A^{i}_{k, x_k} \cdot A^{i \dagger}_{k, x_k}.
    \end{equation}
    As $\M$ satisfies ALO, we must have that $P^{O, r}_k \circ A^{i}_{k, x_k} \circ P^{I, t}_k =0$ if $r \neq \C{O}_k(t)$ if the agent $k$ has non-trivial input. Hence, we can write
    \begin{equation}
        \M|_{\C{L}(\C{H}^{A^I_k, \C{T}})} = \sum_i \sum_{t, t'} A^{i, t}_{k, x_k} \cdot A^{i, t' \dagger}_{k, x_k} \otimes \ket{\C{O}_k(t)}\bra{t} \cdot \ket{t'} \bra{\C{O}_k(t')}
    \end{equation}
    where $A^{i, t}_{k, x_k} \otimes \ket{\C{O}_k(t)}\bra{t} = A^i_{k, x_k} \circ P^{I, t}$ ($=P^{O, \C{O}_k(t)} \circ A^i_{k, x_k} \circ P^{I, t}$ due to ALO) if $A_k$ has non-trivial input. If $A_k$ has trivial input, then $\M$ must be a state in the 1-message space, as otherwise AO could not be satisfied, and we can write (note that the Kraus operators are states in the Hilbert space here)
    \begin{equation}
        \M = \sum_{i} \sum_{t, t'} A^{i, t}_{k, x_k} A^{i, t' \dagger}_{k, x_k} \otimes \ket{t} \bra{t'}.
    \end{equation}
    We define now if $A_k$ has non-trivial input
    \begin{equation}
        \Mp|_{\C{L}(\C{H}^{A^{\prime I}_k, \C{T}^I_k})} = \sum_i \sum_{t, t', r, r'} A^{i, r}_{k, x_k} \cdot A^{i, r' \dagger}_{k, x_k} \otimes \ket{\C{O}_k(r)}\bra{r} \cdot \ket{r'} \bra{\C{O}_k(r')}^{\C{T}_{\alpha,k}} \otimes \ket{t+1}\bra{t} \cdot \ket{t'} \bra{t'+1}
    \end{equation}
    where the system $\C{T}_{\alpha,k}$ is considered to be part of the message (one can think of $\C{T}_{\alpha, k}$ as storing when the message should have been sent, or equivalently when it should be processed by the causal box, in the original protocol). Clearly, this is time-independent as we can write (somewhat abusing notation)
    \begin{equation}\label{eq:tiao0}
        \Mp|_{\C{L}(\C{H}^{A^{\prime I}_k, \C{T}^I_k})} = \M|_{\C{L}(\C{H}^{A^I_k, \C{T}^I_k})} \otimes \C{T}_{+1}
    \end{equation}
    where $\M|_{\C{L}(\C{H}^{A^I_k, \C{T}})}$ acts only on the message part (i.e., $A^I_k \otimes \C{T}_{\alpha,k}$) of the state and $\C{T}_{+1}$ acts only on the time stamp part of the state (increasing it by one). This can be extended to a causal box on the full Fock space due to \cref{lemma:extension}.

    On the other hand, if $A_k$ has trivial input, we define
    \begin{equation}\label{eq:trivagentnewdef}
        \Mp = \sum_i \sum_{r, r'} A^{i, r}_{k, x_k}  A^{i, r' \dagger}_{k, x_k} \otimes \ket{r}\bra{r'}^{\C{T}_{\alpha,k}} \otimes \ket{t=1} \bra{t=1}
    \end{equation}

    We further define the causal box $\C{C}': \C{L}(\bigotimes_{k=0}^N \C{F}(\C{H}^{A^O_k} \otimes \C{H}^{\C{T}_{\alpha, k}} \otimes \C{H}^{\C{T}^I_k})) \rightarrow \C{L}(\bigotimes_{k=1}^N \C{F}(\C{H}^{A^I_k} \otimes \C{H}^{\C{T}_{\alpha, k}} \otimes \C{H}^{\C{T}^I_k}))$ 
    \begin{equation}
        \C{C}' = \bigotimes_{k=n+1}^{N}\C{COPY}_{\C{T}^I_k,k} \circ (\C{C} \otimes \text{tr}_{\C{T}_{\alpha}\beta}) \circ \bigotimes_{k=1}^{N} \C{V}_{diff,k}
    \end{equation}
    where $\C{T}_{\alpha} = \bigotimes_{k=1}^N \C{T}_{\alpha,k}$,
    \begin{equation}
        \C{COPY}_{\C{T}^I_k,k}|_{\C{L}(\C{H}^{A^{I}_k, \C{T}})} = \sum_{ss' \in \C{T}^I_k} \C{I}^{A^I_k} \otimes \ket{s}\bra{s'}^{\C{T}_{\alpha,k}} \otimes \ket{s} \bra{s} \cdot \ket{s'}\bra{s'}
    \end{equation}
    coherently copies the time stamp after $\C{C}$ acts onto $\C{T}_{\alpha}$ and $\C{V}_{diff, k} = V_{diff, k} \cdot V_{diff, k}^\dagger$ with 
     \begin{gather}
     \begin{aligned}
        V_{diff,k}&|_{\C{H}^{\C{T}_{\alpha,k}, \C{T}^O_k}} =\\
        &= \begin{cases}
        \sum_{t \leq r} \ket{0}^{\beta} \otimes \ket{t-1}^{\C{T}_{\alpha,k}}\bra{r}^{\C{T}_{\alpha,k}} \otimes \ket{r}\bra{t} + \sum_{t > r} \ket{1}^{\beta}\otimes \ket{r}^{\C{T}_{\alpha,k}}  \bra{r}^{\C{T}_{\alpha,k}} \otimes \ket{t}\bra{t}, &k \leq n \\
        \sum_{t \leq r} \ket{0}^{\beta} \otimes \ket{\C{O}_k(t-1)-r}^{\C{T}_{\alpha,k}} \bra{r}^{\C{T}_{\alpha,k}} \otimes \ket{r}  \bra{t} + \sum_{t > r} \ket{1}^{\beta} \otimes \ket{r}^{\C{T}_{\alpha,k}} \bra{r}^{\C{T}_{\alpha,k}} \otimes \ket{t}  \bra{t}, &k>n 
        \end{cases}
    \end{aligned}
    \end{gather}
    where $n$ is the number of agents with trivial input (this map essentially swaps the time stamp and $\C{T}_{\alpha,k}$ and updates the value stored in $\C{T}_{\alpha, k}$ according to the function $t\mapsto t-1$  for agents with trivial input or $t\mapsto \C{O}_k(t-1)-r$ for agents with non-trivial output. In this way, $\C{C}$ takes the correct time information into account while the updated value on $\C{T}_{\alpha, k}$ encodes whether the values of $t, r$ match based on what the LO condition of the initial protocol would tell us. The second term is purely to define the map on the entire 1-message space and can be ignored otherwise). This is an isometry as it preserves orthonormality of the basis given by the state $\ket{r}^{\C{T}_{\alpha,k}} \ket{t}$. If $k \leq n$, it maps these states to $\ket{0}^{\beta} \ket{t-1}^{\C{T}_{\alpha, k}} \ket{r}$ or $\ket{1}^{\beta} \ket{r}^{\C{T}_{\alpha, k}} \ket{t}$. If $k > n$, they get mapped $\ket{0}^{\beta} \ket{\C{O}_k(t-1)-r}^{\C{T}_{\alpha, k}} \ket{r}$ (orthogonality follows from injectivity of $\C{O}_k$) or $\ket{1}^{\beta} \ket{r}^{\C{T}_{\alpha, k}} \ket{t}$.
    
    The map $\C{COPY}_{\C{T}^I_k,k}$ does not change the time stamp of a single message while the map $\C{V}_{diff,k}$ either does not change it or increases it. Hence, we can extend them to full pseudo-causal boxes using \cref{lemma:extension}. By \cref{lemma:itscb}, $\C{C}'$ is then a causal box. 

    This new causal box protocol satisfies ALO w.r.t. $\C{O}_k(t)=t+1$ for all agents with non-trivial input as all agent operations output a single message at $t+1$ when inputting a message at $t$ due to \cref{eq:tiao0}. It also satisfies AO as, due to being extended using \cref{lemma:extension}, the images of $\Mp$ as well as of $\C{COPY}_{\C{T}^I_k,k}$ are the one-message space of the respective agent, i.e., $\C{P}^{\prime O}_k \circ \Mp = \Mp$ and $\C{P}^{\prime I}_k \circ \C{COPY}_{\C{T}^I_k,k} = \C{COPY}_{\C{T}^I_k,k}$ (where we added primes on the projectors on the Fock spaces of the new protocol to distinguish them from the projectors on the Fock spaces of the original protocol). Hence,
    \begin{gather}
    \begin{aligned}
        \CM{\C{C}'}{\C{P}^{\prime O}_k\circ \Mp \circ \C{P}^{\prime I}_k} &= \CM{\bigotimes_{l=n+1}^{N}\C{P}^{\prime I}_l \circ \C{COPY}_{\C{T}^I_l,l} \circ (\C{C} \otimes \text{tr}_{\C{T}_{\alpha}\beta}) \circ \bigotimes_{l'=1}^{N} \C{V}_{diff,l'}}{\C{P}^{\prime O}_k \circ \Mp}\\
        &= \CM{\C{C}'}{ \Mp }
    \end{aligned}
    \end{gather}
    where we used cyclicity of loop composition. 
    
    To check behavioural equivalence with $\mathfrak{P}$ let us first calculate
    \begin{equation}
        \C{V}_{diff,k} \circ \Mp \circ \C{COPY}_{\C{T}^I_k,k} \circ \C{P}^I_k
    \end{equation}
    For agents with non-trivial input, we find (using again that both $\Mp$ and $\C{COPY}_{\C{T}^I_k,k}$ output to the 1-message space)
    \begin{gather}
    \begin{aligned}
        \C{V}_{diff,k}& \circ \Mp   \circ \C{COPY}_{\C{T}^I_k,k}  \circ \C{P}^{I}_k\\
        =&\C{V}_{diff,k} \circ \sum_i \sum_{\substack{tt'rr'\\ss'}} A^{i, r}_{k, x_k} \cdot A^{i, r' \dagger}_{k, x_k} \otimes \ket{\C{O}_k(r)}\braket{r|s} \braket{s'|r'} \bra{\C{O}_k(r')}^{\C{T}_{\alpha,k}} \otimes \ket{t+1}\braket{t|s}  \bra{s}\cdot \ket{s'} \braket{s'|t'} \bra{t'+1}\\
        =& \C{V}_{diff,k} \circ \sum_i \sum_{tt'} A^{i, t}_{k, x_k} \cdot A^{i, t' \dagger}_{k, x_k} \otimes \ket{\C{O}_k(t)} \bra{\C{O}_k(t')}^{\C{T}_{\alpha,k}} \otimes \ket{t+1} \bra{t}\cdot \ket{t'} \bra{t'+1}\\
        =& \sum_i \sum_{tt'} A^{i, t}_{k, x_k} \cdot A^{i, t' \dagger}_{k, x_k} \otimes \ket{0}\bra{0}^{\beta} \otimes \ket{0} \bra{0}^{\C{T}_{\alpha,k}} \otimes \ket{\C{O}_k(t)} \bra{t}\cdot \ket{t'}\bra{\C{O}_k(t')}  \\
        =& \M \circ \C{P}^I_k \otimes \ket{0}\bra{0}^{\beta} \otimes \ket{0} \bra{0}^{\C{T}_{\alpha,k}}
    \end{aligned}
    \end{gather}
    where we used that by definition $V_{diff,k} \ket{\C{O}_k(t)}^{\C{T}_{\alpha,k}} \ket{t} = \ket{\C{O}_k(t) -\C{O}_k(t)}^{\C{T}_{\alpha,k}} \ket{\C{O}_k(t)}$ for agents with non-trivial input.
    
    If the input of $A_k$ is trivial, we find 
    \begin{gather}
    \begin{aligned}
        \C{V}_{diff,k} \circ \Mp &= \sum_i \sum_{r, r'} A^{i, r}_{k, x_k}  A^{i, r' \dagger}_{k, x_k} \otimes \ket{0}\bra{0}^{\C{T}_{\alpha,k}} \otimes \ket{r} \bra{r'} \\
        &= \M  \otimes \ket{0}\bra{0}^{\beta} \otimes \ket{0} \bra{0}^{\C{T}_{\alpha,k}}
    \end{aligned}
    \end{gather}
    where we used that $V_{diff, k}\ket{r}^{\C{T}_{\alpha, k}} \ket{t=1} = \ket{0}^{\C{T}_{\alpha, k}} \ket{t=r}$ for agents with trivial input. Loop composing this with $\C{C} \otimes \text{tr}_{\C{T}_{\alpha} \beta}$ (using associativity of loop composition), we find
    \begin{gather}
    \begin{aligned}
        \CM{\bigotimes_{l=n+1}^{N}\C{COPY}_{\C{T}^I_l,l} \circ \C{P}^I_l \circ (\C{C} \otimes \text{tr}_{\C{T}_{\alpha}\beta}) \circ \bigotimes_{l'=1}^{N} \C{V}_{diff,l'}}{\Mp} &= \CM{\C{C}}{\M \circ \C{P}^I_k} \\
        &= \CM{\C{C}}{\M}
    \end{aligned}
    \end{gather}
    where we used AO (in the form of \cref{coro:aonormalise}) of $\mathfrak{P}$ to obtain the second line. As the second line has trace 1 we can drop the 1-message space projectors on the LHS by \cref{coro:normalise}. Thus, we obtain behavioural equivalence
    \begin{equation}
        \CM{\C{C}'}{\Mp} = \CM{\C{C}'}{\Mp}
    \end{equation}
    
In order to get rid of agents with trivial input, we first consider a relabelling $\C{R}_S(t) = t+1$ for all systems $S \neq P$ and $\C{R}_P(t) = t$. Note that in the latter case $t=1$, and from this it follows straightforwardly that this is a relabelling in accordance with \cref{lemma:relabeling}. Agents with trivial input which are not the global past thus have operations of the form $\rho_{x_k} \otimes \ket{t=2}\bra{t=2}$, where $\rho_{x_k} \in \C{L}(\C{H}^{A^O_k} \otimes \C{H}^{R_k^{x_k}, T})$ in this relabelled protocol which follows directly from applying the relabelling to \cref{eq:trivagentnewdef}. Replace each such operation with the map which maps an $m$-message it receives at $t$ to $\bigodot_m \rho$ (with $\bigodot_{m=0} \rho_{x_k}$ being the vacuum as always) at $t+1$ (this evidently satisfies ALO with $\C{O}_k(t) = t+1$. Thus, the agents $1,...,n$ no longer have trivial input. Replace the causal box $\C{C}'$ with $\C{C}' \otimes \bigotimes_{k=1}^n \ket{0, t=1}\bra{0, t=1}^{A^I_k}$. It is immediate that this is still a process box protocol and is behaviourally equivalent to the one we constructed previously.
\end{proof}

\addcontrol*
\begin{proof}
    Throughout this proof, for a linear map $V: \C{H}^A \rightarrow \C{H}^{BC}$, we will use the notation (note that the states $\ket{\psi_C}$ in the below can be subnormalised)
    \begin{equation}
        \text{Im}_C(V) := \text{span}\{\ket{\psi_C} \in \C{H}^C|\exists \ket{\psi_A} \in \C{H}^A, \ket{\psi_B} \in \C{H}^B: \ket{\psi_C} = \bra{\psi_B} V \ket{\psi_A}\} \subseteq \C{H}^C.
    \end{equation}
    We remark that $\text{Im}(V) \subseteq \C{H}^B \otimes \text{Im}_C(V)$. 
    
    Let $V_1,...,V_{M+1}$ be a sequence representation of $\C{C}$ such that $V_{m+1}$ takes inputs during time step $t=2m+1$ and produces outputs during time step $t=2m+2$. Denote with $V := V_{M+1} \circ ... \circ V_1$ the resulting purification of $\C{C}$. 

    Due to \cref{lemma:cbone}, we can restrict the co-domain of $V_{m+1}$
    \begin{equation}
        V_{m+1}: \C{F}^{A^O_{O(m)}, t=2m+1} \otimes \C{H}^{\beta_m} \rightarrow (\C{H}^{A^I_{I(m)}, t=2m+2}  \oplus \ket{\Omega, t=2m+2}^{A^I_{I(m)}})\otimes \C{H}^{\beta_{m+1}}
    \end{equation}
    where $\beta_m$ denotes an ancilla. W.l.o.g., we assume that $\C{H}^{\beta_{m+1}} = \text{Im}_{\beta_{m+1}} (V_{m+1})$. If this is not the case, replace $\C{H}^{\beta_{m+1}}$ in the above equation with $\text{Im}_{\beta_{m+1}} (V_{m+1})$.

    We now claim that for each $m$, there exists an ancillary space $\alpha_m$ and an isometry 
    \begin{equation}
        \C{E}_m: \C{H}^{\beta_m} \rightarrow \C{H}^{\alpha_m} \otimes \bigoplus_{\C{K}: |\C{K}| \leq m} \ket{\C{K}}
    \end{equation}
    which can be written as
    \begin{equation}\label{eq:addcontrola}
        \C{E}_m = \bigoplus_{\C{K}: |\C{K}| \leq m} \C{E}^{\C{K}_m} \otimes \ket{\C{K}}
    \end{equation}
    where the map $\C{E}^{\C{K}_m}$ acts on the subspace 
    \begin{equation}\label{eq:addcontrola2}
        \text{Im}_{\beta_m}(P^{I, t \leq 2m}_{\C{K}} \circ V_{\leq m})
    \end{equation}
    where $V_{\leq m} := V_{m}\circ…\circ V_1$ and $P^{I, t\leq2m}_{\C{K}}$ is the projector onto the subspace 
    \begin{equation}
        \bigotimes_{k \in \C{K}} \C{H}^{A^I_k, t\leq 2m} 
        \subsetneq \bigotimes_{k=1}^{N+1} \F{A^I_k}{t\leq 2m},
    \end{equation}
    and maps $\bar{V}_1,...,\bar{V}_{M+1}$ which satisfy \cref{eq:addcontrol0} and (where $m \neq 1$)
    \begin{gather}\label{eq:addcontrolb}
    \begin{aligned}
        \bar{V}_1 &= \C{E}_1 \circ V_1 \\
        \bar{V}_{m+1} \circ \C{E}_m &= \C{E}_{m+1} \circ V_{m+1} 
    \end{aligned}
    \end{gather}
    The latter implies that $\bar{V}_1,..., \bar{V}_{M+1}$ is a sequence representation of $\C{C}$,
    \begin{gather}
    \begin{aligned}
        \bar{V}_{M+1} \circ ... \circ \bar{V}_1 &= \bar{V}_{M+1} \circ ... \circ \bar{V}_2 \circ \C{E}_1 \circ V_1 \\
        &= \bar{V}_{M+1} \circ ... \circ \bar{V}_3 \circ \C{E}_2 \circ V_2 \circ V_1 \\
        &= \C{E}_{M+1} \circ V_{M+1} \circ ... \circ V_1 \\
        &= \C{E}_{M+1} \circ V
    \end{aligned}
    \end{gather}
    where we repeatedly used \cref{eq:addcontrolb}. As $\C{E}_{M+1}$ acts only on the ancilla of the output of $V$ and is an isometry, the map $\bar{V} := \C{E}_{M+1} \circ V$ is still a purification of $\C{C}$.

    We prove the above claim by induction over $m$. 
    
    Consider $V_1$ and the spaces $\text{Im}_{\beta_1}(P^{I, t=2}_{I(0)} \circ V_1)$ and $\text{Im}_{\beta_1}(P^{I, \Omega, t=2}_{I(0)} \circ V_1)$, where $P^{I, \Omega, t=2}_{I(0)}$ is the pure projector on the vacuum on the wire $A^I_{I(0)}$ at $t=2$. We claim that these two spaces are orthogonal. To see this, notice first that 
    \begin{equation}\label{eq:addcontrolc}
        \text{Im}(V) \subseteq \bigotimes_{k} \C{H}^{A^I_k, \C{T}} \otimes \C{H}^{\beta_{F}}
    \end{equation}
    which is simply the purified version of property 4 of \cref{thm:simplifying}. Hence, if we project onto the zero or one-message space at $t=2$, we find 
    \begin{gather}
    \begin{aligned}
        \text{Im}(V_{\geq 2} \circ P^{I, t=2}_{I(0)} \circ V_1) &\subseteq \C{H}^{A^I_{I(0)}, t=2} \otimes \bigotimes_{k \neq I(0)} \C{H}^{A^I_k, t>2} \otimes \C{H}^{\beta_{F}} \\
        \text{Im}(V_{\geq 2} \circ P^{I, \Omega, t=2}_{I(0)} \circ V_1) &\subseteq \ket{\Omega, t=2}^{A^I_{I(0)}} \otimes \bigotimes_{k} \C{H}^{A^I_k, t>2} \otimes \C{H}^{\beta_{F}}.
    \end{aligned}
    \end{gather}
    This implies that
    \begin{gather}
    \begin{aligned}
        \text{Im}(V_{\geq 2}|_{\text{Im}_{\beta_1}(P^{I, t=2}_{I(0)} \circ V_1)}) &\subseteq \bigotimes_{k \neq I(0)} \C{H}^{A^I_k, t>2} \otimes \C{H}^{\beta_{F}} \\
        \text{Im}(V_{\geq 2}|_{\text{Im}_{\beta_1}(P^{I, \Omega, t=2}_{I(0)} \circ V_1)}) &\subseteq \bigotimes_{k} \C{H}^{A^I_k, t>2} \otimes \C{H}^{\beta_{F}}.
    \end{aligned}
    \end{gather}
    These are orthogonal subspaces as (using the wire isomorphisms) in the first line the space associated to the agent $I(0)$ is vacuum at all times and in the second line it is the 1-message space. Hence, $V_{\geq 2}$, an isometry, maps $\text{Im}_{\beta_1}(P^{I, t=2}_{I(0)} \circ V_1)$ and $\text{Im}_{\beta_1}(P^{I, \Omega, t=2}_{I(0)} \circ V_1)$ to orthogonal subspaces. 
    This is only possible if these two spaces are themselves orthogonal.

    Thus, there exists an isometry 
    \begin{equation}
        \C{E}_1: \C{H}^{\beta_1} \rightarrow \C{H}^{\alpha_1} \otimes (\ket{\emptyset} \oplus \ket{\{I(0)\}})
    \end{equation}
    for some large enough $\alpha_1$ with
    \begin{gather}
    \begin{aligned}
        \C{E}_1|_{\text{Im}_{\beta_1}(P^{I,t=2}_{I(0)} \circ V_1)} &= \C{E}^{I(0)}_1 \otimes \ket{\{I(0)\}} \\
        \C{E}_1|_{\text{Im}_{\beta_1}(P^{I, \Omega, t=2}_{I(0)} \circ V_1)} &= \C{E}^{\emptyset}_1 \otimes \ket{\emptyset},
    \end{aligned}
    \end{gather}
    where $\C{E}^{I(0)}_1,\C{E}^{\emptyset}_1: \C{H}^{\beta_1} \rightarrow \C{H}^{\alpha_1}$ are arbitrary isometries. Hence, $\C{E}_1$ satisfies \cref{eq:addcontrola} by definition.
    
    We define
    \begin{equation}
        \bar{V}_1 := \C{E}_1 \circ V_1
    \end{equation}
    which then acts as set out in \cref{eq:addcontrol0} by construction.

    Assume now that we have defined the isometries $\C{E}_m, \bar{V}_m$ up to some $m$. We now aim to find $\C{E}_{m+1}$.

    We proceed analogously to the base case $m=1$. 
    Let $\C{K} \subseteq \C{N}$ such that $|\C{K}| \leq m$. Using again \cref{eq:addcontrolc}, we find that 
    
    \begin{equation}
        \text{Im}(V_{>(m+1)} \circ P^{I, t \leq 2m+2}_{\C{K}} \circ V_{\leq (m+1)}) \subseteq \bigotimes_{k \in \C{K}} \C{H}^{A^I_k, \leq t=2m+2} \otimes \bigotimes_{l \not \in \C{K}} \C{H}^{A^I_l, t>2m+2} \otimes \C{H}^{\beta_{F}}
    \end{equation}
    which then implies that
    \begin{equation}
        \text{Im}(V_{>(m+1)}|_{\text{Im}_{\beta_{m+1}}(P^{I, t \leq 2m+2}_{\C{K}} \circ V_{\leq (m+1)})}) \subseteq \bigotimes_{l \not \in \C{K}} \C{H}^{A^I_l, t>2m+2} \otimes \C{H}^{\beta_{F}}.
    \end{equation}
    Hence, $V_{>(m+1)}$, an isometry, maps $\text{Im}_{\beta_{m}}(P^{I, t \leq 2m+2}_{\C{K}}\circ V_{\leq m})$ to orthogonal subspaces for different $\C{K}$ (note that even a difference in one index makes the spaces for two sets $\C{K}, \C{K}'$ in the above equation orthogonal). This is only possible if these spaces are themselves orthogonal. Note further that since $\bigoplus_{\C{K}}P^{I, t \leq 2m+2}_{\C{K}}$ acts as identity on $\text{Im}(V_{\leq (m+1)})$, we find that
    \begin{equation}
        \bigoplus_{\C{K}} \text{Im}_{\beta_{m+1}}(P^{I, t \leq 2m+2}_{\C{K}}\circ V_{\leq (m+1)}) = \text{Im}_{\beta_{m+1}} (V_{\leq (m+1)}) = \C{H}^{\beta_{m+1}}.
    \end{equation}
    We thus define $\C{E}_{m+1}$ by its action on these orthogonal subspaces,
    \begin{equation}
        \C{E}_{m+1}|_{\text{Im}_{\beta_{m}}(P^{I, t \leq 2m+2}_{\C{K}}\circ V_{\leq (m+1)})} = \C{E}_{m+1}^{\C{K}} \otimes \ket{\C{K}}
    \end{equation}
    for some arbitrary isometries
    \begin{equation}
        \C{E}^{\C{K}}_{m+1}: \text{Im}_{\beta_{m}}(P^{I, t \leq 2m+2}_{\C{K}}\circ V_{\leq (m+1)}) \rightarrow \C{H}^{\alpha_{m+1}}
    \end{equation}
    and a large enough space $\alpha_{m+1}$. Thus, we have found an isometry $\C{E}_{m+1}$ which satisfies the decomposition \cref{eq:addcontrola}. We now define (time stamps in the below are implicit)
    \begin{equation}\label{eq:addcontrol4}
        \bar{V}_{m+1} \ket{\psi}^{A^O_{O_m}} \ket{\alpha_m}^{\alpha_m} \ket{\C{K}} := \begin{cases}
            \C{E}_{m+1} \circ V_{m+1} \circ \C{E}^{-1}_{m}, &\ket{\alpha_m}^{\alpha_m} \ket{\C{K}} \in \text{Im}(\C{E}_m) \\
            \ket{\Omega}^{A^I_{I(m)}} \ket{\psi}^{\alpha_{m+1}'} \ket{\alpha_m}^{\alpha''_{m+1}} \ket{\C{K}}, &\ket{\alpha_m}^{\alpha_m} \ket{\C{K}} \in \text{Im}(\C{E}_m)^\perp, I(m) \in \C{K} \\
            \ket{0}^{A^I_{I(m)}} \ket{\psi}^{\alpha_{m+1}'} \ket{\alpha_m}^{\alpha''_{m+1}} \ket{\C{K} \cup I(m)}, &\ket{\alpha_m}^{\alpha_m} \ket{\C{K}} \in \text{Im}(\C{E}_m)^\perp, I(m) \not \in \C{K}
        \end{cases}
    \end{equation}
  
    where $\C{H}^{\alpha'_{m+1}} \cong \C{F}^{A^O_{O(m)}, t=2m+1}, \C{H}^{\alpha''_{m+1}} \cong \C{H}^{\alpha_m}$ and $\C{H}^{\alpha'_{m+1} \alpha''_{m+1}}$ is some subspace of $\C{H}^{\alpha_{m+1}}$ which is orthogonal to $\text{Im}_{\alpha_{m+1}}(\C{E}_{m+1})$ (note that we can always assume such a subspace to exist w.l.o.g., as otherwise we can simply consider a larger ancilla).

    By definition, we then have the desired relation between $\bar{V}_{m+1}, V_{m+1}$
    \begin{equation}
        \bar{V}_{m+1} \circ \C{E}_m = \C{E}_{m+1} \circ V_{m+1} \circ \C{E}^{-1}_m \circ \C{E}_m = \C{E}_{m+1} \circ V_{m+1}.
    \end{equation} 

    The map $\bar{V}_{m+1}$ is an isometry as its restriction to each of the three cases in the definition is clearly an isometry (in each case, it is defined as a composition of isometries). As these cases correspond to orthogonal subspaces and are mapped to orthogonal subspaces the overall map $\bar{V}_{m+1}$ is also an isometry.
    
    Further, it satisfies the property in \cref{eq:addcontrol0}. This is obvious in the bottom two cases of the definition. For the top case, let $\ket{\alpha_m}^{\alpha_m} \ket{\C{K}} \in \text{Im}(\C{E}_m)$ (note that the image of $\C{E}_m$ is spanned by such product states due to \cref{eq:addcontrola} which holds by induction assumption). We need to show that $\bar{V}_{m+1}$ maps the control to $\ket{\C{K}}$ if it outputs vacuum on the agent's wire and to $\ket{\C{K} \cup I(m)}$ if it outputs non-vacuum on the agent's wire. Using \cref{eq:addcontrola,eq:addcontrola2}, we can assume that 
    \begin{equation}
        \C{E}_{m}^{-1}\ket{\alpha_m}^{\alpha_m} \ket{\C{K}} \in \text{Im}_{\beta_m}(P^{I, t \leq 2m}_{\C{K}} \circ V_{\leq m}).
    \end{equation}
    From this, it immediately follows that 
    \begin{gather}
    \begin{aligned}
        P^{I, t=2m+2}_{I(m)} &\circ V_{m+1} \circ \C{E}^{-1}_m \ket{\psi}^{A^O_{O_m}} \ket{\alpha_m}^{\alpha_m} \ket{\C{K}} \in \text{Im}_{A^I_{I(m)}\beta_{m+1}}(P^{I, t \leq 2m+2}_{\C{K} \cup I(m)}\circ V_{\leq (m+1)}) \\
        P^{I, \Omega, t=2m+2}_{I(m)} &\circ V_{m+1} \circ \C{E}^{-1}_m \ket{\psi}^{A^O_{O_m}} \ket{\alpha_m}^{\alpha_m} \ket{\C{K}} \in \text{Im}_{A^I_{I(m)}\beta_{m+1}}(P^{I, t \leq 2m+2}_{\C{K}}\circ V_{\leq (m+1)}).
    \end{aligned}
    \end{gather}
    The statement then follows from the fact that by definition $\C{E}_{m+1}$ maps states in $\text{Im}_{A^I_{I(m)}\beta_{m+1}}(P^{I, t \leq 2m+2}_{\C{K} \cup I(m)}\circ V_{\leq (m+1)})$ to a control $\ket{\C{K} \cup I(m)}$ and states in $\text{Im}_{A^I_{I(m)}\beta_{m+1}}(P^{I, t \leq 2m+2}_{\C{K}}\circ V_{\leq (m+1)})$ to a control $\ket{\C{K}}$ or more explicitly 
    \begin{gather}
    \begin{aligned}
        \bar{V}_{m+1} \ket{\psi}^{A^O_{O_m}} \ket{\alpha_m}^{\alpha_m} \ket{\C{K}} =& \C{E}_{m+1} \circ V_{m+1} \circ \C{E}^{-1}_m \ket{\psi}^{A^O_{O_m}} \ket{\alpha_m}^{\alpha_m} \ket{\C{K}} \\
        =& \C{E}_{m+1}|_{\text{Im}_{\beta_{m+1}}(P^{I, t \leq 2m+2}_{\C{K} \cup I(m)}\circ V_{\leq (m+1)})} \circ P^{I, t=2m+2}_{I(m)} \circ V_{m+1} \circ \C{E}^{-1}_m \ket{\psi}^{A^O_{O_m}} \ket{\alpha_m}^{\alpha_m} \ket{\C{K}} \\
        +& \C{E}_{m+1}|_{\text{Im}_{\beta_{m+1}}(P^{I, t \leq 2m+2}_{\C{K}}\circ V_{\leq (m+1)})} \circ P^{I, \Omega, t=2m+2}_{I(m)} \circ V_{m+1} \circ \C{E}^{-1}_m \ket{\psi}^{A^O_{O_m}} \ket{\alpha_m}^{\alpha_m} \ket{\C{K}} \\
        =& (\C{E}_{m+1}^{\C{K} \cup I(m)} \circ P^{I, t=2m+2}_{I(m)} \circ V_{m+1} \circ \C{E}^{-1}_m \ket{\psi}^{A^O_{O_m}} \ket{\alpha_m}^{\alpha_m} \ket{\C{K}}) \otimes \ket{\C{K} \cup I(m)} \\
        +& (\C{E}_{m+1}^{\C{K}} \circ P^{I, \Omega, t=2m+2}_{I(m)} \circ V_{m+1} \circ \C{E}^{-1}_m \ket{\psi}^{A^O_{O_m}} \ket{\alpha_m}^{\alpha_m} \ket{\C{K}}) \otimes \ket{\C{K}}
    \end{aligned}
    \end{gather}
    where in the second equality we used that $V_{m+1}$ only produces 0- and 1-message states. This completes the proof by induction, and hence we obtain the desired statement.

    Finally, note that if $m=M$, then $I(M) = F \not \in \C{K}$ for all $\C{K}$, and so the second case in \cref{eq:addcontrol4} never occurs while in the first case the output must be non-vacuum to ensure the overall causal box's image lies in the 1-message space. 
\end{proof}

\pbtoqcqc*
\begin{proof}
    It suffices to show that for every process box protocol $\PBP$ satisfying the properties of \cref{thm:simplifying} there exists a QC-QC protocol $\mathfrak{Q}$ which is behaviourally equivalent to $\mathfrak{P}$. 

    Let $\bar{V}_1,..,\bar{V}_{M+1}$ be a sequence representation of $\C{C}$ as in \cref{lemma:addcontrol}.

    We note a number of useful facts. Since any restriction of an isometry is again an isometry, 
    \begin{equation}\label{eq:pbtoqcqc1}
        \bar{V}^{\rightarrow I(m)}_{\C{K},m+1} \otimes \ket{\C{K} \cup I(m)} \bra{\C{K}}+\bar{V}^{\rightarrow \Omega}_{\C{K},m+1} \otimes \ket{\C{K}} \bra{\C{K}}
    \end{equation}
    is an isometry for all $m, \C{K}$ when restricting the control space to the span of $\ket{\C{K}}$.
    For $m=M$, we further have that
    \begin{equation}\label{eq:pbtoqcqc2}
        \bar{V}^{\rightarrow F}_{\C{K}, M+1} \otimes \ket{\C{K} \cup F}\bra{\C{K}} 
    \end{equation}
    is an isometry when restricting the control space to the span of $\ket{\C{K}}$. This immediately implies that $\bar{V}^{\rightarrow F}_{\C{K}, M+1}$ is also an isometry.
    
    Further, for all $m, p$
    \begin{equation}\label{eq:pbtoqcqc3}
        \bar{V}^{\rightarrow I(m) \dagger}_{\C{K},m+1} \bar{V}^{\rightarrow \Omega}_{\C{K},p+1} = 0 
    \end{equation}
    as one of the two maps outputs a non-vacuum state while the other outputs a vacuum state on the agent's wire, which are orthogonal. Similarly, for $m \neq p$ we additionally find
    \begin{gather}
    \begin{aligned}
        \bar{V}^{\rightarrow I(m) \dagger}_{\C{K},m+1} \bar{V}^{\rightarrow I(p)}_{\C{K},p+1} = 0 \\
        \bar{V}^{\rightarrow \Omega \dagger}_{\C{K},m+1} \bar{V}^{\rightarrow \Omega }_{\C{K},p+1} = 0.
    \end{aligned}
    \end{gather}
    This is because one of the two maps outputs on the wire $\alpha_{m+1}$ while the other one outputs onto $\alpha_{p+1}$ which we can treat as orthogonal.
    
    We then also have that
    \begin{equation}
        \bar{V}_{m+1, \C{K}} := \bar{V}^{\rightarrow I(m)}_{\C{K},m+1} +\bar{V}^{\rightarrow \Omega}_{\C{K},m+1} 
    \end{equation}
    is an isometry as
    \begin{gather}
    \begin{aligned}
                \bar{V}_{m+1, \C{K}}^\dagger \bar{V}_{m+1, \C{K}} =& (\bar{V}^{\rightarrow I(m)}_{\C{K},m+1} +\bar{V}^{\rightarrow \Omega}_{\C{K},m+1})^{\dagger} (\bar{V}^{\rightarrow I(m)}_{\C{K},m+1} +\bar{V}^{\rightarrow \Omega}_{\C{K},m+1}) \\
                =& \bar{V}^{\rightarrow I(m) \dagger}_{\C{K},m+1} \bar{V}^{\rightarrow I(m)}_{\C{K},m+1} +\bar{V}^{\rightarrow \Omega \dagger}_{\C{K},m+1} \bar{V}^{\rightarrow \Omega}_{\C{K},m+1} \\
                =& \bra{\C{K}} (\bar{V}^{\rightarrow I(m)}_{\C{K},m+1} \otimes \ket{\C{K} \cup I(m)} \bra{\C{K}}+\bar{V}^{\rightarrow \Omega}_{\C{K},m+1} \otimes \ket{\C{K}} \bra{\C{K}})^\dagger \\
                &(\bar{V}^{\rightarrow I(m)}_{\C{K},m+1} \otimes \ket{\C{K} \cup I(m)} \bra{\C{K}}+\bar{V}^{\rightarrow \Omega}_{\C{K},m+1} \otimes \ket{\C{K}} \bra{\C{K}}) \ket{\C{K}} \\
                =& \mathbb{1}^{A^{I}_{I(m)}, t=2m+2} \otimes \mathbb{1}^{\alpha_m}.
    \end{aligned}
    \end{gather}
    where in the second line we used \cref{eq:pbtoqcqc3} and in the last line we used that \cref{eq:pbtoqcqc1} is an isometry when acting on the control in state $\ket{\C{K}}$.

    We define now 
    \begin{gather}\label{eq:defineqcqcv}
    \begin{aligned}
        V^{\rightarrow I(m)}_{\C{K},m+1}(\cdot) &:= \bra{t=2m+2} \bar{V}^{\rightarrow I(m)}_{\C{K},m+1}|_{A^{O,t=2m+1}_{O(m)} \alpha'_m} ( \cdot \otimes \ket{t=2m+1}) \\
        V^{\rightarrow \Omega}_{\C{K},m+1}(\cdot) &:= \ket{\Omega}^{A^O_{O(m+1)}} \bra{\Omega, t=2m+2}^{A^I_{I(m)}} \bar{V}^{\rightarrow \Omega}_{\C{K},m+1}|_{A^{O,t=2m+1}_{O(m)} \alpha'_m}  (\cdot \otimes \ket{t=2m+1}).
    \end{aligned}
    \end{gather}
    Here, the space $\C{H}^{\alpha'_m}$ is a finite-dimensional subspace of the in general infinite-dimensional space $\C{H}^{\alpha_m}$. Since $V^{\rightarrow I(0)}_{\emptyset, 1}: \C{H}^{P} \rightarrow \C{H}^{A^{I}_{I(0)}} \otimes \C{H}^{\alpha_1}, V^{\rightarrow \Omega}_{\emptyset, 1}: \C{H}^{P} \rightarrow \ket{\Omega}^{A^I_{I(0)}} \otimes \C{H}^{\alpha_1}$ have finite-dimensional domain their images (and thus the span of the union of their images) must be finite-dimensional as well. Hence, there must exist a finite-dimensional subspace $\C{H}^{\alpha'_1} \subseteq \C{H}^{\alpha_1}$ such that $\text{Im}(V^{\rightarrow I(0)}_{\emptyset, 1}), \text{Im}(V^{\rightarrow \Omega}_{\emptyset, 1}) \subseteq \C{H}^{A^{I}_{I(0)}} \otimes \C{H}^{\alpha'_1}$ and we can write
    \begin{equation}\label{eq:restrictcodomain}
        V^{\rightarrow I(0)}_{\emptyset, 1}, V^{\rightarrow \Omega}_{\emptyset, 1}: \C{H}^{P} \rightarrow \C{H}^{A^{I}_{I(0)}} \otimes \C{H}^{\alpha'_1}.
    \end{equation}
    We iterate this procedure and find a subspace $\C{H}^{\alpha'_{m+1}} \subseteq \C{H}^{\alpha_{m+1}}$ for each $m$ with $\text{Im}(V^{\rightarrow I(m)}_{\C{K}, m+1}), \text{Im}(V^{\rightarrow \Omega}_{\C{K}, m+1}) \subseteq \C{H}^{A^{I}_{I(m)}} \otimes \C{H}^{\alpha'_{m+1}}$ for all $\C{K}$ and then restrict the co-domain analogously to \cref{eq:restrictcodomain}. The maps defined in \cref{eq:defineqcqcv} are hence linear maps with finite-dimensional domain and co-domain, making them appropriate for the QC-QC framework.

    We define now the internal operations of the QC-QC, (in the below sum the terms corresponding to $m'=m$ are $V^{\rightarrow I(m')}_{\C{K}_n,m'+1}$)
    \begin{equation}\label{eq:qcqcbuildingblock}
        V^{\rightarrow k_{n+1}}_{\C{K}_{n-1}, k_n} = \sum_{\substack{m',m: m' \geq m\\ O(m) = k_n, I(m') = k_{n+1}}} \underbrace{V^{\rightarrow I(m')}_{\C{K}_n,m'+1} \circ V^{\rightarrow \Omega}_{\C{K}_n,m'} \circ ... \circ V^{\rightarrow \Omega}_{\C{K}_{n},m+1}|_{\C{H}^{A^O_{O(m)} \alpha'_m}}}_{=: V^{\rightarrow k_{n+1}, m -1\rightarrow m'}_{\C{K}_{n-1}, k_n}}
    \end{equation}
    where $\C{K}_n$ is a shorthand for $\C{K}_{n-1} \cup k_n$ and
    \begin{equation}
        V_{n+1} = \sum_{\substack{\C{K}_{n-1} \\ k_n, k_{n+1}}} V^{\rightarrow k_{n+1}}_{\C{K}_{n-1}, k_n} \otimes \ket{\C{K}_{n-1} \cup k_n, k_{n+1}}\bra{\C{K}_{n-1}, k_n}.
    \end{equation}
    This map is defined with an ancillary input wire $\beta_n$ and ancillary output wire $\beta_{n+1}$ both of which can are isomorphic to $\bigoplus_{m=1}^{M+1} \C{H}^{\alpha'_m}$. These maps need to be isometries on their effective input spaces. We can actually show a slightly stronger condition, namely that $V_{n+1}$ is an isometry on the space
    \begin{equation}\label{eq:effspacesuper}
        \bigoplus_{k_n} (\C{H}^{A^O_{k_n}} \otimes \bigoplus_{m: O(m) =k_n} \C{H}^{\alpha'_m} \otimes \bigoplus_{\C{K}_{n-1}} \ket{\C{K}_{n-1}, k_n} )   \end{equation}
    because given our definition of the internal operations this is a superset of the effective input space of $V_{n+1}$ (we recall that the effective input space corresponds to the states which can produced by the preceding internal operations when composed with arbitrary agent operations). Note now that 
    \begin{equation}
        V_{n+1}^\dagger V_{n+1} = \sum_{\substack{\C{K}_{n},\\ k_n, l_n \in \C{K}_n, \\k_{n+1} \not \in \C{K}_n}}V^{\rightarrow k_{n+1} \dagger}_{\C{K}_{n}\backslash k_n, k_n} V^{\rightarrow k_{n+1}}_{\C{K}_{n}\backslash l_n, l_n} \otimes \ket{\C{K}_{n}\backslash k_n, k_n} \bra{\C{K}_n \backslash l_n, l_n}
    \end{equation}
    which implies that the map $V_{n+1}$ is then an isometry on the space in \cref{eq:effspacesuper} if
    \begin{equation}
        \sum_{k_{n+1}}V^{\rightarrow k_{n+1} \dagger}_{\C{K}_{n}\backslash k_n, k_n} V^{\rightarrow k_{n+1}}_{\C{K}_{n}\backslash l_n, l_n} = \delta_{k_n, l_n} \mathbb{1}^{A^{O}_{k_n}} \otimes \bigoplus_{O(m) = k_n} \mathbb{1}^{\alpha'_m}
    \end{equation}
    for all $\C{K}_n$ and $k_n, l_n \in \C{K}_n$. 
    
    Plugging in \cref{eq:qcqcbuildingblock} into the LHS of the above equation yields
    \begin{gather}
    \begin{aligned}
        \sum_{k_{n+1}}V^{\rightarrow k_{n+1} \dagger}_{\C{K}_{n}\backslash k_n, k_n}& V^{\rightarrow k_{n+1}}_{\C{K}_{n}\backslash l_n, l_n} \\
        =& \sum_{k_{n+1}} \sum_{\substack{m',m: m' \geq m\\ O(m) = k_n,\\ I(m') = k_{n+1}}} \sum_{\substack{p', p: p' \geq p\\ O(p) = l_n,\\ I(p') = k_{n+1}}}V^{\rightarrow \Omega \dagger}_{\C{K}_n,m+1}|_{\C{H}^{A^O_{k_n} \alpha'_m}} ... V^{\rightarrow \Omega \dagger}_{\C{K}_{n},m'}\underbrace{V^{\rightarrow I(m') \dagger}_{\C{K}_n,m'+1}  V^{\rightarrow I(p')}_{\C{K}_n,p'+1}}_{=0 \text{ if } m' \neq p'} V^{\rightarrow \Omega}_{\C{K}_n,p'} ...  V^{\rightarrow \Omega}_{\C{K}_{n},p+1}|_{\C{H}^{A^O_{l_n} \alpha'_p}} \\
        =& \sum_{\substack{m',m: m' \geq m\\ O(m) = k_n}} \sum_{\substack{p:\\ O(p) = l_n}} V^{\rightarrow \Omega \dagger}_{\C{K}_n,m+1}|_{\C{H}^{A^O_{k_n} \alpha'_m}} ... V^{\rightarrow \Omega \dagger}_{\C{K}_{n},m'}V^{\rightarrow I(m') \dagger}_{\C{K}_n,m'+1}  V^{\rightarrow I(m')}_{\C{K}_n,m'+1} V^{\rightarrow \Omega}_{\C{K}_n,m'} ...  V^{\rightarrow \Omega}_{\C{K}_{n},p+1}|_{\C{H}^{A^O_{l_n} \alpha'_p}} \\
        =& \sum_{\substack{m:\\ O(m) = k_n}} \sum_{\substack{p:\\ O(p) = l_n}} \Big[V^{\rightarrow \Omega \dagger}_{\C{K}_n,m+1}|_{\C{H}^{A^O_{k_n} \alpha'_m}} ... V^{\rightarrow \Omega \dagger}_{\C{K}_{n},M}\underbrace{V^{\rightarrow F \dagger}_{\C{K}_n,M+1}  V^{\rightarrow F}_{\C{K}_n,M+1}}_{=\mathbb{1}^{A^O_{O(M)}\alpha'_M}} V^{\rightarrow \Omega}_{\C{K}_n,M} ...  V^{\rightarrow \Omega}_{\C{K}_{n},p+1}|_{\C{H}^{A^O_{l_n} \alpha'_p}} \\
        &+ \sum_{M>m' \geq m}  V^{\rightarrow \Omega \dagger}_{\C{K}_n,m+1}|_{\C{H}^{A^O_{k_n} \alpha'_m}} ... V^{\rightarrow \Omega \dagger}_{\C{K}_{n},m'}V^{\rightarrow I(m') \dagger}_{\C{K}_n,m'+1}  V^{\rightarrow I(m')}_{\C{K}_n,m'+1} V^{\rightarrow \Omega}_{\C{K}_n,m'} ...  V^{\rightarrow \Omega}_{\C{K}_{n},p+1}|_{\C{H}^{A^O_{l_n} \alpha'_p}} \Big] \\
        =& \sum_{\substack{m:\\ O(m) = k_n}} \sum_{\substack{p:\\ O(p) = l_n}} \Big[ V^{\rightarrow \Omega \dagger}_{\C{K}_n,m+1}|_{\C{H}^{A^O_{k_n} \alpha'_m}} ... \underbrace{(V^{\rightarrow \Omega \dagger}_{\C{K}_{n},M} V^{\rightarrow \Omega}_{\C{K}_n,M} + V^{\rightarrow I(M) \dagger}_{\C{K}_{n},M} V^{\rightarrow I(M)}_{\C{K}_n,M})}_{=\mathbb{1}^{A^O_{O(M-1)} \alpha'_{M}}} ...  V^{\rightarrow \Omega}_{\C{K}_{n},p+1}|_{\C{H}^{A^O_{l_n} \alpha'_p}} \\
        &+ \sum_{M-1 > m' \geq m} V^{\rightarrow \Omega \dagger}_{\C{K}_n,m+1}|_{\C{H}^{A^O_{k_n} \alpha'_m}} ... V^{\rightarrow \Omega \dagger}_{\C{K}_{n},m'}V^{\rightarrow I(m') \dagger}_{\C{K}_n,m'+1}  V^{\rightarrow I(m')}_{\C{K}_n,m'+1} V^{\rightarrow \Omega}_{\C{K}_n,m'} ...  V^{\rightarrow \Omega}_{\C{K}_{n},p+1}|_{\C{H}^{A^O_{l_n} \alpha'_p}}\Big] \\
        =& \sum_{\substack{m > p:\\ O(m) = k_n, O(p) = l_n}} \underbrace{(V^{\rightarrow \Omega}_{\C{K}_n, m+1} + V^{\rightarrow I(m)}_{\C{K}_n, m+1})^\dagger|_{\C{H}^{A^O_{k_n} \alpha'_m}}(V^{\rightarrow \Omega}_{\C{K}_n, m+1} + V^{\rightarrow I(m)}_{\C{K}_n, m+1})}_{=\mathbb{1}^{A^O_{O(m)}\alpha'_m}} ...  V^{\rightarrow \Omega}_{\C{K}_{n},p+1}|_{\C{H}^{A^O_{l_n} \alpha'_p}} \\
        &+ \sum_{\substack{m < p:\\ O(m) = k_n, O(p) = l_n}} V^{\rightarrow \Omega \dagger}_{\C{K}_n,m+1}|_{\C{H}^{A^O_{k_n} \alpha'_m}}... \underbrace{(V^{\rightarrow \Omega}_{\C{K}_n, p+1} + V^{\rightarrow I(p)}_{\C{K}_n, p+1})^\dagger(V^{\rightarrow \Omega}_{\C{K}_n, p+1} + V^{\rightarrow I(p)}_{\C{K}_n, p+1})|_{\C{H}^{A^O_{k_n} \alpha'_p}}}_{=\mathbb{1}^{A^O_{O(p)}\alpha'_m}}\\
        &+ \sum_{\substack{m:\\ O(m) = k_n, O(m) = l_n}} \underbrace{(V^{\rightarrow \Omega \dagger}_{\C{K}_{n},m+1} V^{\rightarrow \Omega}_{\C{K}_n,m+1} + V^{\rightarrow k_{n+1} \dagger}_{\C{K}_{n},m+1} V^{\rightarrow l_{n+1}}_{\C{K}_n,m+1})|_{\C{H}^{A^O_{k_n} \alpha'_m}}}_{=\mathbb{1}^{A^O_{k_{n}} \alpha'_{m}}} \\
        =& \sum_{\substack{m > p:\\ O(m) = k_n, O(p) = l_n}} \mathbb{1}^{A^O_{O(m)}\alpha'_m} V^{\rightarrow \Omega}_{\C{K}_n,m}...  V^{\rightarrow \Omega}_{\C{K}_{n},p+1}|_{\C{H}^{A^O_{l_n} \alpha'_p}}\\
        &+ \sum_{\substack{m < p:\\ O(m) = k_n, O(p) = l_n}} V^{\rightarrow \Omega \dagger}_{\C{K}_n,m+1}|_{\C{H}^{A^O_{k_n} \alpha'_m}}... V^{\rightarrow \Omega \dagger}_{\C{K}_n,p} \mathbb{1}^{A^O_{O(p)}\alpha'_p}\\
        &+ \sum_{\substack{m:\\ O(m) = k_n, O(m) = l_n}} \underbrace{(V^{\rightarrow \Omega \dagger}_{\C{K}_{n},m+1} V^{\rightarrow \Omega}_{\C{K}_n,m+1} + V^{\rightarrow k_{n+1} \dagger}_{\C{K}_{n},m+1} V^{\rightarrow l_{n+1}}_{\C{K}_n,m+1})|_{\C{H}^{A^O_{k_n} \alpha'_m}}}_{=\mathbb{1}^{A^O_{k_{n}} \alpha'_{m}}} \\
        &= \delta_{k_n, l_n}  \mathbb{1}^{A^{O}_{k_n}} \otimes \bigoplus_{O(m) = k_n} \mathbb{1}^{\alpha'_m}.
    \end{aligned}
    \end{gather}
    Here, we repeatedly used the properties discussed at the beginning of the proof. In the third equality, we split off the last term (corresponding to $m'=M-1$) from the rest of the sum, in the fourth equality we again split off the last term (corresponding to $m'=M$) from the sum over $m'$, combining it with the previously split off term. We repeat this procedure until this is no longer possible which yields the fifth equality. There, in the underbraced equalities, we used that if $U$ is an isometry and $U_1$ is a restriction of $U$, then $U^\dagger U_1 = U_1^\dagger U_1$ which is the identity on the restricted space. In the sixth equality we used that $V^{\rightarrow \Omega}_{\C{K}_n, m}$ outputs a vacuum on the agent wire and $\mathbb{1}^{A^O_k}$ is 0 on vacuum (as it is the identity on the one-message space). Hence, the sums corresponding to $m>p$ and $m<p$ vanish. For the sum corresponding to $m=p$, notice that if $k_n \neq l_n$, it contains no terms, hence, the $\delta_{k_n, l_n}$ in the last line.

    We have thus found that these internal operations define a valid QC-QC. We have to show now that this QC-QC, with appropriate choice of allowed agent operations, is behaviourally equivalent to the original process box protocol.

    For this purpose, we identify an agent's operation $\M$ in the process box protocol whose restriction to the 1-message space we can write as 
    \begin{equation}
        \sum_{i_k} \bar{A}^{i_k}_{k, x_k} \cdot \bar{A}^{i_k \dagger}_{k, x_k} =  \sum_{i_k} \sum_{t, t' \in \C{T}^I_k} A^{i_k}_{k, x_k} \cdot A^{i_k \dagger}_{k, x_k} \otimes \ket{t+1}\bra{t} \cdot \ket{t'}\bra{t'+1} 
    \end{equation}
    according to \cref{lemma:tiao} with the following agent operation in the QC-QC picture
    \begin{equation}
        \sum_{i_k} A^{i_k}_{k, x_k} \cdot A^{i_k \dagger}_{k, x_k}.
    \end{equation}
    Due to linearity and AO, it suffices to show that for any choices of $x_k, i_k$ 
    \begin{equation}
        \CM{\bar{V}}{\bar{A}^{i_k}_{k, x_k}} = \CM{\bar{V}}{\bar{A}^{i_k}_{k, x_k} \circ P^I_k } = \CM{V}{A^{i_k}_{k, x_k}}
    \end{equation}
    where $\bar{V} = \bar{V}_M \circ ... \circ \bar{V}_1$ (and $\bar{V}$ is thus a purification of $\C{C}$ while $V$ is a purification of the constructed QC-QC). 

    Due to associativity of loop composition, we can first calculate the sequential composition of $\bar{V}$ with $\bigotimes_{k=1}^{N+1} P^I_k = \bigotimes_{k=1}^{N+1} \sum_{m_k: I(m_k) = k} P^{I,t=2m_k}_k$. For a given choice of ordering of agents $(k_1,...,k_N)$ and time stamp labels compatible with this ordering, $m_{k_1} < m_{k_2} < ... < m_{k_N}$ we find with the help of the sequence representation of $\bar{V}$ (note that composition is along ancillas only on the RHS below)
    \begin{gather}
    \begin{aligned}
        \bigotimes_{k_n} P^{I,t=2m_{k_n}}_{k_n} \circ \bar{V} =& \bar{V}^{\rightarrow F}_{\C{N}, M+1} \circ \bar{V}^{\rightarrow \Omega}_{\C{N}, M} \circ ... \circ  \bar{V}^{\rightarrow \Omega}_{\C{N},m_{k_N}+2}\circ \bar{V}^{\rightarrow k_N}_{\C{N} \backslash k_N, m_{k_{N}}+1} \circ \bar{V}^{\rightarrow \Omega}_{\C{N}\backslash k_N,m_{k_N}} \circ \\
        &\circ... \circ \bar{V}^{\rightarrow k_{n+1}}_{\{k_1,...,k_n\},m_{k_n}+1} \circ ... \circ \bar{V}^{\rightarrow k_1}_{\emptyset,m_{k_1}+1} \circ ...
    \end{aligned}
    \end{gather}
    In words, the RHS consists of the maps from \cref{lemma:addcontrol} in sequence with the map which sends a non-vacuum state to $k_n$ in the $m_{k_n}-1$-th position for each $k_n$ and the maps which output vacuum in all other positions.

    Composing this with the agent operations, and using that the agent operations map vacuum at $t$ to vacuum at $t+1$ as well as \cref{eq:defineqcqcv} to get rid of the time stamps and the vacuum projectors 
    \begin{gather}
    \begin{aligned}
      \lc{\bigotimes_{k_n}& P^{I,t=2m_{k_n}}_{k_n} \circ\bar{V} \circ \bigotimes_{k=0}^{N+1} \bar{A}^{i_k}_{k, x_k}} \\
        &= A^{i_F}_{F, x_F} \underbrace{V^{\rightarrow F}_{\C{N}, M+1} V^{\rightarrow \Omega}_{\C{N}, M} ... V^{\rightarrow \Omega}_{\C{N},m_{k_N}+2}}_{=V^{\rightarrow F, m_{k_N}+1 \rightarrow M}_{\C{N}\backslash k_N, k_N}} A^{i_{k_N}}_{k_N, x_{k_N}} V^{\rightarrow k_N}_{\C{N} \backslash k_N, m_{k_{N}}+1} ... \\
        &= A^{i_F}_{F, x_F}\circ V^{\rightarrow F, m_{k_N} \rightarrow M}_{\C{N}\backslash k_N, k_N} \circ ...\circ A^{i_{k_n}}_{k_n, x_{k_n}}  \circ V^{\rightarrow k_{n+1}, m_{k_n} \rightarrow m_{k_{n+1}}}_{\C{K}_{n-1}, k_n} \circ ... \circ A^{i_{k_1}}_{k_1, x_{k_1}}  \circ V^{\rightarrow k_1, -1 \rightarrow m_{k_1}}_{\emptyset, \emptyset} \circ A^{i_P}_{P, x_P}.
    \end{aligned}
    \end{gather}
    Finally, summing over all time stamps compatible with a given order of the agents $(k_1,...,k_N)$ and then summing over all possible orders yields
    \begin{gather}\label{eq:pbisqc}
    \begin{aligned}
        \lc{\bar{V} &\circ \bigotimes_{k=0}^{N+1} \sum_{m_k: I(m_k) = k} A^{i_k}_{k, x_k}\otimes  \ket{t=2m_k+1}\bra{t=2m_k}}\\ 
        =& \sum_{(k_1,...,k_N)} \sum_{m_{k_1} < ... < m_{k_N}} A^{i_F}_{F, x_F} \circ V^{\rightarrow F, m_{k_N} \rightarrow M}_{\C{N}\backslash k_N, k_N} \circ ...\circ A^{i_{k_n}}_{k_n, x_{k_n}}  \circ V^{\rightarrow k_{n+1}, m_{k_n} \rightarrow m_{k_{n+1}}}_{\C{K}_{n-1}, k_n} \circ ... \\
        &\circ A^{i_{k_1}}_{k_1, x_{k_1}}  \circ V^{\rightarrow k_1, -1 \rightarrow m_{k_1}}_{\emptyset, \emptyset} \circ A^{i_P}_{P, x_P} \\
        =& \sum_{(k_1,...,k_N)} A^{i_F}_{F, x_F} \circ \sum_{m_{k_N}} V^{\rightarrow F, m_{k_N} \rightarrow M}_{\C{N}\backslash k_N, k_N} \circ A^{i_{k_N}}_{k_N, x_{k_N}} \circ \sum_{m'_{k_N}, m_{k_{N-1}}}V^{\rightarrow k_N, m_{k_{N-1}} \rightarrow m'_{k_N}}_{\C{N} \backslash k_N, k_{N-1}} \circ A^{i_{k_{N-1}}}_{k_{N-1}, x_{k_{N-1}}} \circ  ... \circ \\
        &\circ A^{i_{k_{n+1}}}_{k_{n+1}, x_{k_{n+1}}}\circ \underbrace{\sum_{m'_{k_{n+1}}, m_{k_n}} V^{\rightarrow k_{n+1}, m_{k_n} \rightarrow m'_{k_{n+1}}}_{\C{K}_{n-1}, k_n}}_{=V^{\rightarrow k_{n+1}}_{\C{K}_{n-1}, k_n}} \circ A^{i_{k_{n}}}_{k_{n}, x_{k_{n}}} \circ... \circ A^{i_{k_{1}}}_{k_{1}, x_{k_{1}}}\circ \sum_{m_{k_1}} V^{\rightarrow k_1, -1 \rightarrow m_{k_1}}_{\emptyset, \emptyset} \circ A^{i_P}_{P, x_P} \\
        =& \sum_{(k_1,...,k_N)} A^{i_F}_{F, x_F} \circ V^{\rightarrow F}_{\C{N}\backslash k_N, k_N} \circ A^{i_{k_N}}_{k_N, x_{k_N}} \circ ... \circ A^{i_{k_{n+1}}}_{k_{n+1}, x_{k_{n+1}}} \circ V^{\rightarrow k_{n+1}}_{\C{K}_{n-1}, k_n} \circ A^{i_{k_{n}}}_{k_{n}, x_{k_{n}}} \circ ... A^{i_{k_{1}}}_{k_{1}, x_{k_{1}}} \circ V^{\rightarrow k_1}_{\emptyset, \emptyset} \circ A^{i_P}_{P, x_P}.
    \end{aligned}
    \end{gather}
    In the second equality we used that
    \begin{equation}
        V^{\rightarrow k_{n+2}, m_{k_{n+1}} \rightarrow m'_{k_{n+2}}}_{\C{K}_{n}, k_{n+1}} \circ V^{\rightarrow k_{n+1}, m_{k_n} \rightarrow m'_{k_{n+1}}}_{\C{K}_{n-1}, k_n} = \delta_{m_{k_{n+1}}, m'_{k_{n+1}}} V^{\rightarrow k_{n+2}, m_{k_{n+1}} \rightarrow m'_{k_{n+2}}}_{\C{K}_{n}, k_{n+1}} \circ V^{\rightarrow k_{n+1}, m_{k_n} \rightarrow m_{k_{n+1}}}_{\C{K}_{n-1}, k_n}
    \end{equation}
    as the map $V^{\rightarrow k_{n+1}, m_{k_n} \rightarrow m'_{k_{n+1}}}_{\C{K}_{n-1}, k_n}$ outputs on wire $\alpha'_{m'_{k_{n+1}}}$ on which $V^{\rightarrow k_{n+2}, m_{k_{n+1}} \rightarrow m'_{k_{n+2}}}_{\C{K}_{n}, k_{n+1}}$ is 0 if $m_{k_{n+1}} -1 \neq m'_{k_{n+1}}$.

    In the last line of \cref{eq:pbisqc}, we recognise the action of a QC-QC on maps (cf. \cref{eq:qcqcaction}), hence
    \begin{equation}
        \CM{\bar{V}}{\sum_{m_k} (A^{i_k}_{k, x_k} \otimes \ket{t=2m_k+1}\bra{t=2m_k})} = \CM{V}{A^{i_k}_{k, x_k}}
    \end{equation}
    which (once we sum over the indices $i_k$ and trace out the purifying system) proves that the QC-QC protocol we constructed is behaviourally equivalent to the process box protocol $\mathfrak{P}$. 
\end{proof}

\finalsimp*

\begin{proof}
By \cref{theorem:pbtoqcqc}, the process box protocol $\PBP$ can be mapped to a behaviourally equivalent QC-QC protocol associated with a QC-QC $\C{Q}$. By the results of \cite{salzger2024mappingindefinitecausalorder}, every QC-QC $\C{Q}$ can be mapped to a causal box (referred to as a causal box extension there), together with a mapping of each agent operation in the QC-QC picture to one in the causal box picture such that the AO and LO conditions (referred there as spatiotemporal closed labs) are satisfied and the two scenarios are behaviourally equivalent. In the language of the present paper, this shows that every QC-QC protocol defined by taking a QC-QC and allowing for each agent, all possible quantum instruments compatible with their input and output spaces, can be mapped to a behaviourally equivalent process box protocol. 
Furthermore, the process box protocols in this mapping from QC-QCs satisfies further useful conditions detailed in Sec 4 of \cite{salzger2024mappingindefinitecausalorder}. This includes property 1 (see Sec 4.2, also Fig 8 in \cite{salzger2024mappingindefinitecausalorder}), property 2 (see definitions 3 and 4 and the calculations preceding them in \cite{salzger2024mappingindefinitecausalorder}) and property 3 (in Sec 4.5 in \cite{salzger2024mappingindefinitecausalorder} the causal box is defined via an explicit construction of its sequence representation, where each $V_n$ maps the $m$-message space at $t$ to the $m$-message space at $t+1$) of the above corollary. \Cref{eq:oneatatime} follows from this as follows: Due to AO, only terms as
\begin{equation}
    \bigotimes_{k=1}^{N+1}\bra{i_k,t_k}^{A^I_k} \C{C}' \circ \bigotimes_{k=0}^{N+1} \Mp(\ket{i_k,t_k} \bra{j_k, t'_k}) \bigotimes_{k=1}^{N+1}\ket{j_k, t'_k}^{A^I_k}
\end{equation}
can contribute to the loop composition $\CM{\C{C}'}{\Mp}$. As $\Mp$ maps a 1-message state at $t$ to a 1-message state at $t+1$, we then find that the above vanishes unless for each $t$, the number of messages at $t$ matches the number of matches at $t+2$ for both $\bigotimes_{k=0}^N\ket{i_k,t_k}$ as well as $\bigotimes_{k=0}^N\ket{j_k, t'_k}$. As this holds for all $t$, there must actually be an equal number of messages at each input time. As there are $N$ input times and $N$ messages in total, this number must be 1. Thus, the terms that do not vanish are precisely those which appear on the LHS of \cref{eq:oneatatime}. 

This immediately implies the same also for QC-QC protocols where we allow any subset of possible operations (not necessarily all instruments) for the agents. 

Then, keeping the causal box of the above protocol (obtained from the QC-QC to PB mapping of the previous work), but restricting the agents' operations to be precisely those that appear via the image of the two mappings (PB protocol to QC-QC protocol and QC-QC protocol to PB protocol) applied to the original process box protocol $\PBP$, it is easy to see that we now obtain a new process box protocol $\PBPp$ which is behaviourally equivalent to $\PBP$ and which respects all three properties of the corollary. 

\end{proof}

\end{document}